\documentclass[10pt]{amsart}

\usepackage{amssymb,amsthm,amsmath}
\usepackage[numbers,sort&compress]{natbib}
\usepackage{color}
\usepackage{graphicx}
\usepackage{tikz}
\usepackage{amssymb,amsthm,amsmath}
\usepackage[numbers,sort&compress]{natbib}
\usepackage{color}
\usepackage{graphicx}
\usepackage{tikz}
\begingroup
  \catcode `~=11
  \gdef\mytilde{~}
  \catcode `\|=0
  \catcode `\\=11
  |gdef|mybs{\}
|endgroup
\usepackage[colorlinks,
            linkcolor=red,
            anchorcolor=blue,
            citecolor=green
            ]{hyperref}

\hoffset -3.5pc

\title[ ]{ Universal hierarchical structure  of  quasiperiodic eigenfunctions }
\author{ Svetlana Jitomirskaya}
\address[ Svetlana Jitomirskaya]{ Department of Mathematics, University of California, Irvine, California 92697-3875, USA}
\email{szhitomi@math.uci.edu}

\author{Wencai Liu}
\address[Wencai Liu]{Department of Mathematics, University of California, Irvine, California 92697-3875, USA}\email{liuwencai1226@gmail.com}


\theoremstyle{plain}
\newtheorem{theorem}{Theorem}[section]
\newtheorem{corollary}[theorem]{Corollary}
\newtheorem{lemma}[theorem]{Lemma}
\newtheorem{proposition}[theorem]{Proposition}

\theoremstyle{definition}
\newtheorem{definition}[theorem]{Definition}

\newtheorem{remark}[theorem]{Remark}

\begin{document}


\begin{abstract}
We determine exact exponential asymptotics of eigenfunctions and
of corresponding transfer matrices of
the almost Mathieu operators for all frequencies in the
localization regime. This uncovers a universal structure in their behavior, governed by the continued fraction expansion of the frequency, explaining some predictions in physics literature. In addition it  
proves the
arithmetic version of the frequency
transition conjecture. Finally, it  leads to an explicit description of several non-regularity phenomena
in the corresponding non-uniformly hyperbolic cocycles, which is also of
interest as both the first natural example of some of those phenomena
and, more generally, the first non-artificial model where
non-regularity can be
explicitly studied.

\end{abstract}
\maketitle

\section{Introduction}

A very captivating question and a longstanding theoretical challenge
in solid state physics is to determine/understand the hierarchical
structure of spectral features of operators describing 2D Bloch
electrons in perpendicular magnetic fields, as related to the continued
fraction expansion of the magnetic flux. Such structure was first
predicted in the work of Azbel in 1964 \cite{azbel}.
It was
numerically confirmed through the famous butterfly plot and further
argued for  by Hofstadter  in \cite{hof}, for the spectrum of the almost Mathieu
operator. This was even before this model was linked to the integer quantum
Hall effect \cite{thou} and other important phenomena.  Mathematically, it is known that the spectrum is a
Cantor set for all irrational fluxes \cite{avila2009ten}, and
moreover, even all gaps predicted by the gap labeling are open in the
non-critical case \cite{AJalmost, AYZ2}. Both were very important
challenges in themselves, even though these results, while strongly
indicate, do not describe or explain the hierarchical structure, and
the problem of its description/explanation remains open, even in physics. As
for the understanding the hierarchical behavior of the eigenfunctions,
despite significant numerical studies and even a discovery of  Bethe Ansatz
solutions \cite{wieg,aw} it has also remained an important open challenge even at
the physics level. Certain results indicating the
hierarchical structure in the
corresponding semi-classical/perturbative regimes were previously obtained in
the works of Sinai, Helffer-Sjostrand, and Buslaev-Fedotov (see
\cite{helf,busl,sin}, and also \cite{zhi} for a different model
).

 In this paper we address the latter problem by describing the
 universal self-similar exponential structure of eigenfunctions  throughout the
 entire localization regime. We determine explicit
 universal  functions $f(k)$ and $g(k),$ depending only on the
 Lyapunov exponent and the position of $k$ in the hierarchy defined by
 the denominators $q_n$ of the continued fraction approximants of the
 flux $\alpha,$ that completely define the exponential behavior of,
 correspondingly, eigenfunctions and norms of the transfer matrices of
 the almost Mathieu operators, for all eigenvalues corresponding to
 a.e. phase , see Theorem \ref{Maintheoremdecay}. \footnote{ This paper supplants  our earlier preprint entitled ``Asymptotics of quasiperiodic eigenfunctions". The latter preprint is not intended for publication.} Our result holds for
 {\it all} frequency and coupling pairs in the localization
 regime. Since the behavior is fully determined by the frequency and
 does not depend on the phase, it is the same, eventually, around any
 starting point, so is also seen unfolding at different scales when
 magnified around local eigenfunction maxima, thus describing the
 exponential universality in the hierarchical structure, see, for
 example, Theorems \ref{universal},\ref{addnewth}.


While the almost Mathieu family is precisely the one of main interest in physics literature, it also presents the
simplest case of analytic quasiperiodic operator, so a natural
question is which features discovered for the almost Mathieu would
hold for this more general class. Not all do, in particular, the ones
that exploit the self-dual nature of the family $
H_{\lambda,\alpha,\theta}$ often cannot be expected to hold in
general. In case of Theorems \ref{Maintheoremdecay} and \ref{universal}, we conjecture that they should in
fact hold for general analytic (or even more general) potentials, for a.e. phase and with
$\ln |\lambda |$ replaced by the Lyapunov exponent $L(E)$ (see
Footnote \ref{foot}), but with otherwise the same or very similar
statements. The hierarchical structure theorems
\ref{universal} and \ref{addnewth} are also expected to hold universally for
most (albeit not all, as in the present paper) appropriate local
maxima. Some of our qualitative corollaries may hold in even
higher generality. Establishing this fully would require certain new ideas since so
far even an arithmetic version of localization for the Diophantine
case has not been established for the general analytic family, the
current state-of-the-art result by Bourgain-Goldstein  \cite{bg}
being measure theoretic in $\alpha$. However, some ideas of our method
can already be transferred to general trigonometric polynomials
\cite{jl3}. Moreover, our method was used recently in \cite{hjy} to
show that the
same $f$ and $g$ govern the asymptotics of eigenfunctions and
universality around the local maxima throughout the a.e. localization
regime in another popular object, the Maryland model.

Since we are interested in exponential growth/decay, the behavior of $f$ and $g$  becomes most interesting in case of frequencies with
exponential rate of approximation by the rationals. In general,  localization for quasiperiodic operators is a classical case of a
small denominator problem, and has been traditionally approached
in a perturbative way: through KAM-type schemes for large couplings \cite{sin,fsw,eli} (which, being  KAM-type
schemes, all required Diophantine conditions on frequencies) or through perturbation of periodic operators (Liouville frequency). Unlike the
random case, where, in dimension one, localization holds for all couplings, a
distinctive feature of quasiperiodic operators is the presence of
metal-insulator transitions as couplings increase. Even when
non-perturbative methods, for the almost Mathieu and then for general
analytic potentials, were developed in the 90s \cite{jitomirskaya1999metal,bg}, allowing
to obtain localization for a.e. frequency throughout the
regime of positive Lyapunov exponents, they still required Diophantine
conditions,  and exponentially approximated frequencies that are
neither far from nor close enough to rationals remained a challenge, as
for them there was nothing left to perturb about or to remove. Moreover, it has become clear
that for all frequencies, the true localization threshold should be
arithmetically determined and happen precisely where the exponential
growth provided by the Lyapunov exponent beats the exponential
strength of the small denominators. Thus the most interesting regime -
the neighborhood of the transition - required dealing with the
exponential frequencies not amenable to perturbations/parameter removals, adding a strong number
theoretic flavor to the problem. The precise second transition conjecture was stated
for the almost Mathieu operator \cite{Conjecture}. Our analysis
provides also a
(constructive) solution to the full arithmetic version of the transition
in frequency and explains the role of frequency resonances in
the phenomenon of localization, in a sharp way.

    The almost Mathieu operator (AMO) is the (discrete) quasiperiodic   Schr\"{o}dinger operator  on  $   \ell^2(\mathbb{Z})$:
 \begin{equation*}
 (H_{\lambda,\alpha,\theta}u)(n)=u({n+1})+u({n-1})+ 2\lambda \cos 2\pi (\theta+n\alpha)u(n), 
 \end{equation*}
where $\lambda$ is the coupling, $\alpha $ is the frequency, and $\theta $ is the phase.
\par
It is the central quasiperiodic model due to 
coming
from physics and attracting continued interest there. First appearing
in Peierls \cite{peierls1933theorie}, it arises as related, in two different ways, to a
two-dimensional electron subject to a perpendicular magnetic field and
plays a central role in the Thouless et al theory of the integer
quantum Hall effect. For further background, history, and surveys of results
see \cite{jitomirskaya2015dynamics,dam, sim60,last} and
references therein.

Almost Mathieu operator has a transition from zero to positive
Lyapunov exponents on the spectrum at $|\lambda| =1$ (the critical
coupling)  leading to the
conjecture, dating back to \cite{aa}, that it induces a transition
from absolutely continuous to pure point spectrum.  For Diophantine $\alpha$
this was proved in \cite{jitomirskaya1999metal}. The result was
extended to all $\alpha,\theta$ for $|\lambda| <1$ (the subcritical regime) in
\cite{avila2008absolutely},  solving one of the
 Simon's problems \cite{simonsproblems}. For the supercritical regime
 ($|\lambda| > 1$) it is known however that the nature of the
 spectrum should depend on the arithmetic properties of $\alpha$
 \cite{gordon,as}.

Set
 \begin{equation}\label{Def.Beta}
  \beta= \beta(\alpha)=\limsup_{n\rightarrow\infty}\frac{\ln q_{n+1}}{q_n},
 \end{equation}
where  $ \frac{p_n}{q_n} $  are  the continued fraction  approximants   of $\alpha$.
 \par
For any irrational number $\alpha$, we say that  phase $\theta\in (0,1)$   is
 Diophantine  with respect to $\alpha$, if there exist $ \kappa>0$ and $\nu>0$ such that
 \begin{equation}\label{DCtheta}
   ||2\theta+k\alpha||_{\mathbb{R}/\mathbb{Z}} > \frac{\kappa}{|k|^{\nu}},
 \end{equation}
 for any $k\in \mathbb{Z} \backslash \{0\}
 $, where $||x||_{\mathbb{R}/\mathbb{Z}}=\text{dist}(x,\mathbb{Z})$.
 Clearly, for any irrational number $\alpha$, the set of phases which
 are  Diophantine with respect to $\alpha$ is of full Lebesgue measure.
The conjecture in \cite{Conjecture} states that
   for $\alpha$-Diophantine (thus almost every) $\theta$, $H_{\lambda,\alpha,\theta}$ satisfies
   Anderson localization (i.e., has only pure point spectrum with
   exponentially decaying eigenfunctions)  if  $ |\lambda| >  e^{
     \beta}$, and has, for all $\theta,$ purely singular continuous
   spectrum for  $ 1<|\lambda| <  e^{
     \beta}.$\footnote{\label{1} The original conjecture is slightly stronger in
     that it allows for not just polynomial, but any subexponential
     approximation of $2\theta$ by $k\alpha$. The same goes for our
     proof, with obvious modifications. We choose to present the
     result, and thus also present the conjecture, for a slightly
    stronger Diophantine case in order to slightly simplify the argument.
}

For $\beta =0$ this follows from \cite{jitomirskaya1999metal}.  A
progress towards the localization side of the above conjecture was
made in \cite{avila2009ten} (localization for $|\lambda| > e^{\frac{16}{9}
     \beta},$  as a step in solving the Ten Martini problem) and in \cite{you2013embedding} (in a limited sense, for $|\lambda| >Ce^{
     \beta}).$ The method developed in \cite {avila2009ten} that
   allowed to approach exponentially small denominators on the
localization side was brought to its
technical limits in \cite{MR3340177}, where the result for $|\lambda|
> e^{   \frac{3}{2}\beta}$ was obtained.

Lately, with the development of Avila's global theory and the proof
of the almost reducibility conjecture \cite{art1}, it has become possible to
obtain non-perturbative reducibility directly, allowing to potentially
argue localization for the dual model by duality, as was first done,
in a perturbative regime in \cite{bel}, avoiding the
localization method completely. This was done recently by Avila-You-Zhou
\cite{AYZ}  who proved  the full singular continuous part of the conjecture
and a measure-theoretic (i.e.  almost all $\theta$) version of the pure point part
(see also \cite{jitomirskaya20152} where a simple alternative way to
argue completeness in the duality argument was presented).   The
measure-theoretic (in phase) nature of the pure point result of
\cite{AYZ} is, in fact,
inherent in the duality argument. In contrast, our analysis provides a direct constructive proof for an
arithmetically defined set of $\alpha$-Diophantine $\theta$,  thus
proving the full
arithmetic version of the conjecture.


Our method can be used to also obtain  {\it precise} asymptotics  of {\it arbitrary}
solutions of $ H_{\lambda,\alpha,\theta}\varphi=E\varphi$ where $E$ is an
eigenvalue. Combined with the arguments of Last-Simon
\cite{last1999eigenfunctions}, this allows us to find precise asymptotics of the norms of the
transfer-matrices, providing the first example of this sort for non-uniformly hyperbolic dynamics. Since those norms sometimes differ
significantly from the reciprocals of the eigenfunctions, this leads to
further interesting and unusual consequencies, for example exponential
tangencies between contracted and expanded directions at the
resonant sites.

From this point of view, our analysis also provides, as far as we know, the
first study of the dynamics of  Lyapunov-Perron non-regular
points, in a natural setting. An artificial example of irregular
dynamics can be found in \cite{pesinbook}, p.23, however it is not
even a cocycle over an ergodic transformation, and we are not aware of
other such, even artificial, ergodic examples where the dynamics has
been studied. Loosely, for a
cocycle $A$ over a transformation $ f$ acting on a space $ X$
(Lyapunov-Perron) non-regular points $x \in X$ are the ones at which Oseledets
multiplicative ergodic theorem does not hold coherently in both directions. They therefore form a measure
zero set with respect to any invariant measure on $X$.\footnote{ Although in the
uniformly hyperbolic situations this set can be of full Hausdorff
dimension \cite{bars}.}
Yet, it is precisely the non-regular points that are of interest in
the study of Schr\"odinger cocycles in the non-uniformly hyperbolic
(positive Lyapunov exponent) regime, since spectral measures, for
every fixed phase, are
always supported on energies where there exists a
solution polynomially bounded in both directions, so the (hyperbolic) cocycle
defined at such energies is always
non-regular at precisely the relevant phases. Thus the non-regular points capture the entire action from
the point of view of spectral theory, so become the most important
ones to study. One can also discuss stronger
non-regularity notions: absence of forward regularity and, even
stronger, non-exactness of the Lyapunov exponent
\cite{pesinbook}. While it is not difficult to see that energies in the support of singular continuous spectral
 measure in the non-uniformly hyperbolic regime always provide examples of
 non-exactness, our analysis gives the first non-trivial example of
 non-exactness with non-zero upper limit (Corollary \ref{corleigen}). Finally, as we understand,
 this work provides also the first natural example of  an even stronger
 manifestation of the lack of regularity, the
 exponential tangencies (Corollary \ref{angle}). Tangencies
 between contracted and expanded directions are a characteristic feature of
 nonuniform hyperbolicity (and, in particular, always happen at the
 maxima of the eigenfunctions). They complicate proofs of positivity
 of the Lyapunov exponents and are viewed as a difficulty to avoid
 through e.g. the parameter exclusion \cite{bc,lsy,bj}. However, when the
 tangencies are only subexponentially deep they do not in themselves
 lead to non-exactness. Here we observe the first natural example of
 {\it exponentially} strong tangencies (with the rate determined by the
 arithmetics of $\alpha$ and the positions precisely along the
 sequence of resonances.)

The localization-for-the-exponential-regime method of \cite {avila2009ten} consists of different
arguments for non-resonant (meaning sufficiently far from $jq_n$ on
the corresponding scale) sites and for the resonant ones (the rest). It is the
resonant sites that lead to dealing with the smallest denominators and
that necessitate the  $|\lambda |> e^{\frac{16}{9}
     \beta}$ requirement in \cite {avila2009ten}. Here we start with the
   same basic setup, and only technically modify the non-resonant
   statement of \cite {avila2009ten}. However we develop a completely
   new bootstrap technique to handle the resonant sites, allowing us
   to get to the transition and obtain the fine estimates. The
   estimates from below (that coincide with our estimates from
   above) are also new.
In general, the
   statements that are technically similar to the ones in the existing
   literature are collected in the Appendices, while all the results/proofs
   in the body of the paper are, in their pivotal parts, not like anything
   that has appeared before.

The key elements of the technique developed in this paper are robust
and have made it possible to approach other scenarios. As such, in the
upcoming work  we prove the
sharp phase transition for Diophantine $\alpha$ and all $\theta$ and
establish sharp exponential asymptotics of eigenfunctions and
transfer matrices in the corresponding  pure point regime
\cite{jl2}. Moreover, our analysis reveals there a universal
{\it reflective-hierarchical} structure in the entire regime of
phase-induced resonances, a phenomenon not even previously discovered
in physics literature. Thus while in this paper we develop a complete
understanding of frequency induced resonances, in \cite{jl2} we
develop new methods motivated by the ideas of this manuscript to obtain a
complete understanding of phase induced resonances.
In other follow-up works we determine the exact exponent of the exponential decay rate in
expectation for the Diophantine case \cite{jkl} and study delicate properties of the singular
continuous regime, obtaining upper bounds on fractal dimensions of
the spectral measure and quantum dynamics for the almost Mathieu operator \cite{jlt}, as well as
potentials defined by general trigonometric polynomials \cite{jl3}.

Except for a few standard, general (e.g. uniform upper semicontinuity\color{black}) or very simple to verify statements, this
paper is entirely self contained. The only technically involved
fact that we use without proving it in the paper  is Lemma \ref{lya} \cite{bourgain2002continuity} but
this is not even necessary if we replace $\ln |\lambda|$ by the Lyapunov
exponent $L(E)$ throughout the manuscript. \footnote{\label{foot} In fact,  $\ln |\lambda|$ is being
used in this paper as a shortcut for $L(E).$}


 \par
 \section{ Main results   }
 \par


Let
\begin{equation}\label{G.transfer}
A_{k}(\theta)=\prod_{j=k-1}^{0 }A(\theta+j\alpha)=A(\theta+(k-1)\alpha)A(\theta+(k-2)\alpha)\cdots A(\theta)
\end{equation}
and
\begin{equation}\label{G.transfer1}
A_{-k}(\theta)=A_{k}^{-1}(\theta-k\alpha)
\end{equation}
for $k\geq 1$,
where $A(\theta)=\left(
             \begin{array}{cc}
               E-2\lambda\cos2\pi\theta & -1 \\
               1& 0\\
             \end{array}
           \right)
$.
$A_{k}$  is called the (k-step) transfer matrix. As is clear from the
definition, it also depends on $\theta$ and $E$ but since those
parameters will be usually fixed, we omit this from the notation.

Given $ \alpha\in \mathbb{R}\backslash\mathbb{Q}$ we define  functions
$f,g: \mathbb{Z}^{+}\rightarrow \mathbb{R}^+$ in the following way.
Let $ \frac{p_n}{q_n}$ be the continued fraction approximants to $\alpha$.
For any $\frac{q_n}{2}\leq k< \frac{q_{n+1}}{2}$, define $f(k),g(k)$ as follows:

 \begin{description}
    \item[Case 1]   $q_{n+1}^{\frac{8}{9}}\geq \frac{q_n}{2}$ or  $k\geq q_n$.

    If  $\ell q_n\leq k< (\ell+1)q_n$ with $\ell\geq 1$, set

    \begin{equation}\label{G.decayingratenonresonant2add}
  f(k) = e^{-|k -\ell q_n|\ln|\lambda|}\bar{r} _{\ell}^n+e^{-|k-(\ell +1)q_n|\ln|\lambda|}\bar{r}_{\ell+1}^n,
\end{equation}
and
\begin{equation}\label{G.decayingratenonresonant2addtransfer}
  g(k) = e^{-|k -\ell q_n| \ln|\lambda|}\frac{q_{n+1}}{\bar{r} _{\ell}^n}+e^{-|k-(\ell +1)q_n|\ln|\lambda|}\frac{q_{n+1}}{\bar{r} _{\ell+1}^n},
\end{equation}
where   for    $\ell\geq 1$,
\begin{equation*}
  \bar{r}_{\ell}^n=e^{-(\ln|\lambda|-\frac{\ln q_{n+1}}{q_n}+\frac{\ln \ell}{q_n})\ell q_n}.
\end{equation*}
Set also $\bar{r} _{0}^n=1$ for convenience.\\
If  $ \frac{q_n}{2}\leq k< q_n$, set
    \begin{equation}\label{G.decayingratenonresonant2add0}
  f(k) = e^{-k\ln|\lambda|}+e^{-|k-q_n|\ln|\lambda|}\bar{r}_{1}^n,
\end{equation}
and
\begin{equation}\label{G.decayingratenonresonant2addtransfer0}
  g(k) = e^{k\ln|\lambda|}.
\end{equation}

\item[Case 2]    $q_{n+1}^{\frac{8}{9}}< \frac{q_n}{2}$ and
$\frac{q_n}{2}\leq k\leq \min \{q_n,\frac{q_{n+1}}{2}\}$.
\par

Set
\begin{equation}\label{G.decayingratenonresonant3addred}
f(k)=e^{-k\ln|\lambda| },
\end{equation}

and
\begin{equation}\label{G.decayingratenonresonant3addgred}
g(k)=e^{ k\ln|\lambda|}.
\end{equation}
\end{description}

Notice that $f,g$ only depend on $\alpha$ and $\lambda$ but not on
$\theta$ or $E.$ $f(k)$ decays and $g(k)$ grows
exponentially, globally, at varying rates that depend on the position of $k$ in
the hierarchy defined by the continued fraction expansion of $\alpha,$
see Fig.1 and Fig.2.

We say that $\phi$ is a generalized  eigenfunction of $H$ with generalized
eigenvalue $E$, if
   \begin{equation} \label{shn}
     H\phi=E\phi  ,\text{ and }  |\phi(k)|\leq \hat{C}(1+|k|).
  \end{equation}
Our first main result is  that in the entire regime $|\lambda| >e^{\beta},$ the exponential asymptotics of the
generalized eigenfunctions and norms of transfer matrices at the
generalized eigenvalues are completely determined
by $f(k),g(k)$.
\begin{theorem}\label{Maintheoremdecay}
Let $ \alpha\in \mathbb{R}\backslash\mathbb{Q}$ be such that $
|\lambda| >e^{\beta(\alpha)}$. Suppose $\theta$ is  Diophantine with respect to $\alpha$,
$E$ is a generalized eigenvalue of $H_{\lambda,\alpha,\theta}$ and  $\phi$ is
the  generalized eigenfunction. 
Let $ U(k) =\left(\begin{array}{c}
        \phi(k)\\
       \phi({k-1})
     \end{array}\right)
 $.
Then for any $\varepsilon>0$, there exists $K$ (depending on $\lambda,\alpha, \hat{C},\varepsilon$ and Diophantine constants $\kappa,\nu$) such that for any $|k|\geq K$,  $U(k)$ and $A_k$ 
satisfy
\begin{equation}\label{G.Asymptotics}
 f(|k|)e^{-\varepsilon|k|} \leq ||U(k)||\leq f(|k|)e^{\varepsilon|k|},
\end{equation}
and
\begin{equation}\label{G.Asymptoticstransfer}
 g(|k|)e^{-\varepsilon|k|} \leq ||A_k||\leq g(|k|)e^{\varepsilon|k|}.
\end{equation}

\begin{center}
\begin{tikzpicture}
\draw [->](0,0)--(10.5,0);
\draw [->](0,0)--(0,6);
\draw plot [smooth] coordinates {(2,5)(2.85,3.5)(3.5,4) (3.95,3.5)(4.35,2.7)(5,3)(5.88,1.6)(6.5,2)(7.45,0.7)(8,1)};
\node [above] at (2,5){$\bar{r} _{\ell}^{n}$};
\node [above] at (5,3){$\bar{r} _{\ell+2}^{n}$};
\node [above] at (8,1){$\bar{r} _{\ell+4}^{n}$};
\draw [dashed] (2,5)--(2,0);
\draw [dashed] (2.75,3.55)--(2.75,0);
\draw [dashed] (3.5,4)--(3.5,0);
\draw [dashed] (5,3)--(5,0);
\draw [dashed] (6.5,2)--(6.5,0);
\draw [dashed] (8,1)--(8,0);
\node [below] at (2,0){$\ell q_n$};
\node [below] at (3.5,0){$(\ell+1) q_n$};
\node [below] at (5,0){$(\ell +2)q_n$};
\node [below] at (6.5,0){$(\ell+3) q_n$};
\node [below] at (8,0){$(\ell +4)q_n$};
\node [below] at (10.2,0){$k$};
\node [below] at (9.4,0){$\frac{q_{n+1}}{2}$};
\node [below] at (1,0){$\frac{q_n}{2}$};
\node [left] at (0,4.8){$f(k)$};
\node [below] at (5,-1){Fig.1};
\end{tikzpicture}
\end{center}
\begin{center}
\begin{tikzpicture}
\draw [->](0,0)--(10.5,0);
\draw [->](0,0)--(0,6);
\draw plot [smooth] coordinates {(2,1)(2.63,0.7)(3.5,2)(4.13,1.6)(5,3)(5.73,2.7)(6.5,4)(7.13,3.5)(8,5)};
\node [above] at (2,1){$\frac{q_{n+1}}{\bar{r} _{\ell}^{n}}$};
\node [above] at (5,3){$\frac{q_{n+1}}{\bar{r} _{\ell+2}^{n}}$};
\node [above] at (8,5){$\frac{q_{n+1}}{\bar{r} _{\ell+4}^{n}}$};
\draw [dashed] (2,1)--(2,0);
\draw [dashed] (2.75,0.75)--(2.75,0);
\draw [dashed] (3.5,2)--(3.5,0);
\draw [dashed] (5,3)--(5,0);
\draw [dashed] (6.5,4)--(6.5,0);
\draw [dashed] (8,5)--(8,0);
\node [below] at (2,0){$\ell q_n$};
\node [below] at (3.5,0){$(\ell+1) q_n$};
\node [below] at (5,0){$(\ell +2)q_n$};
\node [below] at (6.5,0){$(\ell+3) q_n$};
\node [below] at (8,0){$(\ell +4)q_n$};
\node [below] at (10.2,0){$k$};
\node [below] at (9.4,0){$\frac{q_{n+1}}{2}$};
\node [below] at (1,0){$\frac{q_n}{2}$};
\node [left] at (0,4.8){$g(k)$};
\node [below] at (5,-1){Fig.2};
\end{tikzpicture}
\end{center}
\end{theorem}


Certainly, there is nothing special about $k=0,$ so the behavior
described in Theorem  \ref{Maintheoremdecay}  happens around arbitrary
point $k=k_0$. This implies the self-similar nature of the
eigenfunctions): $U(k)$ 
behave as
described at scale $q_n$ but when looked at in windows of size $q_k, q_k<
 q_{n-1}$ will demonstrate the same universal behavior around appropriate
local maxima/minima.


To make the above precise, let $\phi$ be an  eigenfunction, and   $ U(k) =\left(\begin{array}{c}
        \phi(k)\\
       \phi({k-1})
     \end{array}\right)
 $. Let $I^j_{\varsigma_1,\varsigma_2}=[-\varsigma_1 q_j,  \varsigma_2
 q_j]$,  for some $
 0<\varsigma_1,\varsigma_2\leq 1$.  We will say $k_0$ is a local $j$-maximum of $\phi$
if
 $||U(k_0)||\geq ||U(k)||$ for  $ k-k_0\in
 I^j_{\varsigma_1,\varsigma_2} $. Occasionally, we will also use
 terminology $(j,\varsigma)$-maximum for a local $j$-maximum on an
 interval $I^j_{\varsigma,\varsigma}.$

 We will say a local $j$-maximum $k_0$
 is {\it nonresonant} if
 \begin{equation*}
   ||2\theta+(2k_0+k)\alpha||_{\mathbb{R}/\mathbb{Z}} > \frac{\kappa}{{q_{j-1}}^{\nu}},
 \end{equation*}
  for all $|k|\leq 2 q_{j-1}$
 and
 \begin{equation}\label{DCdoublesize}
   ||2\theta+(2k_0+k)\alpha||_{\mathbb{R}/\mathbb{Z}} > \frac{\kappa}{{|k|}^{\nu}},
 \end{equation}
 for all $2q_{j-1}<|k|\leq 2 q_{j}$.

We will say a local $j$-maximum
 is {\it strongly nonresonant} if
 \begin{equation}\label{DCdoublesizestrong}
   ||2\theta+(2k_0+k)\alpha||_{\mathbb{R}/\mathbb{Z}} > \frac{\kappa}{{|k|}^{\nu}},
 \end{equation}
 for all $0<|k|\leq 2 q_{j}$.
 \par


An immediate corollary of Theorem \ref{Maintheoremdecay} is the
universality of behavior at all (strongly) nonresonant local maxima.
\begin{theorem} \label{universal}
Given $\varepsilon>0$, there exists $j(\varepsilon)<\infty$ such that if  $k_0$ is a local $j$-maximum for $j>j(\epsilon)$, then the following two statements hold:

If $k_0$ is  nonresonant, then
\begin{equation}\label{max}
 f(|s|)e^{-\varepsilon|s|} \leq \frac{||U(k_0+s)||}{||U(k_0)||}\leq f(|s|)e^{\varepsilon|s|},
\end{equation}
for all $2s\in  I^j_{\varsigma_1,\varsigma_2},\;
|s|>\frac{q_{j-1}}{2}.$

If $k_0$ is  strongly nonresonant, then
\begin{equation}\label{max1}
 f(|s|)e^{-\varepsilon|s|} \leq \frac{||U(k_0+s)||}{||U(k_0)||}\leq f(|s|)e^{\varepsilon|s|},
\end{equation}
for all $2s\in  I^j_{\varsigma_1,\varsigma_2}.$
\end{theorem}
\begin{remark}\begin{enumerate}
\item For the neighborhood of a local $j$-maximum described in
  the Theorem \ref{universal}  only the behavior of $f(s)$ for
  $q_{j-1}/2<|s|\leq q_{j}/2$ is  relevant. Thus $f$ implicitly depends
  on $j$ but through the scale-independent mechanism described in \eqref{G.decayingratenonresonant2add},\eqref{G.decayingratenonresonant2add0} and \eqref{G.decayingratenonresonant3addred}.
\item Actually,
 a  modification in our  proof allows to formulate (\ref{max}) in Theorem
 \ref{universal} with non-resonant condition  (\ref{DCdoublesize})
only required  for  $2q_{j-1}<|k|\leq q_{j}$ rather than  for  $2q_{j-1}<|k|\leq 2q_{j}$.
\end{enumerate}
\end{remark}

In case $\beta(\alpha)>0,$ Theorem \ref{Maintheoremdecay} also
guarantees an abundance (and a hierarchical structure) of local
maxima of each eigenfunction. Let $k_0$ be a global maximum\footnote{If there are several, what
  follows is true for each.} \label{globalmax}.

\begin{center}
\begin{tikzpicture}[thick, scale=1.8]
\node [below] at (3.0,5.6){Universal hierarchical structure  of  an eigenfunction};
\draw [->](-1.3,0)--(8.5,0);
\draw[dashed](3.4,4.05)--(3.4,0);
\draw[dashed](5.1,3)--(5.1,0);
\node [below] at (5.1,0){ $b_1$};
\draw[dashed](6.6,2)--(6.6,0.0);
\node [below] at (6.6,0){ $b_2$ };
\draw[dashed](1.7,3)--(1.7,0);
\node [below] at (1.7,0){ $b_{-1}$ };
\draw[dashed](0.2,2)--(0.2,0);
\node [below] at (0.2,0){$b_{-2}$ };
\draw[dashed](3.4,4.05)--(3.4,4.7);
\node [below] at (3.4,0){$k_0$ };
\node [below] at (5.1,-0.25){Local maximum of depth 1};
\node [below] at (1.7,-0.25){Local maximum of depth 1};
\node [below] at (3.4,5.0){Global maximum };
\node [below] at (3,-0.75){Fig.3};
\draw plot [smooth] coordinates {(-1.2,1)(-0.65,0.7)(-0.23,1.75)(-0.13,1.7)(-0.03,1.85)(0.07,1.8) (0.2,2) (0.33,1.8)(0.43,1.85)(0.53,1.7)
(0.92,1.6)(1.27,2.75)(1.37,2.7)(1.47,2.85)(1.57,2.8) (1.7,3)(1.83,2.8) (1.93,2.85)(2.03,2.7)(2.13,2.75)(2.45,2.6)
 (3.3,4)
(3.4,4.05)(3.5,4)(4.35,2.6)(4.67,2.75)(4.77,2.7)(4.87,2.85)(4.97,2.8)(5.1,3) (5.23,2.8)(5.33,2.85)(5.43,2.7)(5.53,2.75)(5.88,1.6)(6.17,1.75)(6.27,1.7)(6.37,1.85)(6.47,1.8)(6.6,2)(6.73,1.8)(6.83,1.85)
(6.93,1.7)(7.03,1.75)(7.45,0.7)(8,1)};
\draw [->](7.5,3)--(7,1.8);
\node [below] at (7.5,3.25){$b_{2,2}$};
\draw [->](7,3)--(6.82,1.83);
\node [below] at (7 ,3.25){$b_{2,1}$};

\draw [->](4.5,4)--(4.85,2.85);
\node [below] at (4.5,4.25){$b_{1,-1}$};

\draw [->](5.8,4.2)--(5.33,2.85);
\node [below] at (5.8,4.45){$b_{1,1}$};

\draw [->](6.1,4)--(5.48,2.78);
\node [below] at (6.1,4.25){$b_{1,2}$};
\draw[-](4.4,2)--(4.4,3.5);
\draw[-](4.4,3.5)--(5.7,3.5);
\draw[-](5.7,3.5)--(5.7,2);
\draw[-](5.7,2)--(4.4,2);
\node [below] at (5,2.5){Window I};

\end{tikzpicture}
\end{center}

\begin{center}
\begin{tikzpicture}[thick, scale=1.8]
\node [below] at (3.3,5.5){Window I};
\draw [->](-1.3,0)--(8.5,0);
\draw[dashed](3.4,4.05)--(3.4,0);
\draw[dashed](5.1,3)--(5.1,0);
\node [below] at (5.1,0){ $b_{1,1}$};
\draw[dashed](6.6,2)--(6.6,0.0);
\node [below] at (6.6,0){ $b_{1,2}$ };
\draw[dashed](1.7,3)--(1.7,0);
\node [below] at (1.7,0){ $b_{1,-1}$ };
\draw[dashed](0.2,2)--(0.2,0);
\node [below] at (0.2,0){$b_{1,-2}$ };
\node [below] at (3.4,0){$b_{1}$ };
\node [below] at (5.1,-0.25){Local maximum of depth 2};
\node [below] at (1.7,-0.25){Local maximum of depth 2};
\node [below] at (3.4,5.0){Local maximum  of depth 1};
\node [below] at (3,-0.75){Fig.4};
\draw plot [smooth] coordinates {(-1.2,1)(-0.65,0.7)(-0.23,1.75)(-0.13,1.7)(-0.03,1.85)(0.07,1.8) (0.2,2) (0.33,1.8)(0.43,1.85)(0.53,1.7)
(0.92,1.6)(1.27,2.75)(1.37,2.7)(1.47,2.85)(1.57,2.8) (1.7,3)(1.83,2.8) (1.93,2.85)(2.03,2.7)(2.13,2.75)(2.45,2.6)
 (3.3,4)
(3.4,4.05)(3.5,4)(4.35,2.6)(4.67,2.75)(4.77,2.7)(4.87,2.85)(4.97,2.8)(5.1,3) (5.23,2.8)(5.33,2.85)(5.43,2.7)(5.53,2.75)(5.88,1.6)(6.17,1.75)(6.27,1.7)(6.37,1.85)(6.47,1.8)(6.6,2)(6.73,1.8)(6.83,1.85)
(6.93,1.7)(7.03,1.75)(7.45,0.7)(8,1)};
\draw [->](7.5,3)--(7,1.8);
\node [below] at (7.5,3.25){$b_{1,2,2}$};
\draw [->](7,3)--(6.82,1.83);
\node [below] at (7 ,3.25){$b_{1,2,1}$};

\draw [->](4.5,4)--(4.85,2.85);
\node [below] at (4.5,4.25){$b_{1,1,-1}$};

\draw [->](5.8,4.2)--(5.33,2.85);
\node [below] at (5.8,4.45){$b_{1,1,1}$};

\draw [->](6.1,4)--(5.48,2.78);
\node [below] at (6.1,4.25){$b_{1,1,2}$};

\end{tikzpicture}
\end{center}
We first describe the hierarchical structure of local maxima
informally. We will say that a scale $n_{j_0}$ is exponential if $\ln
q_{n_{j_0}+1}>c q_{n_{j_0}}.$ Then there is a {\it constant} scale $\hat
n_0$ thus a constant $C:= q_{\hat{n}_0+1} ,$ such that for any
exponential scale $n_j$ and any eigenfunction there are local
$n_j$-maxima within distance
$C$ of $k_0 +s q_{n_{j_0}}$ for each $0<|s|<e^{c q_{n_{j_0}}}.$ Moreover,
these are all the local $n_{j_0}$-maxima  in $[k_0-e^{c q_{n_{j_0}}}, k_0+e^{c
  q_{n_{j_0}}}]$. The exponential
behavior of the eigenfunction
in the local  neighborhood (of size $_\mytilde  q_{n_{j_0}}$) of each such local maximum, normalized
by the value at the local maximum is given by $f$. Note that only
exponential behavior at the corresponding scale is determined by $f$
and fluctuations of much smaller size are invisible.  Now, let $n_{j_1}<n_{j_0}$ be
another exponential scale. Denoting
``depth 1''  local maximum located near $k_0 +a_{n_{j_0}} q_{n_{j_0}}$ by
$b_{a_{n_{j_0}}}$ we then have a similar picture around $b_{a_{n_{j_0}}}:$
there are local $n_{j_1}$-maxima
in the vicinity of
$b_{a_{n_{j_0}}} +s q_{n_{j_1}}$ for each $0<|s|<e^{c
  q_{n_{j_1}}}$. Again, this describes all the local
$q_{n_{j_1}}$-maxima within an exponentially large interval. And
again, the
exponential (for the $n_{j_1}$ scale) behavior in the local  neighborhood (of size $_\mytilde q_{n_{j_1}}$) of each such local maximum, normalized
by the value at the local maximum is given by $f$. Denoting those
``depth 2'' local maxima located near $b_{a_{n_{j_0}}} +a_{n_{j_1}}
  q_{n_{j_1}},$ by $b_{a_{n_{j_0}},a_{n_{j_1}}}$ we then get the same
  picture taking the magnifying glass another level deeper and so on. At the end we obtain a
  complete hierarchical structure of local maxima that we denote by
  $b_{a_{n_{j_0}},a_{n_{j_1}},...,a_{n_{j_s}}}$ with each
  ``depth $s+1$" local maximum $b_{a_{n_{j_0}},a_{n_{j_1}},...,a_{n_{j_s}}}$ being in the corresponding
  vicinity of the ``depth $s$" local maximum
  $b_{a_{n_{j_0}},a_{n_{j_1}},...,a_{n_{j_{s-1}}}}$ and with universal
  behavior at the corresponding scale around each.  The quality of the
  approximation of the position of the next maximum gets lower with each level of
  depth, yet the depth of the
  hierarchy that can be so achieved is at least $j/2-C$, see Corollary
  \ref{corhie}. Fig. 3 schematically illustrates the structure of
  local maxima of depth one and two, and Fig. 4 illustrates that the
 the neighborhood of a local maximum appropriately magnified looks
 like a picture of the global maximum.

  We now describe the hierarchical structure precisely. Suppose
  \begin{equation}\label{DCthetaprime}
   ||2(\theta+k_0\alpha)+k\alpha||_{\mathbb{R}/\mathbb{Z}} > \frac{\kappa}{|k|^{\nu}},
 \end{equation}
 for any $k\in \mathbb{Z} \backslash \{0\}$.  Fix  $0<\varsigma,\epsilon$ with $\varsigma+2\epsilon <1.$ Let $n_j\to\infty$ be such that $\ln q_{n_j+1}\geq
(\varsigma+2\epsilon) \ln |\lambda|q_{n_j} .$ Let $
\mathfrak{c}_{j}=(\ln q_{n_{j}+1}-\ln
|a_{n_{j}}|)/\ln |\lambda|q_{n_{j}}-\epsilon$. We have $\mathfrak{c}_{j} >\epsilon$ for $0<a_{n_j} <
e^{\varsigma \ln |\lambda| q_{n_j}}$. Then we have

\begin{theorem}\label{addnewth}
There exists   $\hat{n}_0 (\alpha,\lambda,\kappa ,\nu ,\epsilon)<\infty$  such that for any   $j_0>j_1>\cdots>j_k$, $n_{j_k}\geq \hat{n}_0+k$, and
  $0<a_{n_{j_i}} <
e^{\varsigma \ln |\lambda| q_{n_{j_i}}}, i=0,1,\ldots,k,$ for all $0\leq
s\leq k$ there exists a
local $n_{j_s}$-maximum $b_{a_{n_{j_0}},a_{n_{j_1}},...,a_{n_{j_s}}}$ on the
interval  $b_{a_{n_{j_0}},a_{n_{j_1}},...,a_{n_{j_s}}}+I^{n_{j_s}}_{\mathfrak{c_{j_s}},1}$ for all $0\leq s\leq k$
such that the
following holds:
\begin{description}
  \item[I]   $|b_{a_{n_{j_0}}}-(k_0 +a_{n_{j_0}}q_{n_{j_0}})|\leq q_{\hat{n}_0+1},$
  \item[II] For any $1\leq s\leq k ,$
$|b_{a_{n_{j_0}},a_{n_{j_1}},...,a_{n_{j_s}}}-(b_{a_{n_{j_0}},a_{n_{j_1}},...,a_{n_{j_{s-1}}}} +a_{n_{j_s}}q_{n_{j_s}})|\leq   q_{\hat{n}_0+s+1}$.
  \item[III] if $2(x-b_{a_{n_{j_0}},a_{n_{j_1}},...,a_{n_{j_k}}})\in
    I^{{n_{j_k}}}_{\mathfrak{c}_{j_k},1}$ and $|x-b_{a_{n_{j_0}},a_{n_{j_1}},...,a_{n_{j_k}}}|\geq q_{\hat{n}_0+k}$\color{black},  then  for each
    $s=0,1,...,k,$

\begin{equation}\label{G.add1local}
 f(x_s)e^{-\varepsilon|x_s|} \leq \frac{||U(x)||}{||U(b_{a_{n_{j_0}},a_{n_{j_1}},...,a_{n_{j_s}}})||}\leq f(x_s)e^{\varepsilon|x_s|},
\end{equation}
where $x_s=|x-b_{a_{n_{j_0}},a_{n_{j_1}},...,a_{n_{j_s}}}|$ is large
enough.
\end{description}
Moreover, {\it every} local $n_{j_s}$-maximum on the interval
$b_{a_{n_j},a_{n_{j_1}},...,a_{n_{j_{s-1}}}}+[ -e^{\epsilon\ln\lambda
  q_{n_{j_s}}}, e^{\epsilon\ln\lambda
  q_{n_{j_s}}}]$ is of the form
$b_{a_{n_{j_0}},a_{n_{j_1}},...,a_{n_{j_s}}}$ for some $a_{n_{j_s}}.$

\end{theorem}
\begin{remark}
By I of Theorem \ref{addnewth}, the local maximum can be determined up
to a constant $K_0=q_{\hat{n}_0+1}$. Actually,  if $k_0$ is only a local
$n_{j}+1$-maximum, we can still make sure
that  I, II and III of Theorem \ref{addnewth} hold.
This is   the local version of
Theorem   \ref{addnewth}, see   Theorem \ref{addnewthlocal}.
\end{remark}
\begin{remark}
$q_{\hat{n}_0+1}$ is the scale at which {\it phase} resonances of
$\theta + k_0\alpha$ still can appear. Notably, it determines the
precision of pinpointing local $n_{j_0}$-maxima in a (exponentially large in
$q_{n_{j_0}})$ neighborhood of $k_0,$  for any $j_0$. When we go down the
hierarchy, the precision decreases, but note that except for the very
last scale it stays at least  iterated logarithmically \footnote{for most scales
even much less} small in the
corresponding scale $q_{n_{j_s}}$
\end{remark}
Thus for
$x\in b_{a_{n_{j_0}},a_{n_{j_1}},...,a_{n_{j_s}}}+[-\frac{c_{j_s}}{2}
q_{n_{j_s}},\frac{1}{2}
q_{n_{j_s}}] $, the
behavior of $\phi(x)$ is described by the same universal $f$ in each
$q_{n_{j_s}}$-window around the corresponding local maximum $b_{a_{n_{j_0}},a_{n_{j_1}},...,a_{n_{j_s}}}$,$
s=0,1,...,k.$ We call such a structure {\it hierarchical}, and we will
say that a local $j$-maximum is $k$-hierarchical if the complete hierarchy goes
down at least $k$ levels (for a precise definition see Section
\ref{hier}). We then have an immediate corollary

\begin{corollary}\label{corhie}
There exists $C=C(\alpha,\lambda,\kappa ,\nu ,\epsilon)$ such that every  local
$n_j$-maximum in $[k_0-e^{\varsigma\ln
  |\lambda|q_{n_j}},k_0+e^{\varsigma\ln |\lambda|q_{n_j}}]$ is at least
$(j/2-C)$-hierarchical.
\end{corollary}
\begin{remark}
The estimate on the depth of the hierarchy in the corollary assumes
the worst case scenario when all scales after $\hat{n}_0$ are
Liouville. Otherwise the hierarchical structure will go even much
deeper. Note that a local $n_j$-maximum that is not an $n_{j+1}$-maximum
cannot be $k$-hierarchical for $k>j.$
\end{remark}


Another interesting corollary of Theorem \ref{Maintheoremdecay} is


 \begin{theorem}\label{Conjecture}
 Let $ \alpha\in \mathbb{R}\backslash\mathbb{Q}$ be such that
 $|\lambda|>e^{\beta(\alpha)}$ and $\theta$ is Diophantine with respect to $\alpha$.
 Then   $ H_{\lambda,\alpha,\theta}$  has   Anderson localization, with
 eigenfunctions decaying at the rate $\ln |\lambda| -\beta$.
\end{theorem}

This solves the arithmetic version of the second transition conjecture
in that it establishes localization throughout the entire regime of
$(\alpha,\lambda)$ where localization may hold for any $\theta$ (see the discussion in
the introduction), for an arithmetically
defined full measure set of $\theta.$


We note that Theorem \ref{Conjecture}  cannot be upgraded to {\it all}
$\theta$
in the regime  $ |\lambda| >  e^{
     \beta}$ \cite{jitomirskaya1994operators} so exclusion of a
   certain arithmetically defined set where the spectrum must be
   singular continuous is necessary. There is a conjecture
   of where in this regime the transition in $\theta$ happens
   \cite{Conjecture} but we do not explore it in this work. The sharp
   transition in $\theta$ for Diophantine $\alpha$ will be established
   in the follow-up work \cite{jl2}.
Also, it could be added that, for all $\theta ,$  $ H_{\lambda,\alpha,\theta}$  has no
localization (i.e., no exponentially decaying eigenfunctions) if $|\lambda|= e^{\beta}$ (see   Appendix \ref{decayingrate}).
\par
\begin{remark}
Theorems \ref{Maintheoremdecay}, \ref{Conjecture} cover the
  optimal range of $(\alpha,\lambda)$ for a.e. $\theta.$ For Theorem \ref{Conjecture}, even though some $\theta$
  have to be excluded \cite{jitomirskaya1994operators}, we do not claim the
  Diophantine condition on $\theta$ is optimal. At the same time, exponentially
  strong $\theta$-resonances (exponentially small lower bound in
  (\ref{DCtheta}) instead of a polynomial) will make Theorem
  \ref{Maintheoremdecay} false as stated, no matter how small the
  exponent, and would require differently defined $f$ and $g.$ In
  \cite{jl2} we obtain $f'$ and $g'$ that govern the exponential
  behavior of eigenfunctions and transfermatrices for {\it all}
  $\theta$ throughout the entire pure point regime corresponding to
  Diophantine $\alpha.$
\end{remark}

Let $\psi(k)$ denote  any solution to $
H_{\lambda,\alpha,\theta}\psi=E\psi$ that is linearly independent with respect to
$\phi(k)$. Let $\tilde{U} (k)=\left(\begin{array}{c}
        \psi(k)\\
       \psi(k-1)
     \end{array}\right)
 $. An immediate counterpart of (\ref{G.Asymptoticstransfer}) is the following
\begin{corollary}\label{C.anysolution}
Under the conditions of Theorem \ref{Maintheoremdecay} for large $k$ vectors  $\tilde{U}(k)$ satisfy
\begin{equation}\label{G.anysolution}
   g(|k|)e^{-\varepsilon|k|} \leq ||\tilde{U}(k)||\leq g(|k|)e^{\varepsilon|k|}.
\end{equation}
\end{corollary}
Thus every solution is  expanding at the rate
$g(k)$ except for one that is exponentially decaying at the rate
$f(k)$.

It is well known that for $E$ in the spectrum the dynamics of the
transfer-matrix cocycle $A_k$ is nonuniformly hyperbolic. Moreover,
$E$ being a generalized eigenvalue of $ H_{\lambda,\alpha,\theta}$
already implies that the behavior of $A_k$ is non-regular. Theorem
\ref{Maintheoremdecay} provides precise information on how the
non-regular behavior unfolds in this case. Previously, a study of some features
of the non-regular
behavior for the almost Mathieu operator was made
in \cite{fk}. We are not aware though of other
non-artificially constructed examples of non-uniformly hyperbolic systems where non-regular
behavior can be described with such precision as in the present work.

The information provided by Theorem \ref{Maintheoremdecay} leads to
many interesting corollaries which will be explored elsewhere. Here we
only want to list a few immediate sharp consequences.

\begin{corollary}\label{corlyap}
Under the condition of Theorem \ref{Maintheoremdecay}, we have
\begin{enumerate}
\item [i)]
 $$ \limsup_{k\to \infty}\frac{\ln ||A_k||}{k}=\limsup_{k\to \infty}\frac{\ln||\tilde{U}(k)||}{k}=\ln|\lambda|,$$

\item[ii)]
 $$ \liminf_{k\to \infty}\frac{\ln ||A_k||}{k}=\liminf_{k\to \infty}\frac{\ln||\tilde{U}(k)||}{k}=\ln|\lambda|-\beta.$$
\item [iii)]
Outside an explicit sequence of lower density zero, \footnote{It will be clear
    from the proof that the sequence with convergence to the Lyapunov
    exponent contains $q_n, n=1,\cdots.$}
 $$ \lim_{k\to \infty}\frac{\ln ||A_k||}{k}=\lim_{k\to \infty}\frac{\ln||\tilde{U}(k)||}{k}=\ln|\lambda|.$$
\end{enumerate}
\end{corollary}

Therefore the  Lyapunov behavior for the norm fails to hold only along a sequence of
density zero. It is interesting that the situation is different for
the eigenfunctions. While, just like the overall growth of $\|A_k\|$
is $\ln|\lambda|-\beta$, the overall rate of decay of the eigenfunctions is
also $\ln|\lambda|-\beta$, they however decay at the Lyapunov rate only  outside a
sequence of positive upper density. 
That is
\begin{corollary}\label{corleigen}
Under the condition of Theorem \ref{Maintheoremdecay}, we have
\begin{enumerate}
\item [i)]
 $$ \limsup_{k\to \infty}\frac{-\ln||U(k)||}{k}=\ln|\lambda|,$$

\item[ii)]
 $$ \liminf_{k\to \infty}\frac{-\ln||U(k)||}{k}=\ln|\lambda|-\beta.$$
\item [iii)]
There is an explicit  sequence of upper density $1- \frac{1}{2}\frac{\beta}{\ln|\lambda|}$\color{black}, \footnote{It will be clear
    from the proof that the sequence  contains
    $\lfloor\frac{q_n}{2}\rfloor, n=1,\cdots.$}, along which
 $$ \lim_{k\to \infty}\frac{-\ln||U(k)||}{k}=\ln|\lambda|.$$
\item [iv)] There is an explicit sequence of upper  density
  $ \frac{1}{2}\frac{\beta}{\ln|\lambda|},$\footnote{As will be clear from the proof, this sequence can
    have  lower  density ranging from $ 0$ to $ \frac{1}{2}\frac{\beta}{\ln|\lambda|}$ depending on finer
    continued fraction properties of $\alpha.$}\color{black} along
  which
$$ \limsup_{k\to \infty}\frac{-\ln||U(k)||}{k}<\ln|\lambda|.$$
\end{enumerate}
\end{corollary}

The fact that $g$ is not always the reciprocal of $f$ leads also to
another interesting phenomenon.

Let $0\leq \delta_k\leq \frac{\pi}{2}$ be the angle between vectors
${U} (k)$ and $\tilde{U} (k)$.

\begin{corollary}\label{angle}
We have
 \begin{equation}\label{G.anglelimsup}
   \limsup_{k\to \infty}\frac{\ln \delta_k}{k}=0,
 \end{equation}
 and
 \begin{equation}\label{G.angleliminf}
   \liminf_{k\to \infty}\frac{\ln \delta_k}{k}=-\beta.
 \end{equation}
\end{corollary}

As becomes clear from the proof, neighborhoods of  resonances $q_n$ are the places of exponential tangencies between
contracted and expanded directions, with the rate
approaching $-\beta$ along a subsequence.\footnote{In fact the rate is
  close to $-\frac{\ln q_{n+1}}{q_n}$ for any large $n.$} Exponential tangencies also happen around points
of the form $jq_n$ but at lower strength. This means, in particular,
that $A_k$ with $k\sim q_n$ is exponentially close to a matrix with
the trace $e^{(\ln |\lambda|-\beta)k}.$

The rest of this paper is organized in the following way. We list the
definitions and standard preliminaries in Section \ref{prel}. We also include there
the  non-resonant regularity statement. While similar to the
corresponding statements in \cite{avila2009ten,MR3340177,MR3292353},
it differs in enough technical details that a proof is needed for
completeness. We present this proof in Appendix B. Section \ref{boot}
is devoted to the bootstrap localization argument, establishing sharp
upper bounds for the resonant case. Section \ref{gor} is devoted to
the lower bounds. In Section \ref{part1} we prove the statements about eigenfunctions:
(\ref{G.Asymptotics}) of Theorem \ref{Maintheoremdecay}, Theorems  \ref{universal}  and
\ref{Conjecture}. In Section \ref{hier}, we will prove the
hierarchical structure Theorem   \ref{addnewth} and Corollary \ref{corhie}. In Section \ref{part2}, we study the growth of
transfer matrices and prove (\ref{G.Asymptoticstransfer}) of Theorem
\ref{Maintheoremdecay}. The remaining Corollaries are proved in Section \ref{cor}.

 \section{Preliminaries}\label{prel}
Fix $ \alpha\in \mathbb{R}\backslash\mathbb{Q}$
  such that $ \beta(\alpha) <\infty$.  Unless  stated otherwise,   we
  always  assume $\lambda>e^{\beta}$ (for $\lambda<-e^{\beta}$,
  notice that
  $H_{\lambda,\alpha,\theta}=H_{-\lambda,\alpha,\theta+\frac{1}{2}}$),
  $\theta$ is Diophantine with respect to $\alpha$
and $E$ is a generalized eigenvalue.
  We also assume
  $\phi$  is the corresponding generalized eigenfunction of $H_{\lambda,\alpha,\theta}$.
  Without loss of generality assume $|\phi(0)|^2+|\phi(-1)|^2=1$. Let
  $ \psi$ be  any  solution to $H_{\lambda,\alpha,\theta} \psi=E\psi$ linear independent with respect to  $\phi$, i.e.,
  $|\psi(0)|^2+|\psi(-1)|^2=1$ and
  \begin{equation}\label{W0}
    \phi(-1)\psi(0)-\phi(0)\psi(-1)=c,
  \end{equation}
  where $c\neq 0$.

  Then
    by the constancy of the Wronskian, one has
  \begin{equation}\label{W}
    \phi(k+1)\psi(k)-\phi(k)\psi(k+1)=c.
  \end{equation}
  We also will denote by $\varphi$  an {\it arbitrary} solution, so
  either $\psi$ or $\phi$. Thus  for  any $k,m$, one has
 \begin{equation}\label{G.new17}
    \left(\begin{array}{c}
                                                                                            \varphi(k+m) \\
                                                                                           \varphi(k+m-1)                                                                                        \end{array}\right)
                                                                                           =A_{k}(\theta+m\alpha)
 \left(\begin{array}{c}
                                                                                            \varphi(m) \\
                                                                                          \varphi(m-1)                                                                                        \end{array}\right).
 \end{equation}
 The Lyapunov exponent 
is given  by
 \begin{equation}\label{G21}
    L(E)=\lim_{k\rightarrow\infty} \frac{1}{k}\int_{\mathbb{R}/\mathbb{Z}} \ln \| A_k(\theta)\|d\theta.
 \end{equation}
The Lyapunov exponent can be computed precisely for $E$ in the
spectrum of $H_{\lambda,\alpha,\theta}$. We denote the spectrum by
$\Sigma_{\lambda,\alpha}$ (it  does not depend on $\theta$).
 \begin{lemma}\cite{
bourgain2002continuity}\label{lya}
For $E\in \Sigma_{\lambda,\alpha}$ and $\lambda>1$, we have
 $L(E)=\ln\lambda$.
 \end{lemma}
 Recall that we always assume $E\in \Sigma_{\lambda,\alpha}$ so by
 upper semicontinuity and unique ergodicity \color{black} (e.g. \cite{furman1997multiplicative})
one has
\begin{equation}\label{G23}
   \ln\lambda=\lim_{k\rightarrow\infty} \sup_{\theta\in\mathbb{R}/ \mathbb{Z}}\frac{1}{k} \ln \| A_k(\theta)\|,
\end{equation}
that is,  the convergence in (\ref{G23}) is  uniform   with respect to  $\theta\in\mathbb{R}$.
 Precisely, $ \forall \varepsilon >0$,
\begin{equation}\label{G24}
  \| A_k(\theta)\|\leq e^{(\ln\lambda+\varepsilon)k},  \text {for  }k  \text { large enough}.
\end{equation}


We start with the basic setup going back to
\cite{jitomirskaya1999metal}. Let us denote
$$ P_k(\theta)=\det(R_{[0,k-1]}(H_{\lambda,\alpha,\theta}-E) R_{[0,k-1]}).$$
It is easy to check  that
\begin{equation}\label{G34}
 A_{k}(\theta)=
\left(
  \begin{array}{cc}
   P_k(\theta) &- P_{k-1}(\theta+\alpha)\\
    P_{k-1}(\theta) & - P_{k-2}(\theta+\alpha) \\
  \end{array}
\right).
\end{equation}

    \par
    By Cramer's rule 
 for given  $x_1$ and $x_2=x_1+k-1$, with
     $ y\in I=[x_1,x_2] \subset \mathbb{Z}$,  one has
     \begin{eqnarray}
       |G_I(x_1,y)| &=&  \left| \frac{P_{x_2-y}(\theta+(y+1)\alpha)}{P_{k}(\theta+x_1\alpha)}\right|,\label{Cramer1}\\
       |G_I(y,x_2)| &=&\left|\frac{P_{y-x_1}(\theta+x_1\alpha)}{P_{k}(\theta+x_1\alpha)} \right|.\label{Cramer2}
     \end{eqnarray}
By  (\ref{G24})  and (\ref{G34}), the numerators in  (\ref{Cramer1}) and (\ref{Cramer2}) can be bounded uniformly with respect to $\theta$. Namely,
for any $\varepsilon>0$,
\begin{equation}\label{Numerator}
    | P_k(\theta)|\leq e^{(\ln \lambda+\varepsilon)k}
\end{equation}
for $k$ large enough.
\begin{definition}\label{Def.Regular}
Fix $\tau > 0$, $0<\delta<1/2$. A point $y\in\mathbb{Z}$ will be called $(\tau,k)$ regular with $\delta$ if there exists an
interval $[x_1,x_2]$  containing $y$, where $x_2=x_1+k-1$, such that
\begin{equation*}
  | G_{[x_1,x_2]}(y,x_i)|<e^{-\tau|y-x_i|} \text{ and } |y-x_i|\geq \delta k \text{ for }i=1,2.
\end{equation*}
\end{definition}
It is  easy to check that
 \begin{equation}\label{Block}
   \varphi(x)= -G_{[x_1 ,x_2]}(x_1,x ) \varphi(x_1-1)-G_{[x_1 ,x_2]}(x,x_2) \varphi(x_2+1),
 \end{equation}
 where  $ x\in I=[x_1,x_2] \subset \mathbb{Z}$.

       \begin{definition}
     We  say that the set $\{\theta_1, \cdots ,\theta_{k+1}\}$ is $ \epsilon$-uniform if
      \begin{equation}\label{Def.Uniform}
        \max_{ x\in[-1,1]}\max_{i=1,\cdots,k+1}\prod_{ j=1 , j\neq i }^{k+1}\frac{|x-\cos2\pi\theta_j|}
        {|\cos2\pi\theta_i-\cos2\pi\theta_j|}<e^{k\epsilon}.
      \end{equation}
     \end{definition}
      Let $A_{k,r}=\{\theta\in\mathbb{R} \;|\;P_k(\cos2\pi  ( \theta -\frac{1}{2}(k-1)\alpha )  )|\leq e^{(k+1)r}\} $ with $k\in \mathbb{N}$ and $r>0$.
     We have the following Lemma.
      \begin{lemma}\label{Le.Uniform}(\text{Lemma 9.3 },\cite{avila2009ten})
      Suppose  $\{\theta_1, \cdots ,\theta_{k+1}\}$ is  $ \epsilon_1$-uniform. Then there exists some $\theta_i$ in set  $\{\theta_1, \cdots ,\theta_{k+1}\}$ such that
     $\theta_i\notin A_{k,\ln\lambda-\epsilon}$ if    $ \epsilon>\epsilon_1$ and $ k$
      is sufficiently large.
      \end{lemma}
\begin{proof}
Straightforward calculation.
\end{proof}
     We say $\theta $ is   {\it $n$-Diophantine} with respect to $\alpha,$
     if for some $\kappa>0,\nu>1$  the following hold
 \begin{equation}\label{DCthetaaddweak}
   ||2\theta+k\alpha||_{\mathbb{R}/\mathbb{Z}} > \frac{\kappa}{{q_{n}}^{\nu}},
 \end{equation}
  for all $|k|\leq 2 q_{n}$
 and
 \begin{equation}\label{DCdoublesizeweak}
   ||2\theta+k\alpha||_{\mathbb{R}/\mathbb{Z}} > \frac{\kappa}{{|k|}^{\nu}},
 \end{equation}
 for all $2q_{n}<|k|\leq  2q_{n+1}$.

     Define $b_n=   q_n^{t} $ with $\frac{8}{9}\leq t<1$ ($t$ will be defined later). For any $k>0$, 
        we will distinguish two cases with respect to $n$:
         \par
        (i)   $|k-\ell q_n|\leq b_n$ for some $\ell\geq1$,  called
        {\it $n-$resonance}.
          \par
        (ii)     $|k-\ell q_n|> b_n$ for all $\ell\geq0$, called  {\it
          $n-$nonresonance}.

For the  $n-$nonresonant $y$,  let $n_0$ be the least positive integer such that $4q_{n-n_0}\leq dist(y,  q_n\mathbb{Z})$.
 Let $s$ be the
largest positive integer such that $4sq_{n-n_0}\leq dist(y,q_n\mathbb{Z}) $.
Notice that $n_0\leq C(\alpha)$.

\par
The following theorem is similar to a statement appearing
in\cite{avila2009ten}  with modifications in
\cite{MR3340177,MR3292353}. We present a proof in Appendix B.
\begin{theorem}\label{Th.Nonresonant}
Assume $  \lambda>e^{\beta(\alpha)}$. Suppose   either

i) $b_n\leq |y|< C b_{n+1},$ where $C>1$ is a fixed constant, and  $\theta  $ is $n$-Diophantine with respect to $\alpha$

or

ii) $0\leq |y|< q_n$ and $\theta$ satisfies (\ref{DCthetaaddweak})

Then for  any $ \varepsilon>0$ and   $ n$  large enough, if $y$ is $n-$nonresonant,  we have $ y$ is $  (\ln\lambda+8\ln (s q_{n-n_0}/q_{n-n_0+1})/q_{n-n_0}-\varepsilon,4sq_{n-n_0}-1)$ regular with $\delta=\frac{1}{4}$.
\end{theorem}
\begin{remark}\label{lastversion2}
If $\theta$ is $n-1$-Diophantine with respect to $\alpha$, then (\ref{DCthetaaddweak}) holds.
\end{remark}
\begin{remark}
In the nonresonant case, for any $ \varepsilon>0,\frac{8}{9}\leq t<1$, one has  $\ln \lambda+8\ln (s q_{n-n_0}/q_{n-n_0+1})/q_{n-n_0}\geq \ln\lambda-8(1-t)\beta-\varepsilon>0$. In addition,
 we have  $\ln\lambda+8\ln (s q_{n-n_0}/q_{n-n_0+1})/q_{n-n_0}\geq \ln \lambda -2\varepsilon$ if $t$  is close to $1$.
\end{remark}
\begin{remark}
In the present paper, we only use Theorem \ref{Th.Nonresonant} with
$C=50C_{\star}$, where $C_{\star}$ is given by (\ref{G.defC}) (see the next section).
\end{remark}
\section{ Bootstrap resonant localization}\label{boot}

In this section we assume $\theta$ is $n-$Diophantine with respect to $\alpha$.
Clearly, it is enough to consider $k>0$.
In this section we study the resonant case.
Suppose there exists some $k\in[b_n,b_{n+1}]$ such that $k$ is
$n-$resonant. Then
 we have $b_{n+1}\geq \frac{q_n}{2}$.
 For any $\varepsilon>0$,
 choose   $\eta =\frac{\varepsilon}{C}$, where  $C$ is a  large
 constant (depending on $\lambda,\alpha$).

 Let
 \begin{equation}\label{G.defC}
  C_{\ast}= 2(1+\lfloor\frac{\ln \lambda}{\ln\lambda-\beta}\rfloor),
 \end{equation}
 where $\lfloor m\rfloor$ denotes  the smallest integer not exceeding $m$.

For an arbitrary solution $\varphi$ satisfying $H\varphi =E\varphi$, let
\begin{equation*}
r_j^{n,\varphi}= \sup_{|r|\leq 10 \eta }|\varphi(jq_n+rq_n)|,
\end{equation*}
 where   $|j| \leq 50 C_{\ast}\frac{b_{n+1}}{q_n}$.

Fix $\psi$ satisfying (\ref{W0}) and  denote by
\begin{equation*}
R^n_j= r_j^{n,\psi},
\end{equation*}
and
\begin{equation*}
r^n_j= r_j^{n,\phi}.
\end{equation*}
Since we
 keep $n$ fixed in this section we omit the
dependence on $n$ from the notation and write $r_j^{\varphi},R_j,$
and $r_j$.

Note that    
below we always  assume $n$ is large enough.\footnote{\label{foot1}
  The required largeness of $n$ will depend on $\alpha,\theta,\hat{C}$ in
  (\ref{shn})  and
  $\varepsilon$ whenever  $\varepsilon$ is (implicitly) present in the
statement. }
In the next Lemma and its variant, Lemma \ref{Le.resonant20}, we establish exponential decay of the eigenfunctions at non-resonant points, at the nearly Lyapunov rate, with respect to the distance to the resonances.
\begin{lemma}\label{Le.resonant}
Let   $k\in [jq_n,(j+1)q_n]$ with $ dist(k,q_n\mathbb{Z})\geq     10\eta q_n$.
 Suppose  either

i)  $|j|\leq 48 C_{\ast}\frac{b_{n+1}}{q_n}$ and $b_{n+1}\geq \frac{q_n}{2}$,

  or

  ii) $j=0$,

then for sufficiently large
$n $,
\begin{equation}\label{Intervalk}
    |\varphi(k)|\leq\max\{ r_j ^{\varphi}\exp\{-(\ln \lambda- 2\eta)(d_j-3\eta q_n)\},r_{j+1}^{\varphi}\exp\{-(\ln \lambda- 2\eta)(d_{j+1}-3\eta q_n)\}\},
\end{equation}
where $d_j=|k-jq_n|$ and  $d_{j+1}=|k-jq_n-q_n|$.
\end{lemma}

\begin{proof}
The proof builds on the ideas used in the proof of Lemma 9.11 in \cite{avila2009ten} and Lemma 3.2 in \cite{MR3340177}.
However it requires a more careful approach. 

We first prove the case i).

\par
   For any $y$ $  \in [j q_n+\eta q_{n },(j+1) q_n-
\eta q_{n }]$,  apply  i) of Theorem \ref{Th.Nonresonant} with  $C=50C_{\star}$.
 Notice that in this case, we have
 \begin{equation*}
  \ln\lambda+8\ln (s q_{n-n_0}/q_{n-n_0+1})/q_{n-n_0}-\eta \geq    \ln\lambda- 2\eta.
 \end{equation*}
  Thus    $y$ is regular with $\tau=\ln\lambda -2\eta$.
Therefore
 there exists an interval $ I(y)=[x_1,x_2]\subset
[  j  q_n,(j+1)q_n]$
such that $y\in I(y)$ and
\begin{equation}\label{G329}
    \text{dist}(y,\partial I(y))\geq  \frac{1}{4} |I(y)|   \geq  q_{n-n_0}
\end{equation}
and
\begin{equation}\label{G330}
  |G_{I(y)}(y,x_i)| \leq e^{-(\ln \lambda- 2 \eta)|y-x_i|},\;i=1,2,
\end{equation}
where
 $ \partial I(y)$ is the boundary of the interval $I(y)$, i.e.,$\{x_1,x_2\}$, and  $ |I(y)|$ is the  size of $I(y)\cap\mathbb{Z} $, i.e., $ |I(y)|=x_2-x_1+1$.
   For $z  \in  \partial I(y)$,  let
  $z' $ be the neighbor of $z$, (i.e., $|z-z'|=1$) not belonging to $I(y)$.
\par
If $x_2+1\leq (j+1)q_n-\eta q_n$ or  $x_1-1\geq  j q_n+ \eta q_n$,
we can expand $\varphi(x_2+1)$ or $\varphi(x_1-1)$ using (\ref{Block}). We can continue this process until we arrive to $z$
such that $z+1>(j+1)q_n- \eta q_n$ or  $z-1< j  q_n+ \eta q_n$, or the iterating number reaches
$\lfloor\frac{2 q_n}{ q_{n-n_0}}\rfloor$. Thus, by (\ref{Block})
\begin{equation}\label{G331}
   \varphi(k)=\displaystyle\sum_{s ; z_{i+1}\in\partial I(z_i^\prime)}
G_{I(k)}(k,z_1) G_{I(z_1^\prime)}
(z_1^\prime,z_2)\cdots G_{I(z_s^\prime)}
(z_s^\prime,z_{s+1})\varphi(z_{s+1}^\prime),
\end{equation}
where in each term of the summation one has
$j q_n+\eta q_{n }+1\leq z_i\leq (j+1) q_n-\eta
q_{n }-1$, $i=1,\cdots,s,$ and
  either $z_{s+1} \notin [j q_n+\eta q_{n }+1,(j+1) q_n-
 \eta q_{n }-1]$, $s+1 < \lfloor\frac{2 q_n}{ q_{n-n_0}}\rfloor$; or
$s+1= \lfloor\frac{2 q_n}{ q_{n-n_0}}\rfloor$.
We should mention that $z_{s+1}\in[jq_n,(j+1)q_n] $.
\par
 If $z_{s+1} \in [j q_n,j q_n+ \eta q_{n }]$, $s+1 < \lfloor\frac{2 q_n}{ q_{n-n_0}}\rfloor$,
this implies
\begin{equation*}
    |\varphi(z_{s+1}^\prime)|\leq r_j^{\varphi}.
\end{equation*}

By  (\ref{G330}), we have
\begin{equation*}
    \nonumber
   | G_{I(k)}(k,z_1) G_{I(z_1^\prime)}
(z_1^\prime,z_2)\cdots G_{I(z_s^\prime)}
(z_s^\prime,z_{s+1})\varphi(z_{s+1}^\prime)|
\end{equation*}
\begin{eqnarray}
\nonumber
&\leq & r_j ^{\varphi}e^{-(\ln\lambda-2\eta)(|k-z_1|+\sum_{i=1}^{s}|z_i^\prime-z_{i+1}|)}
 \\
\nonumber
&\leq & r_j ^{\varphi}e^{-(\ln\lambda-2\eta)(|k-z_{s+1}|-(s+1))}  \\
&\leq & r_j ^{\varphi} e^{-(\ln\lambda- 2\eta)(d_j- 2\eta q_{n } -4-\frac{ 2q_n}{ q_{n-n_0}})}.
\label{G332}
\end{eqnarray}
 If $z_{s+1} \in [(j+1) q_n-\eta q_n,(j+1) q_n ]$, $s+1 < \lfloor\frac{2 q_n}{ q_{n-n_0}}\rfloor$,
 by the same arguments,
  we have
  \begin{equation}\label{G.addGreen}
    | G_{I(k)}(k,z_1) G_{I(z_1^\prime)}
(z_1^\prime,z_2)\cdots G_{I(z_s^\prime)}
(z_s^\prime,z_{s+1})\varphi(z_{s+1}^\prime)|\leq  r_{j+1} ^{\varphi}e^{-(\ln\lambda- 2\eta)(d_{j+1}- 2\eta q_{n } -4-\frac{ 2q_n}{ q_{n-n_0}})}.
  \end{equation}
  If $s+1= \lfloor\frac{2 q_n}{ q_{n-n_0}}\rfloor,$
using   (\ref{G329}) and (\ref{G330}), we obtain
\begin{equation}\label{G333}
     | G_{I(k)}(k,z_1) G_{I(z_1^\prime)}
(z_1^\prime,z_2)\cdots G_{I(z_s^\prime)}
(z_s^\prime,z_{s+1})\varphi(z_{s+1}^\prime)|\leq e^{-(\ln\lambda-2\eta) {q_{n-n_0}} \lfloor\frac{2 q_n}{ q_{n-n_0}}\rfloor}|\varphi(z_{s+1}^\prime)| .
\end{equation}
Notice that the total number of terms in (\ref{G331})
is  at most  $2^{\lfloor\frac{2 q_n}{ q_{n-n_0}}\rfloor}$ and $d_j,d_{j+1}\geq 10\eta q_n$.  By (\ref{G332}), (\ref{G.addGreen}) and (\ref{G333}),  we have
\begin{equation}\label{G.add1}
|\varphi(k)|\leq  \max\{r_j ^{\varphi}e^{-(\ln\lambda-2\eta) (d_j-3\eta q_n) }, r_{j+1}^{\varphi} e^{-(\ln\lambda-2\eta) (d_{j+1}-3\eta q_n) }, \max_{p\in[jq_n,(j+1)q_n]}\{e^{-(\ln\lambda-2\eta) q_n }|\varphi(p)|\}\}.
\end{equation}
Now we will show that for any $p\in[jq_n,(j+1)q_n]$, one has
$ |\varphi(p)|\leq  \max\{ r_j^{\varphi},r_{j+1}^{\varphi}\}$. Then (\ref{G.add1}) implies case i) of   Lemma \ref{Le.resonant}.
Otherwise, by the definition of $r_j^{\varphi}$, if   $|\varphi(p^\prime)|$ is the largest one of $|\varphi(z)|,z\in [jq_n+10\eta q_n+1,(j+1)q_n-10\eta q_n-1]$,
then $|\varphi(p^\prime)|>\max\{ r_j^{\varphi},r_{j+1}^{\varphi}\}$. Applying (\ref{G.add1}) to $\varphi(p^\prime) $ and noticing  that  $dist (p^\prime,q_n\mathbb{Z})\geq 10\eta q_n$,
we get
\begin{equation*}
|\varphi(p^\prime)|\leq   e^{-7(\ln\lambda-2\eta)\eta q_n  } \max\{ r_j^{\varphi},r_{j+1}^{\varphi},|\varphi(p^\prime)|\}.
\end{equation*}
This is impossible  because $|\varphi(p^\prime)|>\max\{ r_j^{\varphi},r_{j+1}^{\varphi}\}$.

Now we turn to the proof of case ii). Notice that in proving case i) of Lemma \ref{Le.resonant}, we only used case i) of Theorem \ref{Th.Nonresonant}.
Using  case ii) of Theorem \ref{Th.Nonresonant} instead we can prove
case ii) of  Lemma \ref{Le.resonant} by the same reasoning. In order
to avoid repetition, we
omit the details.
\end{proof}

Lemma \ref{Le.resonant} is sufficient for our current purposes, but for the purposes of Section \ref{hier} we will need a similar statement that allows for shifts and reflections. For $B\in \mathbb{Z},$ let $r^{n,\varphi}_{j,\pm}(B)= \sup_{|r|\leq 10 \eta }|\varphi(B\pm (jq_n+rq_n))|.$ For  $y\in[B\pm jq_n\pm \eta q_n, B\pm (j+1)q_n \mp \eta q_n]$, let $n_0$ be the least positive integer such that $4q_{n-n_0}\leq dist(y-B,  q_n\mathbb{Z})$ and  $s$ be the
largest positive integer such that $4sq_{n-n_0}\leq dist(y-B,q_n\mathbb{Z}) $.
Since we only used the appropriate regularity of the non-resonant $y$, the proof of Lemma \ref{Le.resonant} also establishes the following Lemma
\begin{lemma}\label{Le.resonant20}
Suppose for any $y\in[B\pm jq_n\pm \eta q_n, B\pm (j+1)q_n \mp \eta q_n]$,  $ y$ is $  (\ln\lambda+8\ln (s q_{n-n_0}/q_{n-n_0+1})/q_{n-n_0}-\varepsilon,4sq_{n-n_0}-1)$ regular with $\delta=\frac{1}{4}$.
Let   $k-B\in\pm [jq_n,(j+1)q_n]$ with $ dist(k-B,q_n\mathbb{Z})\geq     10\eta q_n$.
 Suppose  either

i)  $|j|\leq 48 C_{\ast}\frac{b_{n+1}}{q_n}$ and $b_{n+1}\geq \frac{q_n}{2}$,

  or

  ii) $j=0$,

then for sufficiently large
$n $,


Then we have
\begin{equation}\label{Intervalk20}
    |\varphi(k)|\leq\max\{ r_{j,\pm} ^{\varphi}(B)\exp\{-(\ln \lambda- 2\eta)(d_j-3\eta q_n)\},r_{j+1,\pm}^{\varphi}(B)\exp\{-(\ln \lambda- 2\eta)(d_{j+1}-3\eta q_n)\}\}.
\end{equation}
where $d_j=|k-B\mp jq_n|$ and  $d_{j+1}=|k-B\mp (j+1)q_n|$.
\end{lemma}

By Theorem \ref{Th.Nonresonant} , Lemma  \ref{Le.resonant}  is a particular case of Lemma \ref{Le.resonant20}, when $B=0$ and the sign is a $+.$ Going back to this case, we will prove

\begin{lemma}\label{Le.r_jvarphi}
For $  1\leq j   \leq  46C_{\star}  \frac{b_{n+1}}{q_n}$ with $b_{n+1}\geq \frac{q_n}{2}$, the following holds
\begin{equation}\label{G.Secondadd2}
   r_j^{\varphi}\leq  \max\{r_{j\pm1}^{\varphi}\frac{q_{n+1}}{j}\exp\{-(\ln \lambda - C\eta)q_n\}  \}.
\end{equation}
\end{lemma}
\begin{proof}
Fix  
$j$  with $  1\leq j   \leq 46C_{\ast}\frac{b_{n+1}}{q_n}$  and $|r|\leq 10\eta q_n$.
Set $I_1, I_2\subset \mathbb{Z}$ as follows
\begin{eqnarray*}
  I_1 &=& [-\lfloor\frac{1}{2}q_n\rfloor, q_n-\lfloor\frac{1}{2}q_n\rfloor-1], \\
   I_2 &=& [ j  q_n-\lfloor\frac{1}{2}q_n\rfloor, (j +1)q_n-\lfloor\frac{1}{2}q_n\rfloor-1 ].
\end{eqnarray*}

Let $\theta_m=\theta+m\alpha$ for $m\in I_1\cup I_2$. The set $\{\theta_m\}_{m\in I_1\cup I_2}$
consists of $2q_n$ elements.
\par

\par
By arguments similar to those in Lemma 9.13 in \cite{avila2009ten} or Theorem 3.1 in \cite{MR3340177}, one has $\{\theta_m\}$ is
$\frac{\ln q_{n+1}-\ln j}{2q_n}+\varepsilon$
uniform
for any $\varepsilon>0.$  Since our case is slightly different we
prove it as Theorem \ref{Thapp}  in Appendix B.
    Combining with  Lemma  \ref{Le.Uniform}, there exists some $j_0$ with  $j_0\in I_1\cup I_2$
    such that
      $ \theta_{j _0}\notin  A_{2q_n-1,\ln\lambda-\frac{\ln q_{n+1}-\ln j}{2q_n}-\eta }$.

      First, we assume $j_0\in I_2$.
      \par
      Set $I=[j_0-q_n+1,j_0+q_n-1]=[x_1,x_2]$.   In (\ref{Numerator}), let $\varepsilon=\eta.$   Combining with  (\ref{Cramer1})  and  (\ref{Cramer2}),
 it is easy to verify
\begin{equation*}
|G_I(jq_n+r,x_i)|\leq e^{(\ln\lambda+\eta )(2q_n-1-|jq_n+r-x_i|)-(2q_n-1)(\ln\lambda -\frac{\ln q_{n+1}-\ln j}{2q_n}-\eta)}.
\end{equation*}
Using (\ref{Block}), we obtain
\begin{equation}\label{Iterationr_j}
    |\varphi(j q_n+r)|  \leq \sum_{i=1,2} \frac{q_{n+1}}{j}e^{5\eta q_n}|\varphi(x_i^{\prime})|e^{-|jq_n+r-x_i|\ln \lambda },
\end{equation}
where $x_1^{\prime}=x_1-1$ and $x_2^{\prime}=x_2+1$.

Let $d_j^i=  |x_i-jq_n|  $, $i=1,2.$
It is easy to check that
\begin{equation}\label{G.Distance}
   |jq_n+r-x_i|+d_j^i,|jq_n+r-x_i|+d_{j\pm1}^i \geq q_n-|r|,
\end{equation}
and
\begin{equation}\label{G.Distanceadd}
    |jq_n+r-x_i|+d_{j\pm2}^i \geq 2q_n-|r|.
\end{equation}
If $ dist(x_i,q_n\mathbb{Z})\geq     10\eta q_n $, then we bound $\varphi(x_i)$  in  (\ref{Iterationr_j}) using (\ref{Intervalk}).
If   $dist(x_i,q_n\mathbb{Z})\leq 10\eta q_n$, then we bound $\varphi(x_i)$  in  (\ref{Iterationr_j}) by some proper $r_j$.
Combining with (\ref{G.Distance}), (\ref{G.Distanceadd}),
we have
\begin{equation*}
   r_j^{\varphi}\leq \max\{r_{j\pm1}^{\varphi}\frac{q_{n+1}}{j}\exp\{-(\ln \lambda - C\eta)q_n\}   ,r_{j}^{\varphi}\frac{q_{n+1}}{j}\exp\{-(\ln \lambda - C\eta)q_n\}   ,r_{j\pm 2}^{\varphi}\frac{q_{n+1}}{j}\exp\{-2(\ln \lambda - C\eta)q_n\}\}.
\end{equation*}

However
\begin{eqnarray*}
 r_j ^{\varphi}&\leq&  r_{j}^{\varphi}\frac{q_{n+1}}{j}\exp\{-(\ln \lambda - C\eta)q_n\} \\
   &\leq& r_{j}^{\varphi}\exp\{-(\ln \lambda-\beta- C\eta)q_n\}
\end{eqnarray*}
cannot happen, so we must have
\begin{equation}\label{G.Secondadd1}
   r_j^{\varphi}\leq  \max\{r_{j\pm1}^{\varphi}\frac{q_{n+1}}{j}\exp\{-(\ln \lambda - C\eta)q_n\}   ,r_{j\pm 2}^{\varphi}\frac{q_{n+1}}{j}\exp\{-2(\ln \lambda - C\eta)q_n\}\}.
\end{equation}
In particular,
\begin{equation}\label{G.Secondadd1zhou}
   r_j^{\varphi}\leq \exp\{-(\ln \lambda-\beta-C\eta)q_n\} \max\{r_{j\pm1}^{\varphi}\,r_{j\pm 2}^{\varphi}\}.
\end{equation}
If  $j_0\in I_1$,  then (\ref{G.Secondadd1zhou}) holds for  $j=0$.  Let $\varphi=\phi$ in (\ref{G.Secondadd1zhou}).  We  get
\begin{equation*}
 |\phi(0)|,|\phi(-1)|\leq  \exp\{-(\ln \lambda-\beta-C\eta)q_n\},
\end{equation*}
 this is in contradiction with  $|\phi(0)|^2+|\phi(-1)|^2=1$.
Therefore  $j_0\in I_2$,
so   (\ref{G.Secondadd1}) holds for any $\varphi$.

  By (\ref{G.new17}) and (\ref{G24}), we have
  \begin{equation}\label{G.new18}
  ||\left(\begin{array}{c}
                                                                                            \varphi(k_1) \\
                                                                                           \varphi(k_1-1)                                                                                        \end{array}\right)||\geq Ce^{-(\ln\lambda+\varepsilon)|k_1-k_2|}||\left(\begin{array}{c}
                                                                                            \varphi(k_2) \\
                                                                                           \varphi(k_2-1)                                                                                        \end{array}\right)||.
\end{equation}
  This implies
  \begin{equation*}
    r_{j\pm 2}^{\varphi}\leq  r_{j\pm 1}^{\varphi}\exp\{(\ln \lambda +C\eta)q_n\},
  \end{equation*}
  thus (\ref{G.Secondadd1}) becomes
  \begin{equation}\label{G.Secondadd2}
   r_j^{\varphi}\leq  \max\{r_{j\pm1}^{\varphi}\frac{q_{n+1}}{j}\exp\{-(\ln \lambda - C\eta)q_n\}  \},
\end{equation}
for any $  1\leq j  \leq 46C_{\ast}\frac{b_{n+1}}{q_n}$.
\end{proof}

For solution $\phi$ and $\psi$ we can also get a more subtle estimate.

\begin{theorem}\label{Le.r_j}
For $  1\leq j   \leq  10  \frac{b_{n+1}}{q_n}$ with $b_{n+1}\geq \frac{q_n}{2}$, the following holds
\begin{equation}\label{r_jnew}
   r_{j}\leq   r_{j-1}\exp\{-(\ln \lambda - C\eta)q_n\}\frac{q_{n+1}}{j} .
\end{equation}

\end{theorem}
\begin{proof}


\par
Let $\varphi=\phi$ in Lemma \ref{Le.r_jvarphi}. We must have
\begin{equation}\label{G.Secondadd2phi}
   r_j\leq  \max\{r_{j\pm1}\frac{q_{n+1}}{j}\exp\{-(\ln \lambda - C\eta)q_n\}  \},
\end{equation}
for any $  1\leq j  \leq 46C_{\ast}\frac{b_{n+1}}{q_n}$.

Suppose for some $  1\leq j  \leq 10\frac{b_{n+1}}{q_n}$, the following holds,
\begin{equation}\label{G.Secondadd3}
 r_j\leq   r_{j+1}\frac{q_{n+1}}{j}\exp\{-(\ln \lambda - C\eta)q_n\}\leq r_{j+1} \exp\{-(\ln \lambda-\beta-C\eta)q_n\} .
\end{equation}
Applying (\ref{G.Secondadd2phi}) to $j+1$, we obtain
\begin{equation}\label{G.Secondadd4}
   r_{j+1}\leq  \max\{r_{j},r_{j+2}\}\frac{q_{n+1}}{j+1}\exp\{-(\ln \lambda - C\eta)q_n\}  .
\end{equation}
Combining with (\ref{G.Secondadd3}), we must have
\begin{equation}\label{G.Secondadd5}
   r_{j+1}\leq   r_{j+2} \exp\{-(\ln \lambda-\beta-C\eta)q_n  \}.
\end{equation}
Generally, for any $0<p\leq (C_{\ast}+1)j -1$, we obtain
\begin{equation}\label{G.Secondadd5}
   r_{j+p}\leq   r_{j+p+1} \exp\{-(\ln \lambda-\beta-C\eta)q_n  \}.
\end{equation}
Thus
\begin{equation}\label{G.Secondadd6}
   r_{(C_{\ast}+1)j}\geq   r_{j} \exp\{(\ln \lambda-\beta-C\eta) C_{\ast}j q_n  \}.
\end{equation}
Clearly, by (\ref{G.new18}), one has
\begin{equation*}
 r_{j}\geq  \exp\{-(\ln \lambda +C\eta)j q_n  \}.
\end{equation*}
Then
\begin{equation}\label{G.Secondadd9}
   r_{ (C_{\ast}+1)j}\geq   \exp\{((C_{\ast}-1)\ln \lambda-C_{\ast}\beta-C\eta) j q_n  \}.
\end{equation}
By the definition of $C_{\ast}$, one has
\begin{equation*}
(C_{\ast}-1)\ln \lambda-C_{\ast}\beta> 0.
\end{equation*}
Thus  (\ref{G.Secondadd9}) is in contradiction with  the fact that  $|\phi(k)|\leq 1+|k|$.

Now that (\ref{G.Secondadd3}) can not happen,   from (\ref{G.Secondadd2phi}), we must have
\begin{equation}\label{G.Secondadd7}
   r_j\leq  r_{j-1}\frac{q_{n+1}}{j}\exp\{-(\ln \lambda - C\eta)q_n\} .
\end{equation}
\end{proof}
\begin{theorem}\label{Le.r_jpsi}
For $  0\leq j   \leq  8  \frac{b_{n+1}}{q_n}$ with $b_{n+1}\geq \frac{q_n}{2}$, the following holds
\begin{equation}\label{r_jnewpsi}
   R_{j}\leq   R_{j+1}\exp\{-(\ln \lambda - C\eta)q_n\}\frac{q_{n+1}}{j+1} .
\end{equation}

\end{theorem}
\begin{proof}
If $j=0$, (\ref{r_jnewpsi}) holds directly by (\ref{r_jnew}) (applying
it with $j=1$) and (\ref{W}).
Now we consider $j\geq 1$.
Let $\varphi=\psi$ in Lemma \ref{Le.r_jvarphi}. Then  (\ref{G.Secondadd2}) also holds for $R_j$ with $j\geq 1$, that is
\begin{equation}\label{G.Secondadd2psi}
   R_j\leq  \max\{R_{j\pm1}\frac{q_{n+1}}{j+1}\exp\{-(\ln \lambda - C\eta)q_n\}  \}.
\end{equation}
Suppose  for some $j\geq 1$
\begin{equation}\label{G.Secondadd2psi1}
   R_j\leq   R_{j-1}\frac{q_{n+1}}{j+1}\exp\{-(\ln \lambda - C\eta)q_n\}.
\end{equation}
Applying (\ref{G.Secondadd2psi}) to $j-1$ and taking into account (\ref{G.Secondadd2psi1}), one has
\begin{equation}\label{G.Secondadd2psi5}
   R_{j-1}\leq R_{j-2}\frac{q_{n+1}}{j}\exp\{-(\ln \lambda - C\eta)q_n\} .
\end{equation}
Iterating $j$ times, we must have
\begin{equation}\label{G.Secondadd2psi2}
   R_j\leq   R_{0}\frac{q_{n+1}^j}{(j+1)!}\exp\{-(\ln \lambda - C\eta)jq_n\}\leq  R_{0} \exp\{-(\ln \lambda -\beta -C\eta)jq_n\}.
\end{equation}
Similarly, iterating (\ref{r_jnew}) $j$ times,
 we have
 \begin{equation}\label{G.Secondadd2psi3}
   r_j\leq    r_{0} \exp\{-(\ln \lambda -\beta -C\eta)jq_n\}.
\end{equation}
(\ref{G.Secondadd2psi2}) and (\ref{G.Secondadd2psi3}) contradict  (\ref{W}).
This implies  (\ref{G.Secondadd2psi1}) can not happen, thus we must have
(\ref{r_jnewpsi}).
\end{proof}
\section{Lower bounds on decaying solution in the resonant case}\label{gor}
In this section we assume $\theta$ is $n-$Diophantine with respect to $\alpha$.
We  will 
study the lower bound   on $\phi$ for the resonant sites.  Recall that  $b_{n+1}\geq\frac{q_n}{2}$ in this case.
 \begin{theorem}\label{Th.Lowerbound}
 Let  $ \tilde{r}_{j}=||\left(\begin{array}{c}
        \phi(j q_n)\\
       \phi(j q_n-1)
     \end{array}\right)||
 $.
  Suppose $1\leq j\leq 8\frac{b_{n+1} }{q_n}$ with $b_{n+1}\geq\frac{q_n}{2}$, then we must have
 \begin{equation}\label{G.new8}
  \tilde{r}_{j}\geq \frac{q_{n+1}}{j}e^{-(\ln \lambda+\varepsilon)q_n}\tilde{r}_{j-1}.
 \end{equation}
 \end{theorem}
We first list two standard facts.
 \begin{lemma}(\cite{simon1985almost} )\label{Lemgordonidea2}
Let $A^1,A^2,\cdots,A^n$ and  $B^1,B^2,\cdots,B^n$ be $2\times2$ matrices with $||\prod_{m=0}^{\ell-1}A^{j+m}||\leq C e^{d\ell}$ for some constant $C$ and $d$.
Then
\begin{equation*}
    ||(A^n+B^n)\cdots(A^1+B^1)-A^n\cdots A^1||\leq Ce^{dn} (\prod_{j=1}^n(1+Ce^{-d}||B^j||)-1).
\end{equation*}

\end{lemma}
\begin{lemma}
For any $\varepsilon>0$ and large $n$ the following hold,
\begin{equation}\label{G.new20}
||A_{q_n}(\theta+ {q}_n\alpha)-A_{ {q}_n}(\theta)||\leq \frac{1}{q_{n+1}}e^{(\ln\lambda+\varepsilon)q_n},
\end{equation}
and
\begin{equation}\label{G.newnew20}
||A_{q_n}^{-1}(\theta+ {q}_n\alpha)-A_{ {q}_n}^{-1}(\theta)||\leq \frac{1}{q_{n+1}}e^{(\ln\lambda+\varepsilon)q_n}.
\end{equation}
\end{lemma}
\begin{proof}
We only prove (\ref{G.new20}) for simplicity.
By the DC approximation(or see (\ref{GDC2}) in Appendix), we have
\begin{equation*}
    ||{q}_n\alpha||_{\mathbb{R}/\mathbb{Z}}\leq \frac{1}{q_{n+1}} .
\end{equation*}

This implies
\begin{equation*}
  ||A(\theta+q_n\alpha)-A(\theta)||\leq \frac{C}{q_{n+1}}.
\end{equation*}
Applying Lemma \ref{Lemgordonidea2} and (\ref{G24}),
 one has
 \begin{equation}\label{G.new21}
 ||A_{q_n}(\theta+ {q}_n\alpha)-A_{ {q}_n}(\theta)||\leq  e^{(\ln\lambda+\varepsilon)q_n}((1+\frac{C}{q_{n+1}})^{q_{n}}-1).
 \end{equation}
 Using the fact $|e^y-1|\leq ye^y$ for $y>0$,
 we obtain
 \begin{eqnarray*}
   (1+\frac{C}{q_{n+1}})^{q_{n}}-1  &\leq& q_n (1+\frac{C}{q_{n+1}})^{q_{n}}\ln(1+\frac{C}{q_{n+1}})\\
    &\leq&  C\frac{q_n}{q_{n+1}}.
 \end{eqnarray*}
 Combining this with (\ref{G.new21}) completes the proof.
\end{proof}
 \begin{lemma}\label{Le.Lowerbound}
 For any $0\leq j\leq 8\frac{b_{n+1} }{q_n}-1$,
 one of the following two estimates must hold,
\begin{equation}\label{G.new6}
  \tilde{r}_{j+1}\geq \frac{q_{n+1}}{j+1}e^{-(\ln \lambda+\varepsilon)q_n}\tilde{r}_{j},
\end{equation}
or
\begin{equation}\label{G.new7}
  \tilde{r}_{j+1}\tilde{r}_{j-1}\geq   (1-\frac{1}{10(j+1)})^2(1-\frac{1}{10(j+1)^2})\tilde{r}_{j}^2.
\end{equation}
 \end{lemma}
 \begin{proof}
 Suppose
\begin{equation}\label{G.Assumption}
  \tilde{r}_{j+1}\leq \frac{q_{n+1}}{j+1}e^{-(\ln \lambda+\varepsilon)q_n}\tilde{r}_{j}.
\end{equation}

Let  $ U_{j}= \left(\begin{array}{c}
        \phi(j q_n)\\
       \phi(j q_n-1)
     \end{array}\right)
 $,
 then  for $n>0$, one has
 \begin{equation*}
     U_{j}=A_{q_n}(\theta+(j -1)q_n\alpha)  U_{j-1}.
 \end{equation*}

Denote  $B=A_{{q}_n}(\theta+jq_n\alpha)$. Notice that  $\det B=1$. We have
\begin{equation}\label{Trnew}
    B^2+(\text{Tr}  B)  B +I=0.
\end{equation}
\textbf{Case 1:}   $\text{Tr}  B\leq \frac{\tilde{r}_{j}}{\gamma\tilde{r}_{j+1}}$, where
\begin{equation*}
  1-\frac{1}{\gamma}=\frac{1}{10(j+1)}.
\end{equation*}
Applying (\ref{Trnew}) to $U_{j}$,
 one has
 \begin{equation}\label{Tr1}
    B^2U_{j}+(\text{Tr}  B)  BU_{j} +U_{j}=0.
\end{equation}
Notice that $
  U_{j+1}=  BU_{j}
$, thus
\begin{equation*}
  ||(\text{Tr}  B)  BU_{j}||\leq \frac{1}{\gamma} \tilde{r}_{j}.
\end{equation*}
Thus we have
 \begin{equation}\label{G.new1}
 ||B^2 U_j||\geq   (1-\frac{1}{\gamma}) \tilde{r}_{j}=\frac{1}{10(j+1)}\tilde{r}_{j}.
\end{equation}
This is impossible. Indeed,  from the following estimate
\begin{eqnarray*}
  ||U_{j+2}-B^2U_j|| &\leq& ||A_{ q_n}(\theta+ (j+1)q _n\alpha)-A_{q_n}(\theta+jq_n\alpha)||\;\;||U_{j+1}||\\
   &\leq & e^{(\ln\lambda+\frac{1}{2}\varepsilon) q_n}\frac{1}{q_{n+1}}\tilde{r}_{j+1}\\
 &\leq& \frac{1}{100(j+1)}\tilde{r}_{j},
\end{eqnarray*}
where the second  inequality holds by (\ref{G.new20})
and the third inequality holds by  assumption (\ref{G.Assumption}),
combining with (\ref{G.new1}), one has

\begin{equation}\label{G.new2}
 ||U_{j+2}||=    \tilde{r}_{j+2} \geq \frac{9}{100(j+1)}  \tilde{r}_{j}.
\end{equation}
However, by (\ref{r_jnew}) and (\ref{G.new18}),
\begin{equation*}
    \tilde{r}_{j+2} \leq \frac{q^2_{n+1}}{(j+1)(j+2)} e^{-2(\ln\lambda-C\eta)q_n} \tilde{r}_{j}.
\end{equation*}
    This is  in contradiction with  (\ref{G.new2}).
\par
\textbf{Case 2:} It remains to consider
\begin{equation}\label{G.new5}
\text{Tr}  B\geq \frac{\tilde{r}_{j}}{\gamma\tilde{r}_{j+1}}.
\end{equation}
 From (\ref{Trnew}),
\begin{equation}\label{Tr2}
    B U_{j}+(\text{Tr}  B)   U_{j} +B^{-1}U_{j}=0.
\end{equation}
First by assumption (\ref{G.Assumption}), one has
\begin{equation*}
  \tilde{r}_{j+1}\leq \frac{1}{10(j+1)}\tilde{r}_{j},
\end{equation*}
then
\begin{eqnarray*}
||BU_j||= \tilde{r}_{j+1} &\leq&  \frac{\tilde{r}_{j}}{\gamma \tilde{r}_{j+1}}\frac{\tilde{r}_{j}}{10(j+1)^2} \\
  &\leq&  ||(\text{Tr}  B) U_j| |\frac{1 }{10(j+1)^2}.
\end{eqnarray*}
Thus by (\ref{Tr2}), we have
\begin{eqnarray}
\nonumber
   ||B^{-1}U_j||&\geq &  (1-\frac{1}{10(j+1)^2}) ||(\text{Tr}  B) U_j| |\\
   &\geq &  (1-\frac{1}{10(j+1)^2}) \frac{\tilde{r}_{j}^2}{\gamma\tilde{r}_{j+1}}\label{G.new3} \\
   &\geq & (1-\frac{1}{10(j+1)^2})\frac{1}{\gamma}  \frac{j+1}{q_{n+1}}e^{(\ln\lambda+\varepsilon)q_n}\tilde{r}_{j}\label{G.new4},
\end{eqnarray}
where the second inequality holds by (\ref{G.new5}) and the third inequality hold by (\ref{G.Assumption}).

 \par
  By  (\ref{G.newnew20}), the following holds
\begin{eqnarray}
  ||U_{j-1}-B^{-1}U_j|| &\leq& ||A_{q_n}^{-1}(\theta+(j-1){q}_n\alpha)-A^{-1}_{ {q}_n}(\theta+jq_n\alpha)||\;\;||U_j||\nonumber\\
   &\leq & e^{(\ln\lambda+\frac{1}{2}\varepsilon) q_n}\frac{1}{q_{n+1}}\tilde{r}_{j}\nonumber\\
   &\leq &  \frac{1}{10(j+1)}||B^{-1}U_j||,\label{G.new23}
\end{eqnarray}
where the third inequality holds by (\ref{G.new4}).

Putting  (\ref{G.new3}) and (\ref{G.new23}) together, we have
\begin{eqnarray*}
\tilde{r}_{j-1} =|| U_{j-1}|| &\geq& (1-\frac{1}{10(j+1)})||B^{-1}U_j|| \\
  &\geq&  (1-\frac{1}{10(j+1)})^2  (1-\frac{1}{10(j+1)^2}) \frac{\tilde{r}_{j}^2}{\tilde{r}_{j+1}} .
\end{eqnarray*}
This implies (\ref{G.new7}).
 \end{proof}
 \textbf{Proof of Theorem \ref{Th.Lowerbound}}.
 \begin{proof}
 We can proceed by induction. 
 \par
 Set $j=0$ in Lemma \ref{Le.Lowerbound}.
 By (\ref{r_jnew}), the second case  (\ref{G.new7}) can not happen,
 thus  Theorem \ref{Th.Lowerbound} holds for $j=1$.
 \par
 Suppose   (\ref{G.new8}) holds for $p= j-1$, that is
 \begin{equation}\label{G.new9}
  \tilde{r}_{j-1}\geq \frac{q_{n+1}}{j-1}e^{-(\ln \lambda+\varepsilon)q_n}\tilde{r}_{j-2}.
\end{equation}
We will show (\ref{G.new8}) holds for $p= j$.
 Let us apply Lemma \ref{Le.Lowerbound} to $p=j-1$.
 If (\ref{G.new6}) holds for $p=j-1$, the result follows.
 Otherwise by (\ref{G.new7}),
 we have \begin{eqnarray*}
           \tilde{r}_{j} &\geq &  (1-\frac{1}{10j})^2(1-\frac{1}{10j^2})\tilde{r}_{j-1}\frac{\tilde{r}_{j-1}}{\tilde{r}_{j-2}} \\
             &\geq&(1-\frac{1}{10j})^2(1-\frac{1}{10j^2}) \tilde{r}_{j-1}\frac{q_{n+1}}{j-1}e^{-(\ln \lambda+\varepsilon)q_n}\\
             &\geq&  \tilde{r}_{j-1}\frac{q_{n+1}}{j}e^{-(\ln \lambda+\varepsilon)q_n},
         \end{eqnarray*}
 where the second inequality holds by (\ref{G.new9}).
 \par

 \end{proof}
\section{Decaying solutions. Proof  of (\ref{G.Asymptotics}), Theorems
 \ref{universal}, \ref{addnewth} and \ref{Conjecture}} \label{part1}
In this section the dependence on $n$ will play a role, so we go back
to the  $r_j^n,\tilde{r}_j^n$ notation.
We first give a series of auxiliary facts. Recall footnote \ref{foot1}\color{black}.
\begin{theorem}\label{Th.new1}
Assume $\theta$ is $n-$Diophantine with respect to $\alpha$.
For any  $1\leq j\leq 10\frac{b_{n+1}}{q_n}$ with $b_{n+1}\geq \frac{q_n}{2}$, we have
\begin{equation*}
\bar{r}_j^n e^{-\varepsilon j q_n }\leq r_j^n\leq \bar{r}_j^n e^{\varepsilon j q_n }
\end{equation*}
and
\begin{equation*}
\bar{r}_j^n e^{-\varepsilon j q_n } \leq\tilde{r}_j^n\leq \bar{r}_j^n e^{\varepsilon j q_n }.
\end{equation*}
\end{theorem}
\begin{proof}
For any $\varepsilon>0$,  we choose $\eta$ small enough.
Using (\ref{r_jnew}) $j$ times, we have
\begin{equation*}
r_{j}^n\leq    \frac{q_{n+1}^{j}}{j!}\exp\{ -(\ln \lambda -\varepsilon)j q_n\}.
\end{equation*}
Similarly, using (\ref{G.new8}) $j$ times, we have
\begin{equation*}
\tilde{r}_{ j}^n\geq  \frac{q_{n+1}^{j}}{j!}\exp\{ -(\ln \lambda +\varepsilon)j q_n\}.
\end{equation*}

 By Stirling formula and (\ref{G.new18}), we obtain the theorem.
\end{proof}
\begin{theorem}\label{Th.new2}
Assume $\theta$ is $n-$Diophantine with respect to $\alpha$.
Assume $j q_n\leq k<(j+1)q_n$ with $0\leq j \leq 8\frac{b_{n+1}}{q_n}$, $b_{n+1}\geq \frac{q_n}{2}$ and $k\geq \frac{q_n}{4}$.
We have
\begin{equation}\label{G.newnewnew1add}
  ||U(k)||\leq \max\{e^{-|k-j q_n|\ln\lambda}\tilde{r}_{j}^n,e^{-|k-(j +1)q_n|\ln\lambda}\tilde{r}_{j+1}^n\}e^{\varepsilon q_n},
\end{equation}
\begin{equation}\label{G.newnewnew2add}
  ||U(k)||\geq \max\{e^{-|k -jq_n|\ln\lambda}\tilde{r} _{j}^n,e^{-|k-(j +1)q_n|\ln\lambda}\tilde{r}_{j+1}^n\}e^{-\varepsilon q_n}.
\end{equation}
In particular, we have
\begin{equation}\label{G.newnewnew1}
  ||U(k)||\leq \max\{e^{-|k-j q_n|\ln\lambda}\bar{r}_{j}^n,e^{-|k-(j +1)q_n|\ln\lambda}\bar{r}_{j+1}^n\}e^{\varepsilon|k|},
\end{equation}
\begin{equation}\label{G.newnewnew2}
  ||U(k)||\geq \max\{e^{-|k -jq_n|\ln\lambda}\bar{r} _{j}^n,e^{-|k-(j +1)q_n|\ln\lambda}\bar{r}_{j+1}^n\}e^{-\varepsilon|k|}.
\end{equation}
\end{theorem}

 \begin{proof}
(\ref{G.newnewnew1}) and (\ref{G.newnewnew2}) just follows from  (\ref{G.newnewnew1add}), (\ref{G.newnewnew2add}) and Theorem \ref{Th.new1}.
Thus it suffices to prove (\ref{G.newnewnew1add}), (\ref{G.newnewnew2add}).
 Clearly, by (\ref{G.new18}), one has
\begin{equation*}
  ||U(k)||\geq \max\{e^{-|k-jq_n|\ln\lambda}\tilde{r}_{j}^n,e^{-|k-(j +1)q_n|\ln\lambda}\tilde{r}_{j+1}^n\}e^{-\varepsilon q_n}.
\end{equation*}
This implies  (\ref{G.newnewnew2}) by  Theorem \ref{Th.new1}.
\par

We now turn to   (\ref{G.newnewnew1}).
If $|k-j q_n|\leq  10\eta q_n$ or $|k-(j +1)q_n|\leq  10\eta q_n$, the
result follows from Theorem \ref{Th.new1} and (\ref{G.new18}).
 If $|k-j q_n|\geq  10\eta q_n$  and  $|k-(j +1)q_n|\geq  10\eta q_n$,
 it follows from  Lemma \ref{Le.resonant}, Theorem \ref{Th.new1} and (\ref{G.new18}).
 \end{proof}
 \begin{theorem}\label{Th.new3}

    For $q_n^{\frac{8}{9}}\leq k\leq \frac{q_n}{2}$, let $n_0$ be the smallest positive   integer such that
$q_{n-n_0}\leq k< q_{n-n_0+1}$.
Suppose $j q_{n-n_0}\leq k< (j+1)q_{n-n_0}$ with $j\geq 1$.
If    $\theta$ is $k-$Diophantine with respect to $\alpha$ for  $k=n-n_0$ and $k=n-1$,\color{black}
 then
\begin{equation}\label{G.newnewnew4}
||U(k)||\leq \max\{e^{-|k-j q_{n-n_0}|\ln\lambda}\bar{r}_{j}^{n-n_0},e^{-|k-(j+1)q_{n-n_0}|\ln\lambda}\bar{r}_{j+1}^{n-n_0}\}e^{\varepsilon k},
\end{equation}
and
\begin{equation}\label{G.newnewnew5}
 ||U(k)||\geq  \max\{e^{-|k -j q_{n-n_0}|\ln\lambda}\bar{r} _{j}^{n-n_0},e^{-|k-(j +1)q_{n-n_0}|\ln\lambda}\bar{r}_{j+1}^{n-n_0}\}e^{-\varepsilon k}.
\end{equation}
 \end{theorem}

 \begin{proof}
Set $t_0=1-\frac{\varepsilon}{8\beta}$.  Let $t=t_0$ in the definition of resonance, i.e. $b_n=q_n^{t_0}$.

 Case 1:  $k\leq  q_{n-n_0+1}^{t_0}$. In this case, one has
 \begin{equation*}
   q_{n-n_0} \leq k\leq  q_{n-n_0+1}^{t_0}.
 \end{equation*}
 The result holds by Theorem \ref{Th.new2}.
 \par
 Case 2:  $k\geq   q_{n-n_0+1}^{t_0}$. Then
 \begin{eqnarray*}
    \bar{r} _{j}^{n-n_0} &\leq &  \exp\{ -(\ln\lambda -\frac{\ln q_{n-n_0+1}}{q_{n-n_0}}+\frac{\ln q_{n-n_0+1}^{t_0}}{q_{n-n_0}}-\varepsilon) jq_{n-n_0} \}\\
    &\leq &  \exp\{ -(\ln\lambda-(1-t_0)\beta-\varepsilon)jq_{n-n_0}\} \\
    &\leq &  \exp\{ -(\ln\lambda-2\varepsilon)jq_{n-n_0}\},
 \end{eqnarray*}
 where the third inequality holds by the definition of $t_0$.
Noting that $k\leq q_{n-n_0+1}$, one has
 \begin{equation*}
   \bar{r} _{j}^{n-n_0}\geq\exp\{- jq_{n-n_0}(\ln\lambda+\varepsilon)\}.
 \end{equation*}
 Similarly,
 \begin{equation*}
   \exp\{- (j+1)q_{n-n_0}(\ln\lambda+\varepsilon)\}\leq  \bar{r} _{j+1}^{n-n_0} \leq \exp\{- (j+1)q_{n-n_0}(\ln\lambda-\varepsilon)\}.
 \end{equation*}
 Thus in order to prove   case 2, it suffices to show
 \begin{equation}\label{G.Gnew1}
  e^{-(\ln\lambda+\varepsilon)k} \leq ||U(k)||\leq e^{-(\ln\lambda-\varepsilon)k}.
 \end{equation}
 The left inequality holds by (\ref{G.new18}).

 We start to prove the right inequality.
 For any $ y \in [\varepsilon k,k]$ or $ y \in [q_n-k,q_n-\varepsilon k]$,  let $n_0^{\prime}$ be the least positive integer such that $4q_{n-n_0^{\prime}}\leq dist(y,  q_n\mathbb{Z})$, thus $n_0^{\prime}\geq n_0$.
Let $s$ be the
largest positive integer such that $4sq_{n-n_0^{\prime}}\leq dist(y,q_n\mathbb{Z}) $.

Set $I_1, I_2\subset \mathbb{Z}$ as follows
\begin{eqnarray*}
  I_1 &=& [-sq_{n-n_0^{\prime}},sq_{n-n_0^{\prime}}-1], \\
   I_2 &=& [ y-sq_{n-n_0^{\prime}},y+sq_{n-n_0^{\prime}}-1 ],
\end{eqnarray*}
and let $\theta_j=\theta+j\alpha$ for $j\in I_1\cup I_2$. The set $\{\theta_j\}_{j\in I_1\cup I_2}$
consists of $4sq_{n-n_0^{\prime}}$ elements. By case ii) of Theorem \ref{Th.Nonresonant}, $ y$ is $  (\ln\lambda+8\ln (s q_{n-n_0^{\prime}}/q_{n-n_0^{\prime}+1})/q_{n-n_0^{\prime}}-\varepsilon,4sq_{n-n_0^{\prime}}-1)$ regular with $\delta=\frac{1}{4}$.
Notice that
 \begin{equation*}
  (s+1)q_{n-n_0^{\prime}} \geq \varepsilon q_{n-n_0+1}^{t_0}\geq \varepsilon q_{n-n_0^{\prime}+1}^{t_0},
 \end{equation*}
 thus we have
\begin{eqnarray*}
  \ln \lambda+8\ln (s q_{n-n_0^{\prime}}/q_{n-n_0^{\prime}+1})/q_{n-n_0^{\prime}} &\geq&\ln \lambda -8(1-t_0)\beta-\varepsilon \\
   &\geq& \ln\lambda-2\varepsilon.
\end{eqnarray*}
This implies for any
 $ y \in [\varepsilon k,k]$ or $ y \in [q_n-k,q_n-\varepsilon k]$,      there exists an interval $ I(y)=[x_1,x_2]\subset
[ 0,q_n]$ with
 $y\in I(y)$ such that
\begin{eqnarray}\label{G.newnew13}
    dist(y,\partial I(y)) &\geq&  \frac{1}{2} q_{n-n_0^{\prime}}
\end{eqnarray}
and
\begin{equation}\label{G.newnew14}
  |G_{I(y)}(y,x_i)| \leq e^{-(\ln\lambda-\varepsilon )|y-x_i|},\;i=1,2.
\end{equation}
For any $y\in(k,q_n-k)$,
  let $s$ be the
largest positive integer such that $sq_{n-n_0}\leq dist(y,q_n\mathbb{Z}) $ and   set $I_1, I_2\subset \mathbb{Z}$ as follows:
\begin{eqnarray*}
  I_1 &=& [-\lfloor\frac{sq_{n-n_0}}{2}\rfloor,sq_{n-n_0}-\lfloor\frac{sq_{n-n_0}}{2}\rfloor-1], \\
   I_2 &=& [ y-\lfloor\frac{sq_{n-n_0}}{2}\rfloor,y+sq_{n-n_0}-\lfloor\frac{sq_{n-n_0}}{2}\rfloor-1 ].
\end{eqnarray*}
By the same reason,   (\ref{G.newnew13}) and (\ref{G.newnew14}) also hold for $n_0^{\prime}=n_0$.

Arguing exactly as in the proof of Lemma  \ref{Le.resonant}, with
(\ref{G329}) replaced with (\ref{G.newnew13}) and (\ref{G330}) with (\ref{G.newnew14}),
we obtain
\begin{equation}\label{G.newnew15}
    ||U(k)||\leq\max\{\hat{r}_{ 0}  \exp\{-(\ln \lambda- 2\varepsilon)(k-3\varepsilon k)\},\hat{r}_{ 1}^n\exp\{-(\ln \lambda- 2\varepsilon)(q_n-k-3\varepsilon k)\}\},
\end{equation}
where $ \hat{r}_j=\max_{|r|\leq 10\varepsilon  } ||U(jq_n+rk)||$ with $j=0,1$.
 Using that $k\leq \frac{q_n}{2}$,  one has
 \begin{equation*}
   ||U(k)||\leq e^{-(\ln\lambda-\varepsilon)k},
 \end{equation*}
 this implies (\ref{G.Gnew1}) and thus the  theorem.

 \end{proof}
 \begin{remark}\label{lastversion1}
 The assumption that $\theta$  is $n-n_0$ Diophantine with respect to
 $\alpha$  is sufficient for the proof of case 1.
The assumption that $\theta$   satisfies (\ref{DCthetaaddweak})  is
sufficient for the proof of  case ii) of Theorem \ref{Th.Nonresonant},
and therefore for the proof of case 2.
Then
by Remark \ref{lastversion2}, the assumption that $\theta$  is $n-1$ Diophantine with respect to $\alpha$  is
sufficient for the proof of  case 2.
 \end{remark}


 \begin{remark}\label{Redoublesize}
 Suppose we only consider  $q_n^{\frac{8}{9}}\leq k\leq  c q_n$ with $c\leq \frac{1}{2}$ in Theorem \ref{Th.new3}.
  Theorem  \ref{Th.new3} still holds if  we only have (\ref{shn})
  for function $\phi(k)$ on $[-\mathfrak{c}q_n,2  c q_n]$ for some $\mathfrak{c}>0$.
 \end{remark}
 In order to prove (\ref{G.Asymptotics}), it suffices to prove the
 following Theorem, which is a stronger local version of (\ref{G.Asymptotics}).
 \begin{theorem}\label{Maintheoremdecaylocal}
Let $ \alpha\in \mathbb{R}\backslash\mathbb{Q}$ be such that $
|\lambda| >e^{\beta(\alpha)}$.
Suppose $E$ is a generalized eigenvalue of $H_{\lambda,\alpha,\theta}$ and  $\phi$ is
the  generalized eigenfunction. 
Let $ U(k) =\left(\begin{array}{c}
        \phi(k)\\
       \phi({k-1})
     \end{array}\right)
 $.
Then for any $\varepsilon>0,\kappa>0,\nu>1$, there exists $\hat{n}_0$
(depending  on $\alpha,E,\kappa,\nu,\varepsilon$\footnote{The dependence on
  $E$ is through the constant $\hat{C}$ in (\ref{shn}).\label{E}} )such that, if $\theta$ is  $n$-Diophantine with respect to
$\alpha$    with Diophantine constants $\kappa,\nu$ for some $n\geq \hat{n}_0$, we have $U(k)$
satisfy
\begin{equation}\label{G.Asymptoticslocal}
 f(|k|)e^{-\varepsilon|k|} \leq ||U(k)||\leq f(|k|)e^{\varepsilon|k|},
\end{equation}
for $\frac{q_n}{2}\leq|k|\leq\frac{q_{n+1}}{2}$.
\end{theorem}

 \begin{proof}
It remains to collect several already proved statements that cover
different scenarios.

 Case i: $\frac{q_n}{2}  \leq q_{n+1}^{\frac{8}{9}}.$
 \par
 For   $\frac{q_n}{2}\leq k \leq 4q_{n+1}^{\frac{8}{9}}$ 
 the result follows from Theorem \ref{Th.new2}.
 \par
 For $  4q_{n+1}^{\frac{8}{9}}\leq k\leq \frac{q_{n+1}}{2}$,
  (\ref{G.Asymptoticslocal}) follows  from    Theorem \ref{Th.new3}  (notice that  now $k\geq 2q_n$, so $n_0=1$).

Case ii: $ q_{n+1}^{\frac{8}{9}}\leq \frac{q_n}{2}. $
 \par
 Case ii.1: $\frac{q_n}{2}\leq k\leq \min \{q_n,\frac{q_{n+1}}{2}\}$.

 If $q_n=q_{n-1}+q_{n-2}$, then $q_{n-1}\geq \frac{q_n}{2}$. By the proof of case 2 in Theorem \ref{Th.new3} (by Remark \ref{lastversion1}, the assumption that $\theta$ is $n$-Diophantine is enough\color{black}), one has
 for any $q_{n-1}\leq k\leq \min \{q_n,\frac{q_{n+1}}{2}\}$
 \begin{equation*}
   |\phi(k)|\leq e^{-(\ln\lambda-\varepsilon)k}.
 \end{equation*}
 This leads to
 \begin{equation*}
   |\phi(k)|\leq e^{-(\ln\lambda-\varepsilon)k}.
 \end{equation*}
 for $\frac{q_n}{2}\leq k\leq \min \{q_n,\frac{q_{n+1}}{2}\}$.
 This also implies  (\ref{G.Asymptotics}). 

 If $q_n=jq_{n-1}+q_{n-2}$ with $j\geq 2$, we have $ \frac{q_n}{2}\geq q_{n-1}$.  By
 case 2 in Theorem \ref{Th.new3} (by Remark \ref{lastversion1}, the assumption that $\theta$ is $n$-Diophantine is enough\color{black})
  again (with  $n+1-n_0=n-1$), we obtain (\ref{G.Asymptoticslocal}).

 Case ii.2: $q_n\leq k\leq \frac{q_{n+1}}{2}$

 In this case  (\ref{G.Asymptoticslocal}) follows directly from Theorem \ref{Th.new3} (with  $n+1-n_0=n$), because $n_0=1$ so that the fact $\theta$ is
 $n$ Diophantine can guarantee  both cases 1 and 2 of Theorem \ref{Th.new3}\color{black}.

 \end{proof}

\textbf{Proof of Theorem \ref{universal}}
\begin{proof}
The proof follows that of  (\ref{G.Asymptoticslocal}) by shifting by $k_0$ units, and Remark \ref{Redoublesize}.
\end{proof}

  \textbf{Proof of Theorem \ref{Conjecture}}.
  \begin{proof}
  Assume $\theta$ is  Diophantine with respect to $\alpha$.
  First by the definition of $\beta(\alpha),$ 
one has
  for any large $n$ and any $\ell$
  \begin{equation*}
  \bar{r}_{\ell}^n\leq e^{-(\ln\lambda-\beta-\varepsilon)\ell q_n}.
  \end{equation*}
  Combining with the definition of $f(k)$ and  (\ref{G.Asymptotics}), we have
  \begin{equation}\label{G.new15}
    |\phi(k)|\leq e^{-(\ln\lambda-\beta-\varepsilon)k}.
  \end{equation}
  We therefore established that every generalized eigenfunction decays
  exponentially, which by Schnol's \color{black}Theorem \cite{berezanskii1968expansions} implies  the localization statement.
  \par
  By the definition of $\beta(\alpha)$ again, there exists a subsequence ${q}_{n_k}$
of $q_n$ such that
\begin{equation}\label{Seq}
    {q}_{n_k+1}\geq e^{(\beta-\varepsilon){q}_{n_k}}.
\end{equation}

By Theorem \ref{Maintheoremdecay} (or \ref{Th.new1}) and the definition of $\bar{r}_j^n$ we have
for any $k>0$
  \begin{equation}\label{G.new16}
   | |U({q}_{n_k})||\geq e^{-(\ln\lambda-\beta+\varepsilon){q}_{n_k}}.
  \end{equation}
  Thus (\ref{G.new15}) and (\ref{G.new16}) imply that the decay rate is just $\ln\lambda-\beta$.
  \end{proof}
\section{ Hierarchical structure}\label{hier}
As we have already established Theorem \ref{Conjecture} we know that
each generalized eigenfunction decays exponentially so has a global maximum. Assume its global
maximum (see Footnote \ref{globalmax}) is at $0$ and $\phi$  is normalized by $\|\phi\|_{\infty}=1.$ Note that then $\hat{C}$ in
(\ref{shn}) is equal to $1$ so all dependence of the largeness of $n$ on
$E$ (see Footnote \ref{E})  disappears. Theorem \ref{Maintheoremdecay}
provides, for sufficiently large $n$ plenty of local $n$-maxima in the
vicinity of $aq_n$, but determined
only with $\varepsilon q_n$ precision.  In the
next theorem we show that this precision can be improved all the way
to an $n$-independent
constant.  We have


\begin{theorem}\label{localmaximalth1}
Fix $ \kappa,\nu,\epsilon$. Then for sufficiently small $\varepsilon$
there exists $\hat{n}_0 (\kappa,\nu,\lambda,\alpha,\epsilon,\varepsilon)$ such that
 if $\theta$ is $k$-Diophantine for all $\hat{n}_0\leq k \leq n$ with
 Diophantine constants $\kappa,\nu$ and  $\frac{\ln q_{n+1}-\ln j}{q_n}> \epsilon \ln \lambda$ with $\epsilon>0$, then

\begin{equation}\label{Intervalklocal}
   \sup_{k\in [jq_n-\epsilon q_n +\varepsilon q_n,jq_n]} || U(k)||=\sup_{k\in [jq_n-K_0,jq_n]} || U(k)||,
\end{equation}
and
\begin{equation}\label{Intervalklocal1}
   \sup_{k\in [jq_n,jq_n+\epsilon q_n -\varepsilon q_n]} || U(k)||=\sup_{k\in [jq_n,jq_n+K_0]} || U(k)||,
\end{equation}
where  $K_0=q_{\hat{n}_0+1}$.
\end{theorem}
\begin{proof}
 We first give the proof of (\ref{Intervalklocal}).

Let $k_0\in [jq_n-\varepsilon q_n,jq_n]$ be such that
\begin{equation*}
  ||U(k_0)|| = \sup_{k\in [jq_n-\varepsilon q_n,jq_n]} || U(k)||.
\end{equation*}

By (\ref{r_jnew}), (\ref{G.new8}),(\ref{G.newnewnew1add}) and (\ref{G.newnewnew2add}), one has
\begin{equation*}
  ||U(k_0)|| = \sup_{k\in [jq_n-\epsilon q_n+\varepsilon q_n,jq_n]} || U(k)||.
\end{equation*}

Suppose (\ref{Intervalklocal}) does not hold, i.e., $k_0\in [jq_n-\varepsilon q_n,jq_n-K_0]$.

Now we will reflect the elements in $[jq_n-\varepsilon q_n, jq_n]$ at $\frac{j}{2}q_n$.  That is for any element $k\in[jq_n-\varepsilon q_n, jq_n]$,
let $  k^{\prime}= jq_n-k$. Then $k^{\prime}\in[0,\varepsilon q_n]$.

Choose $n^\prime$ such that $b_{n^\prime}\leq  k_0^{\prime}<b_{n^\prime+1}$ (where $k_0^\prime=jq_n-k_0$). Then $n^\prime\geq \hat{n}_0$.

Case 1. $k_0^{\prime}$ is $n^\prime$-nonresonant, i.e., $\text{dist} (k_0^{\prime},q_{n^\prime}\mathbb{Z})\geq b_{n^\prime}$.

  Let $n_0^\prime$ be the least positive integer such that $4q_{n^\prime-n^\prime_0}\leq dist(k_0^\prime,  q_{n^\prime}\mathbb{Z})$.
 Let $s$ be the
largest positive integer such that $4sq_{n^\prime-n^\prime_0}\leq dist(k_0^\prime,  q_{n^\prime}\mathbb{Z})$.
Set $I_1, I_2, I^\prime_2\subset \mathbb{Z}$ as follows
\begin{eqnarray*}
  I_1 &=& [-sq_{n^\prime-n^\prime_0},sq_{n^\prime-n^\prime_0}-1], \\
   I_2 &=& [ k_0-sq_{n^\prime-n^\prime_0}+1,k_0+sq_{n^\prime-n^\prime_0} ],\\
   I^\prime_2 &=& [ k_0^\prime-sq_{n^\prime-n^\prime_0},k_0^\prime+sq_{n^\prime-n^\prime_0}-1 ],
\end{eqnarray*}

%
Notice that $I^\prime_2$ and $I_2$ are reflections of each other about $\frac{j}{2}q_n$.

By the Diophantine condition on $\theta$ with respect to $\alpha$, for any $k_2\in I_2$($k_2^{\prime}\in I_2^\prime$) and $k_1\in I_1$,
we have
\begin{eqnarray}
  \nonumber||2\theta+ (k_1+k_2)\alpha||_{\mathbb{R}/\mathbb{Z}}  &=& ||2\theta- k_2^{\prime}\alpha+j q_n\alpha+k_1\alpha||_{\mathbb{R}/\mathbb{Z}}   \\
\nonumber  &\geq & ||2\theta+(k_1- k_2^{\prime})\alpha||_{\mathbb{R}/\mathbb{Z}}- ||j q_n\alpha||_{\mathbb{R}/\mathbb{Z}}  \\
  \nonumber  &\geq&  ||2\theta+(k_1- k_2^{\prime})\alpha||_{\mathbb{R}/\mathbb{Z}} -\frac{j}{q_{n+1}}\\
  \nonumber &\geq&    ||2\theta+(k_1- k_2^{\prime})\alpha||_{\mathbb{R}/\mathbb{Z}}-e^{-\epsilon \ln\lambda q_n} \\
     &\geq&  \frac{1}{2}  ||2\theta+(k_1- k_2^{\prime})\alpha||_{\mathbb{R}/\mathbb{Z}}\geq \frac{C}{q_{n^\prime}^C},\label{smallqn3}
\end{eqnarray}
where the last inequality holds by the fact $|k_1|, | k_2^{\prime}|\leq C_{\star} b_{n^\prime+1}$ so that we can apply Lemma \ref{Lanaaddsmallest}.

For  any $k_2\in I_2$ and $k_1\in I_1$ ($k_1+k_2^{\prime} \neq 0$ by the construction of $I_1,I_2$), we also have
\begin{eqnarray}
 \nonumber ||  (k_2-k_1)\alpha||_{\mathbb{R}/\mathbb{Z}}  &=& ||  -k_2^{\prime}\alpha+jq_n\alpha-k_1||_{\mathbb{R}/\mathbb{Z}}   \\
 \nonumber &\geq & ||  (-k_1-k_2^{\prime})\alpha||_{\mathbb{R}/\mathbb{Z}}- ||j q_n\alpha||_{\mathbb{R}/\mathbb{Z}}  \\
 \nonumber    &\geq&   ||  (k_1+k_2^{\prime})\alpha||_{\mathbb{R}/\mathbb{Z}} -e^{-\epsilon\ln \lambda q_n}\label{smallqn2}\\
     &\geq& \frac{1}{2}  ||  (k_1+k_2^{\prime})\alpha||_{\mathbb{R}/\mathbb{Z}} \geq  \frac{C}{q_{n^\prime}^C} ,\label{smallqn4}
\end{eqnarray}
where   the last inequality holds by the fact $|k_1|, | k_2^{\prime}|\leq C_{\star} b_{n^\prime+1}$ and $k_1-k_2^{\prime}\neq q_{n^\prime}\mathbb{Z}$ so that we can apply Lemma \ref{Lanaaddsmallest2}.

By   Theorem    \ref{new},   (\ref{smallqn3}) and (\ref{smallqn4}), we have $k_0$ is  $(\hat{k}_0,\ln\lambda-\beta-\varepsilon)$  regular, where $\hat{k}_0 =4sq_{n^\prime-n_0^\prime}-1$. Let $I_2=[x_1,x_2] \subset[jq_n-2\varepsilon q_n,jq_n]$.

By (\ref{Block}), we have
\begin{equation*}
    |\phi(k_0)|\leq e^{-(\ln\lambda-\beta-\varepsilon)\frac{\hat{k}_0}{10} } (|\phi(x_1)|+|\phi(x_0)|)\leq e^{-(\ln\lambda-\beta-\varepsilon)\frac{\hat{k}_0}{10}  } ||U(k_0)||.
\end{equation*}
Similarly,
\begin{equation*}
    |\phi(k_0-1)|\leq e^{-(\ln\lambda-\beta-\varepsilon)\frac{\hat{k}_0}{10} } ||U(k_0)||.
\end{equation*}
The last two inequalities imply that
\begin{equation*}
  ||U(k_0)||  \leq e^{-(\ln\lambda-\beta-\varepsilon)\frac{\hat{k}_0}{10}  }||U(k_0)||.
\end{equation*}
This is impossible.

Case 2. $k_0^\prime$ is $n$-resonant, i.e.,  $|k_0^\prime-\ell q_{n^\prime}|\leq b_{n^\prime}$ for some $\ell$.

From (\ref{smallqn3}) and (\ref{smallqn4}), we know that the small divisor condition does not change under reflection at $\frac{j}{2}q_n$.
Following the proof of (\ref{G.Secondadd2}) and replacing Lemma \ref{Le.resonant} with a combination of  Lemma \ref{Le.resonant20} and Theorem \ref{new}, we have
\begin{equation*}
   \texttt{r}^{n^\prime,\phi}_{\ell} \leq  \exp\{-(\ln \lambda - \beta-\varepsilon)q_{n^\prime}\} \max\{\texttt{r}^{n^\prime,\phi}_{\ell\pm1}   \},
\end{equation*}
where
\begin{equation*}
\texttt{r}_{\ell}^{n^\prime,\phi}= \sup_{|r|\leq 10 \varepsilon }|\phi(j q_n-(\ell q_{n^\prime}+rq_{n^\prime}))|.
\end{equation*}
 This is contradicted to the fact that $k_0$  is the  maximal point because $|k_0^\prime-\ell q_{n^\prime}|\leq b_{n^\prime}$.

This completes the proof of (\ref{Intervalklocal}).

Now we turn to the proof of (\ref{Intervalklocal1}).
Let $k_0^r\in [jq_n,jq_n+\varepsilon q_n]$ be such that
\begin{equation*}
  ||U(k^{r}_0)|| = \sup_{k\in [jq_n,jq_n+\varepsilon q_n]} || U(k)||.
\end{equation*}

Suppose the Theorem does not hold, i.e., $k^r_0\in [jq_n+K_0,jq_n+\varepsilon q_n]$.

In this case we shift the elements in $[jq_n, jq_n+\varepsilon q_n]$ by $-jq_n$. That is for any element $k\in[jq_n,jq_n+\varepsilon q_n]$,
let $  k^{r,\prime}= k-jq_n$. Then $k^{r,\prime}\in[0,\varepsilon q_n]$.
Then  (\ref{Intervalklocal1}) holds by the same proof, only replacing all $k^{\prime} $ with $  k^{r,\prime}.$

\end{proof}

We restate the result of Theorem \ref{localmaximalth1} as a more convenient

\begin{theorem}\label{localmaximalth}
Fix $ \kappa,\nu,\epsilon$.
Then for sufficiently small $\varepsilon$
there exists $\hat{n}_0 (\kappa,\nu,\lambda,\alpha,\epsilon,\varepsilon)$ such that 
if $k_0$ is a local $(n+1)$-maximum, $\theta$ is $k$-Diophantine for all $\hat{n}_0\leq k\leq n$ with Diophantine constants $\kappa,\nu$,
and
\begin{equation*}
\frac{\ln q_{n+1}-\ln j}{q_n}> \epsilon \ln \lambda.
\end{equation*}

Then
\begin{equation}\label{Intervalklocal3}
   \sup_{k\in [k_0+jq_n-\epsilon q_n +\varepsilon q_n,k_0+(j+1)q_n]} || U(k)||=\sup_{k\in [k_0+jq_n-K_0,k_0+jq_n+K_0]} || U(k)||,
\end{equation}
 where $K_0=q_{\hat{n}_0+1}$.
\end{theorem}
\begin{proof}
By shifting the operator by $k_0$ units, we can assume $k_0=0$.
Theorem \ref{localmaximalth1} still holds if  $0$ is a local $(n+1)$-maximum by Remark
\ref{Redoublesize}.

\end{proof}

We will now formulate a local version of the
hierarchical structure Theorem \ref{addnewth}.

  Fix  $0<\varsigma,\epsilon$ with $\varsigma+2\epsilon<1.$ Let $n_j\to\infty$ be such that $\ln q_{n_j+1}\geq
(\varsigma+2\epsilon) \ln |\lambda|q_{n_j} .$ Let $\mathfrak{c}_{j}=(\ln q_{n_{j}+1}-\ln
|a_{n_{j}}|)/ \ln |\lambda|q_{n_{j}}-\epsilon.$ $\mathfrak{c}_{j}>\epsilon$ for $0<a_{n_j} <
e^{\varsigma \ln |\lambda| q_{n_j}}$.

\begin{theorem}\label{addnewthlocal}

Suppose  $k_0$ is a local $(n_{j_0}+1)$-maximum.
Suppose $\theta+k_0\alpha$ is 
Diophantine with respect to $\alpha$ (with Diophantine constants $\kappa,\nu$).
 Then there exists   $\hat{n}_0 (\alpha,\lambda,\kappa,\nu,\epsilon)<\infty$  such that for any  $j_0>j_1>\cdots>j_k, $ $n_{j_k}\geq \hat{n}_0+k$, and
  $0<a_{n_{j_i}} <
e^{\varsigma \ln |\lambda| q_{n_{j_i}}}, i=0,1,\ldots,k$ for all $0\leq
s\leq k$ there exists a
local $n_{j_s}$-maximum $b_{a_{n_{j_0}},a_{n_{j_1}},...,a_{n_{j_s}}}$ on the
interval  $b_{a_{n_{j_0}},a_{n_{j_1}},...,a_{n_{j_s}}}+I^{n_{j_s}}_{\mathfrak{c_{j_s}},1}$ for all $0\leq s\leq k$
such that the
following holds:
\begin{description}
  \item[I]   $|b_{a_{n_{j_0}}}-(k_0 +a_{n_{j_0}}q_{n_{j_0}})|\leq q_{\hat{n}_0+1},$
  \item[II] For any $1\leq s\leq k ,$
$|b_{a_{n_{j_0}},a_{n_{j_1}},...,a_{n_{j_s}}}-(b_{a_{n_{j_0}},a_{n_{j_1}},...,a_{n_{j_{s-1}}}} +a_{n_{j_s}}q_{n_{j_s}})|\leq   q_{\hat{n}_0+s+1}$.
  \item[III] if $2(x-b_{a_{n_{j_0}},a_{n_{j_1}},...,a_{n_{j_k}}})\in
    I^{{n_{j_k}}}_{\mathfrak{c}_{j_k},1}$,  then  for each
    $s=0,1,...,k,$

\begin{equation}\label{G.add1local}
 f(x_s)e^{-\varepsilon|x_s|} \leq \frac{||U(x)||}{||U(b_{a_{n_{j_0}},a_{n_{j_1}},...,a_{n_{j_s}}})||}\leq f(x_s)e^{\varepsilon|x_s|},
\end{equation}
where $x_s=|x-b_{a_{n_{j_0}},a_{n_{j_1}},...,a_{n_{j_s}}}|$ is large
enough.
\end{description}
Moreover, {\it every} local $n_{j_s}$-maximum on the interval
$$b_{a_{n_{j_0}},a_{n_{j_1}},...,a_{n_{j_{s-1}}}}+[ -e^{\epsilon\ln\lambda
  q_{n_{j_s}}}, e^{\epsilon\ln\lambda
  q_{n_{j_s}}}]$$ is of the form
$b_{a_{n_{j_0}},a_{n_{j_1}},...,a_{n_{j_s}}}$ for some $a_{n_{j_s}}.$

\end{theorem}
\begin{proof}
 Let $\hat{n}_0= \hat{n}_0(\kappa,3\nu,\lambda,\alpha,\epsilon,\epsilon/10)$ be given by
 Theorem \ref{localmaximalth} .\footnote{$3\nu$ here can be easily
   relaxed to $(1+\epsilon)\nu.$\label{3}}

As long as \begin{equation}\label{ln}(\ln q_{n+1}-\ln |a_{n}|)/q_n\geq  2\epsilon\ln|\lambda| \end{equation} with $0<\varsigma,\epsilon< 1$,
where $ 0 <|a_{n}|\leq \frac{q_{n+1}}{2q_n}$,   Theorem
\ref{localmaximalth} (upon shifting by $k_0$ units) implies  that there exists a
 local ${n}$-maximum  $ b_{a_n}$ on interval $b_{a_n}
 +I^n_{\epsilon,1}$ such that
 \begin{equation}\label{maxb}
    |b_{a_n}-(a_nq_n+k_0)|\leq  K_0=q_{\hat{n}_0+1}.
 \end{equation}
Now let $n_i$ be such that $\ln q_{n_i+1}\geq
(\varsigma+2\epsilon) \ln |\lambda|q_{n_i} ,i=t_0,t_0+1,\cdots, j$ for some $0<\varsigma,\epsilon<1.$

By  (\ref{maxb}), one has that there exists a local $n_{j_0}$  maximum $b_{a_{n_{j_0}}} =a_{n_{j_0}}  q_{n_{j_0}}+k_0 +\hat{K}_{n_{j_0}}$ with $|\hat{K}_{n_{j_0}}|\leq K_0 $.



Now we will  prove that 
for $0\leq s\leq k$,  there exists
$$|b_{a_{n_{j_0}},a_{n_{j_1}},...,a_{n_{j_s}}}-a_{n_{j_s}}q_{n_{j_s}}-b_{a_{n_{j_0}},a_{n_{j_1}},...,a_{n_{j_{s-1}}}}|\leq K_{s}=q_{\hat{n}_0+s+1}.$$

Notice that $\sum_{i=0}^kK_i\leq 4K_{k}$.
We will now  prove Theorem \ref{addnewth} by  induction on $s$.

By the assumption, one has
  \begin{equation}\label{shiftDCthetaprime}
   ||2\theta+2k_0+k\alpha||_{\mathbb{R}/\mathbb{Z}} > \frac{\kappa}{|k|^{\nu}},
 \end{equation}
 for any $k\in \mathbb{Z} \backslash \{0\}$.

First we prove the case $s=1$.
 By  the Diophantine condition on $\theta$ (\ref{shiftDCthetaprime}),  we have for $  |\ell| \leq  q_{n_{j_1}+1}$,
 the following holds
    \begin{eqnarray}
      \nonumber ||2\theta+(2b_{a_{n_{j_0}}}+\ell)\alpha||_{\mathbb{R}/\mathbb{Z}} &\geq &  ||2\theta+ (2k_0+\ell+2K_{n_{j_0}})\alpha||_{\mathbb{R}/\mathbb{Z}} - ||2a_{n_{j_0}} q_{n_{j_0}}\alpha||_{\mathbb{R}/\mathbb{Z}}\\
     \nonumber  &\geq & \frac{\kappa}{(  2K_0+|\ell|)^{\nu}} -\frac{2a_{n_{j_0}}}{q_{n_{j_0}+1}}\\
     \nonumber  &\geq & \frac{\kappa}{|\max\{K_0,\ell\}|^{2\nu}} -e^{-\epsilon \ln \lambda q_{n_{j_0}}}\\
      &\geq &
      \frac{\kappa}{|\max\{K_0,\ell\}|^{3\nu}}. \label{Gnew201}
    \end{eqnarray}


Therefore $\theta+b_{a_{n_{j_0}}}\alpha$ is
${\hat{n}_0+1}$-Diophantine with respect to $\alpha$ with parameters
$3\nu,\kappa,$  and by Theorem \ref{localmaximalth} again,  there exists a local  $n_{j_1}$-maximum such that
$b_{a_{n_{j_0}},a_{n_{j_1}}} =a_{n_{j_1}}  q_{n_{j_1}}+b_{a_{n_{j_0}}} +\hat{K}_{n_{j_1}}$ with $|\hat{K}_{n_{j_1}}|\leq K_1= q_{\hat{n}_0+2} $. This completes the first step.

Assume Theorem holds for $s=k-1$. It suffices to show it holds for $s=k$.
 By  the Diophantine condition on $\theta$ (\ref{shiftDCthetaprime}) again,  we have for
 $  |\ell| \leq  q_{n_{j_k}+1}$, the following holds,
    \begin{eqnarray*}
       ||2\theta+(2b_{a_{n_{j_0}},a_{n_{j_1},\cdots,a_{n_{j_{k-1}}}}}+\ell)\alpha||_{\mathbb{R}/\mathbb{Z}} &\geq &  ||2\theta+ (2k_0+\ell+2\sum_{s=0}^{k-1}K_{s})\alpha||_{\mathbb{R}/\mathbb{Z}} -\sum_{s=0}^{k-1} ||2a_{n_{j_s}} q_{n_{j_s}}\alpha||_{\mathbb{R}/\mathbb{Z}}\\
       &\geq & \frac{\kappa}{(  8K_{k-1}+|\ell|)^{\nu}} -\sum_{s=0}^{k-1} ||2a_{n_{j_s}} q_{n_{j_s}}\alpha||_{\mathbb{R}/\mathbb{Z}}\\
      &\geq &  \frac{\kappa}{{|\max\{K_{k-1},\ell\}}^{2\nu}}-\sum_{s=0}^{k-1} e^{-\epsilon q_{n_{j_s}}}\\
      &\geq &  \frac{\kappa}{{|\max\{K_{k-1},\ell\}|}^{3\nu}}.
    \end{eqnarray*}

Thus $\theta+b_{a_{n_{j_0}},a_{n_{j_1},\cdots,a_{n_{j_{k-1}}}}}\alpha$ is
${\hat{n}_0+k}$-Diophantine 
 with respect to $\alpha$ with parameters
$3\nu,\kappa,$ and by   Theorem \ref{localmaximalth} again,
there exists a local  $n_{j-k}$-maximum such that
$b_{a_{n_{j_0}},a_{n_{j_1}},\cdots,a_{n_{j_k}}} =a_{n_{j_k}}  q_{n_{j_k}}+b_{a_{n_{j_0}},a_{n_{j_1}},\cdots,a_{n_{j_{k-1}}}} +\hat{K}_{n_{j_k}}$ with $|\hat{K}_{n_{j_k}}|\leq K_{k} =q_{\hat{n}_0+k+1}$.
This implies   II holds for $s=k$.
Thus we complete the proof of I and II.

III, as well as the moreover part, follow from Theorem \ref{universal} directly.
\end{proof}
\textbf{Proof of Theorem \ref{addnewth}}
\begin{proof}
Since $k_0$ is a local $n_{j_0}+1$-maximum 
 for every $j,$
Theorem \ref{addnewth}  follows from Theorem \ref{addnewthlocal} directly.
\end{proof}

Theorem \ref{addnewthlocal} describes a hierarchical structure around
every local $(n_{j_0}+1)$-maximum.

We will say that a local $n_{j_0}$-maximum is
{\it $k$-hierarchical} if there exists $\epsilon >0,$  $j_0>j_{1}>\cdots>j_k$
with $n_{j_i+1}>e^{\epsilon n_{j_i}}$ and, for each
$s=0,1,\ldots,k,$ a  collection of local $n_{j_s}$-maxima,
$\{b_{a_{n_{j_0}},a_{n_{j_1}},\ldots,a_{n_{j_s}}}\}$ such that
\begin{description}
  \item[I] All local $(n_{j_s},\epsilon)$-maxima in
    $[b_{a_{n_{j_0}},a_{n_{j_1},\ldots,a_{n_{j_{s}}}} -e^{\epsilon q_
      {n_{j_{s}}}}, b_{a_{n_{j_0}},a_{n_{j_1}},\cdots,a_{n_{j_{s-1}}}} +e^{\epsilon q_{n_{j_{s}}}}}]$
    are given by $\{b_{a_{n_{j_0}},a_{n_{j_1}},\cdots,a_{n_{j_{s-1}}},a_{n_{j_s}}}\} $ with
    all possible choices of $a_{n_{j_s}}$.

  \item[II] if $2(x-b_{a_{n_{j_0}},a_{n_{j_1}},...,a_{n_{j_k}}})\in
    I^{{n_{j_k}}}_{\epsilon,\epsilon}$,  then  for each
    $s=0,1,...,k,$

\begin{equation}\label{G.add1local}
 f(x_s)e^{-\varepsilon|x_s|} \leq \frac{||U(x)||}{||U(b_{a_{n_{j_0}},a_{n_{j_1}},...,a_{n_{j_s}}})||}\leq f(x_s)e^{\varepsilon|x_s|},
\end{equation}
where $x_s=|x-b_{a_{n_{j_0}},a_{n_{j_1}},...,a_{n_{j_s}}}|$ is large
enough.
\end{description}

{\bf Proof of Corollary \ref{corhie}.} Clearly,
$b_{a_{n_1},\ldots,a_{n_s}}$ of Theorem \ref{addnewthlocal} form the
collection required for the definition of $k$-hierarchy, so it remains
to estimate the number of levels of the hierarchy, that is find
$k$ such that $n_{j_k}\geq \hat{n}_0+k.$  Clearly, $k=j/2-\lfloor\hat{n}_0/2\rfloor$ works.\qed

\section{Growth of transfer matrices. Proof of (\ref{G.Asymptoticstransfer}) }\label{part2}
 Assume $\theta$ is Diophantine with respect to $\alpha$ in this and
 the following section.
 \begin{theorem}\label{Th.new2transfer}
 Let $A(j)=||A_{jq_n}||$.
Assume $j q_n\leq k<(j+1)q_n$ with $0\leq j \leq 48C_{\star}\frac{b_{n+1}}{q_n}$, $b_{n+1}\geq \frac{q_n}{2}$ and $k\geq \frac{q_n}{4}$.
We have
\begin{equation}\label{G.newnewnew1transfernew}
  ||A_k||\leq \max\{e^{-|k-j q_n|\ln\lambda}A(j),e^{-|k-(j +1)q_n|\ln\lambda}A(j+1)\}e^{\varepsilon k},
\end{equation}
\begin{equation}\label{G.newnewnew2transfernew}
  ||A_k||\geq \max\{e^{-|k-j q_n|\ln\lambda}A(j),e^{-|k-(j +1)q_n|\ln\lambda}A(j+1)\}e^{-\varepsilon k}.
\end{equation}
\end{theorem}
\begin{proof}

Let  $\tilde{U} (k)=\left(\begin{array}{c}
        \psi(k)\\
       \psi(k-1)
     \end{array}\right)
 $.
By Last-Simon's arguments ((8.6) in \cite{last1999eigenfunctions}), one has
\begin{equation}\label{G.last}
||A_k||\geq  ||A_k\tilde{U} (0)||\geq c ||A_k|| .
\end{equation}
Then (\ref{G.newnewnew1transfernew}) holds by (\ref{G.last}),(\ref{G.new18}) and   (\ref{Intervalk}).

(\ref{G.newnewnew2transfernew}) holds directly by (\ref{G24}).
\end{proof}

\begin{theorem}\label{Th6.new2transfer}
Assume $1\leq j \leq 8\frac{b_{n+1}}{q_n}$ and $b_{n+1}\geq \frac{q_n}{2}$.
Then
\begin{equation}\label{G.transfer1}
\frac{q_{n+1}}{\bar{r}_{j}^n}e^{-\varepsilon jq_n}\leq   A(j)\leq \frac{q_{n+1}}{\bar{r}_{j}^n}e^{\varepsilon jq_n}.
\end{equation}
\end{theorem}
\begin{proof}
We first show the left inequality. Clearly
\begin{equation}\label{G.transfer2}
||A_k||\geq ||U(k)||^{-1},
\end{equation}
thus by (\ref{G.newnewnew1}) and (\ref{G.newnewnew1transfernew}),
we must have for any $jq_n\leq k< (j+1)q_n$ with $j\geq 0$ and $k\geq \frac{q_n}{4}$,
\begin{equation}\label{G.transfer3}
  \max\{e^{-|k-j q_n|\ln\lambda}A(j),e^{-|k-(j +1)q_n|\ln\lambda}A(j+1)\}e^{\varepsilon k}\geq(\max\{e^{-|k-j q_n|\ln\lambda}\bar{r}_{j}^n,e^{-|k-(j +1)q_n|\ln\lambda}\bar{r}_{j+1}^n\})^{-1}e^{-\varepsilon k}.
\end{equation}
Let
\begin{equation*}
  k_0=(j+1)q_n-\frac{\ln q_{n+1}-\ln (j+1)}{2\ln\lambda}.
\end{equation*}
One has $k_0\geq \frac{q_n}{4}$,
thus
\begin{equation*}
\max\{e^{-|k_0-j q_n|\ln\lambda}\bar{r}_{j}^n,e^{-|k_0-(j +1)q_n|\ln\lambda}\bar{r}_{j+1}^n\}\leq \bar{r}_{j+1}^n(\frac{j+1}{q_{n+1}})^{\frac{1}{2}} e^{\varepsilon k_0}.
\end{equation*}
Combining with (\ref{G.transfer3}), we have
\begin{equation}\label{G.transfer4}
  \max\{e^{-|k_0-j q_n|\ln\lambda}A(j),e^{-|k_0-(j +1)q_n|\ln\lambda}A(j+1)\}\geq \frac{q_{n+1}^{\frac{1}{2}}}{(j+1)^{\frac{1}{2}}}(\bar{r}_{j+1}^n)^{-1}e^{-\varepsilon k_0}.
\end{equation}
This implies that either
\begin{equation}\label{G.transfer5}
A(j)\geq e^{\ln\lambda q_n}(\bar{r}_{j+1}^n)^{-1}e^{-\varepsilon k_0},
\end{equation}
or
\begin{equation}\label{G.transfer6}
A(j+1)\geq \frac{q_{n+1}}{j+1}(\bar{r}_{j+1}^n)^{-1}e^{-\varepsilon k_0}.
\end{equation}
Notice that by (\ref{r_jnewpsi}) and (\ref{G.last}), we have
\begin{equation}\label{G.transfer7}
 A(j+1)\geq A(j) e^{(\ln\lambda-\varepsilon )q_n} \frac{j+1}{q_{n+1}}.
\end{equation}
By (\ref{G.transfer5}),(\ref{G.transfer6}) and (\ref{G.transfer7}), we obtain the left inequality of
(\ref{G.transfer1}).

Now we turn to the proof of the right inequality of
(\ref{G.transfer1}).
By (8.5) and (8.7) in \cite{last1999eigenfunctions}
we have
\begin{equation}\label{G.transfer8}
   ||A_kU(0)||^2\leq ||A_k||^2m(k)^2+||A_k||^{-2},
\end{equation}
where
\begin{equation}\label{G.transfer9}
  m(k)\leq C\sum_{p=k}^{\infty}\frac{1}{||A_p||^2}.
\end{equation}
If  $ k\geq C_{\star}j q_n $ with $j\geq 1$,
 by (\ref{G.new15}) we have
\begin{eqnarray*}
  ||A_k|| &\geq & ||U(k)|| ^{-1}\\
   &\geq & e^{(\ln\lambda-\beta-\varepsilon)k}
\end{eqnarray*}
and by   (\ref{G24}) we have
\begin{equation*}
  A(j)\leq e^{(\ln\lambda+\varepsilon)jq_n}.
\end{equation*}
This implies
\begin{equation}\label{G.lastsimon2}
  ||A_k||\geq A(j)e^{\frac{\ln\lambda-\beta}{2}k}.
\end{equation}

If $j q_n\leq k\leq C_{\star}j q_n$ with $j\geq 1$,
let $j_0 q_n\leq k< (j_0+1)q_n$  with $j \leq j_0\leq C_{\star}j$.
By (\ref{G.newnewnew2transfernew}) and (\ref{G.transfer7}), we have
\begin{eqnarray}
  ||A_k|| &\geq & A(j_0)\max\{e^{-|k-j_0 q_n|\ln\lambda},e^{-|k-(j_0 +1)q_n|\ln\lambda}e^{q_n\ln\lambda } \frac{j_0+1}{q_{n+1}}\}e^{-\varepsilon j_0 q_n} \nonumber \\
   &\geq &  (\frac{j_0+1}{q_{n+1}})^{\frac{1}{2}} A(j_0)e^{-\varepsilon j_0q_n}\nonumber \\
   &\geq & (\frac{j+1}{q_{n+1}})^{\frac{1}{2}} A(j)e^{-\varepsilon jq_n} .\label{G.lastsimon3}
\end{eqnarray}
Thus by (\ref{G.lastsimon2}) and (\ref{G.lastsimon3}), we have
\begin{equation}\label{G.transfer10}
  m(jq_n)\leq \frac{q_{n+1}}{jA(j)^2} e^{\varepsilon jq_n}.
\end{equation}
Let $k=jq_n$ in (\ref{G.transfer8}). One  has
\begin{equation*}
  \tilde{r}_j^2\leq \frac{q_{n+1}^2}{j^2A(j)^2}e^{\varepsilon jq_n}.
\end{equation*}
Thus by (\ref{Th.new1}), we obtain
\begin{equation}\label{G.transfer11}
 A(j)\leq \frac{q_{n+1}}{j\bar{r}_{j}^n}e^{\varepsilon jq_n}.
\end{equation}
This implies the   right inequality of
(\ref{G.transfer1}).
\end{proof}
Theorems \ref{Th.new2transfer} and \ref{Th6.new2transfer} imply the following theorem directly.
\begin{theorem}\label{Th.new2transfernew}
Assume $j q_n\leq k<(j+1)q_n$ with $0\leq j \leq 6\frac{b_{n+1}}{q_n}$, $b_{n+1}\geq \frac{q_n}{2}$.
We have,
for $k\geq q_{n}$,
\begin{equation}\label{G.newnewnew1transfer}
  ||A_k||\leq \max\{e^{-|k-j q_n|\ln\lambda}\frac{q_{n+1}}{\bar{r}_{j}^n},e^{-|k-(j +1)q_n|\ln\lambda}\frac{q_{n+1}}{\bar{r}_{j+1}^n}\}e^{\varepsilon|k|},
\end{equation}
and
\begin{equation}\label{G.newnewnew2transfer}
  ||A_k||\geq \max\{e^{-|k-j q_n|\ln\lambda}\frac{q_{n+1}}{j\bar{r}_{j}^n},e^{-|k-(j +1)q_n|\ln\lambda}\frac{q_{n+1}}{(j+1)\bar{r}_{j+1}^n}\}e^{-\varepsilon|k|}.
\end{equation}
and for $\frac{q_n}{4}\leq k<q_{n}$,
\begin{equation}\label{G.newnewnew1transfer0}
  ||A_k||\leq \max\{e^{-|k|\ln\lambda},e^{-|k-q_n|\ln\lambda}\frac{q_{n+1}}{\bar{r}_{1}^n}\}e^{\varepsilon|k|},
\end{equation}
and
\begin{equation}\label{G.newnewnew2transfer0}
  ||A_k||\geq \max\{e^{-|k|\ln\lambda},e^{-|k-q_n|\ln\lambda}\frac{q_{n+1}}{\bar{r}_{1}^n}\}e^{-\varepsilon|k|}.
\end{equation}
\end{theorem}
\begin{theorem}\label{Th.new3transfer}
    For any $q_n^{\frac{8}{9}}\leq k\leq \frac{q_n}{2}$, let $n_0$ be the smallest positive   integer such that
$q_{n-n_0}\leq k< q_{n-n_0+1}$.
Suppose $j q_{n-n_0}\leq k< (j+1)q_{n-n_0+1}$ with $j\geq 1$,
 then the following holds,
\begin{equation}\label{G.newnewnew4transfer}
||A_k||\leq \max\{e^{-|k-j q_{n-n_0}|\ln\lambda}\frac{q_{n-n_0+1}}{\bar{r}_{j}^{n-n_0}},e^{-|k-(j+1)q_{n-n_0}|\ln\lambda}\frac{q_{n-n_0+1}}{\bar{r}_{j+1}^{n-n_0}}\}e^{\varepsilon|k|},
\end{equation}
and
\begin{equation}\label{G.newnewnew5transfer}
 ||A_k||\geq  \max\{e^{-|k-j q_{n-n_0}|\ln\lambda}\frac{q_{n-n_0+1}}{\bar{r}_{j}^{n-n_0}},e^{-|k-(j+1)q_{n-n_0}|\ln\lambda}\frac{q_{n-n_0+1}}{\bar{r}_{j+1}^{n-n_0}}\}e^{-\varepsilon|k|}.
\end{equation}
 \end{theorem}
 \begin{proof}
 As in the  proof of Theorem \ref{Th.new3}, we  split into the same two cases: 1 and 2.  Case 1 can be done directly by Theorem \ref{Th.new2transfernew}.
For case 2, as in the proof of case 2 of Theorem \ref{Th.new3},  it suffices to show
\begin{equation}\label{G.lasttime}
e^{(\ln\lambda-\varepsilon)k} \leq|| A_k||\leq e^ {(\ln\lambda+\varepsilon)k},
\end{equation}
which follows directly   from (\ref{G.Gnew1}),(\ref{G.transfer2}) and (\ref{G24}).
 \end{proof}
  \textbf{Proof of (\ref{G.Asymptoticstransfer})}
\begin{proof}
The arguments are similar to the proof of (\ref{G.Asymptotics}) and
consist of collecting the already proved facts, with the same cases.

 Case i: $\frac{q_n}{2}  \leq q_{n+1}^{\frac{8}{9}}.$
 \par
 For   $\frac{q_n}{2}\leq k \leq 4q_{n+1}^{\frac{8}{9}}$, 
 the result follows from Theorem \ref{Th.new2transfernew}. 
 \par
 For $  4q_{n+1}^{\frac{8}{9}}\leq k\leq \frac{q_{n+1}}{2}$,   (\ref{G.Asymptoticstransfer}) follows  from    Theorem \ref{Th.new3transfer} (notice that  now $k\geq 2q_n$, thus $n_0=1$).

Case ii: $ q_{n+1}^{\frac{8}{9}}\leq \frac{q_n}{2}. $
 \par
 Case ii.1: $\frac{q_n}{2}\leq k\leq \min \{q_n,\frac{q_{n+1}}{2}\}$.

 If $q_n=q_{n-1}+q_{n-2}$, then $q_{n-1}\geq \frac{q_n}{2}$. This is the case 2 of  Theorem \ref{Th.new3transfer}. By (\ref{G.lasttime}), one has
 for any $q_{n-1}\leq k\leq \min \{q_n,\frac{q_{n+1}}{2}\}$
 \begin{equation*}
   ||A_k||\geq  e^{(\ln\lambda-\varepsilon)k}.
 \end{equation*}
 This leads to
 \begin{equation*}
   ||A_k||\geq e^{(\ln\lambda-\varepsilon)k}.
 \end{equation*}
 for $\frac{q_n}{2}\leq k\leq \min \{q_n,\frac{q_{n+1}}{2}\}$.
 This also implies  (\ref{G.Asymptoticstransfer}).  

 If $q_n=jq_{n-1}+q_{n-2}$ with $j\geq 2$, then $ \frac{q_n}{2}\geq q_{n-1}$.  (\ref{G.Asymptoticstransfer}) follows directly from
 Theorem \ref{Th.new3transfer} (notice that  now  $n+1-n_0=n-1$).

 Case ii.2 $q_n\leq k\leq \frac{q_{n+1}}{2}$

 In this case  (\ref{G.Asymptoticstransfer}) follows directly from Theorem \ref{Th.new3transfer} (notice that  now  $n+1-n_0=n$).

 \end{proof}


\section{Proof of the corollaries }\label{cor}

{\bf Proof of Corollary \ref{corlyap}}
\begin{proof}
Due to (\ref{G24}), i) follows from iii). By Theorem \ref{Maintheoremdecay} to prove iii), it is enough
to show that for any $\varepsilon>0$, sufficiently large $n$ and
$\varepsilon q_n <k< q_n$, we have
\begin{equation} \label{density}
e^{(\ln\lambda-C\varepsilon)k} \leq g(k)\leq e^{(\ln\lambda+C\varepsilon)k}.
\end{equation}
Let $m\leq n+1$ be such that $q_m/2\leq k<q_{m+1}/2$.  If $q_m\leq k$, we are
 in Case 1 of the definition of $g$ with $m\leq n-1.$  Notice that $\ell q_{m}\geq \varepsilon q_n\geq \varepsilon q_{m+1}$, which leads to $\frac{\ln \frac{q_{m+1}}\ell}{q_{m}}$ being small. Then (\ref{density})
follows from (\ref{G.decayingratenonresonant2addtransfer}).  
If $q_m/2<k< q_m$  (\ref{density}) is automatic by the definition of
$g.$

It remains to establish ii). 

First by (\ref{G.new15}),
we must have
\begin{equation}\label{G.Transferliminf1}
  \liminf_{k\to \infty}\frac{\ln ||A_k||}{k}\geq \ln\lambda-\beta.
\end{equation}

 Let $j_k=\lfloor {q}_{n_k+1}^{\varepsilon}\rfloor,$ where sequence ${q}_{n_k}$ is given by (\ref{Seq}).
 Then
 \begin{eqnarray*}
   g(j_k{q}_{n_k}) &=&  \frac{{q}_{n_k+1}}{\bar{r}_{j_k}^{n_k}} \\
     &=& e^{(\ln\lambda-\frac{\ln {q}_{n_k+1}}{q_{n_k}}+\frac{\ln {j_k}}{{q}_{n_k}})j_k{q}_{n_k}} {q}_{n_k+1}\\
     &\leq &  e^{(\ln\lambda-\beta+C\varepsilon)j_k{q}_{n_k}} .
 \end{eqnarray*}
 Combining with
 Theorem \ref{Maintheoremdecay},
  we must have
 \begin{equation}\label{G.Transferliminf2}
  \liminf_{k\to \infty}\frac{\ln ||A_k||}{k}\leq \ln\lambda-\beta.
\end{equation}
ii) holds by (\ref{G.Transferliminf1}) and (\ref{G.Transferliminf2}).
\end{proof}
{\bf Proof of Corollary \ref{corleigen}}

\begin{proof}
i) follows from iii) and ii) follows from (\ref{G.new15}) and (\ref{G.new16}). To establish iii),
we only need to show that for any $\varepsilon q_n\leq k\leq
q_n-\frac{\beta}{2\ln \lambda}q_n-\varepsilon q_n$,
\begin{equation*}
  e^{-(\ln\lambda+C\varepsilon)k}  \leq ||U(k)||\leq e^{-(\ln\lambda-C\varepsilon)k}.
\end{equation*}
Let $m\leq n+1$ be such that $q_m/2\leq k<q_{m+1}/2$.

Case 1: $m\leq n-1$.\color{black}


Then  by Theorem \ref{Maintheoremdecay}, the statement is not
immediate only in case 1, but then it   follows from
(\ref{G.decayingratenonresonant2add}) and
(\ref{G.decayingratenonresonant2add0}) since in that case 
$\bar{r}_{\ell}^m\leq e^{-(\ln\lambda-\varepsilon)\ell q_m}$.
  \color{black}

Case 2: $m= n$ or $m= n+1$.

 In this case we have $\frac{q_m}{2}\leq k\leq q_m-\frac{\beta}{2\ln\lambda}q_m-\varepsilon q_m$ and $\frac{q_m}{2}\leq k<\frac{q_{m+1}}{2}$.
 Then the statement holds by
 Theorem \ref{Maintheoremdecay} and, in case 1,  (\ref{G.decayingratenonresonant2add0}).
\color{black}

To prove iv) it suffices to show that for any
${q}_{n_j}-\frac{\beta}{2\ln\lambda}{q}_{n_j}+\varepsilon {q}_{n_j}\leq k\leq {q}_{n_j}-\varepsilon
{q}_{n_j}$,\color{black}  we have
$$ ||U(k)||\geq e^{-(\ln\lambda-c\varepsilon)k}$$ where
${q}_{n_j}$ is a subsequence satisfying (\ref{Seq}).
 Indeed, under this assumption we are in Case 1
of the definition of $f$ and the second addend dominates in
(\ref{G.decayingratenonresonant2add0}) leading to the statement.
\end{proof}
{\bf Remarks}\begin{itemize}\item If we take for ${q}_{n_k}$ a subsequence with {\it any} bounded
  away from zero
  exponential growth, we still get non-Lyapunov behavior on intervals
  of the form $[{q}_{n_k}-c{q}_{n_k}+\varepsilon {q}_{n_k}, {q}_{n_k}-\varepsilon
{q}_{n_k}]$ for some $c<\frac{\beta}{2\ln\lambda}.$
\item In fact, in all the arguments $\beta$ can be replaced with $ \ln
  q_{n+1}/ q_n$.
\end{itemize}

\color{black}

{\bf Proof of Corollary \ref{angle}}
\begin{proof}
First by (\ref{W}), one has
\begin{equation}\label{G.area}
  ||U(k)|| ||\tilde{U} (k)|| \sin \delta_k =\frac{1}{2}.
\end{equation}
Combining with (\ref{G.last}), we have
\begin{equation}\label{G.sinangle}
 \frac{1}{2||U(k)|||| A_k||}\leq  \sin \delta_k\leq \frac{1}{||U(k)|||| A_k||}.
\end{equation}
We first prove (\ref{G.anglelimsup}). Clearly, it suffices to show
\begin{equation}\label{G.anglelimsup1}
   \limsup_{k\to \infty}\frac{\ln \delta_k}{k}\geq0.
 \end{equation}
  Let $k_j=\lfloor \frac{1}{4}{q}_{n_j+1}\rfloor$ where sequence ${q}_{n_j}$ is given by (\ref{Seq}).
  By Theorem \ref{Maintheoremdecay},
  we must have
  \begin{equation*}
    ||U(k_j)||\leq e^{-(\ln\lambda-\varepsilon)k_j},
  \end{equation*}
  and
   \begin{equation*}
    ||A_{k_j}||\leq e^{(\ln\lambda+\varepsilon)k_j}.
  \end{equation*}
 Combining with (\ref{G.sinangle}), we must have
  \begin{equation}\label{Glimsup}
   \delta_{k_n}\geq e^{-\varepsilon k_n}.
  \end{equation}
This implies   (\ref{G.anglelimsup1}) and also implies (\ref{G.anglelimsup}).

Now we verify (\ref{G.angleliminf}).
By the definition of $f(k),g(k)$, for any large $k$, 
we have
\begin{equation}\label{G.angleliminf1}
   f(k)g(k)\leq e^{(\beta+\varepsilon)k}.
\end{equation}
Then by  (\ref{G.sinangle}) and Theorem \ref{Maintheoremdecay} again, one has
\begin{equation}\label{G.angleliminf2}
   \liminf_{k\to \infty}\frac{\ln \delta_k}{k}\geq-\beta.
 \end{equation}

 Let $k_j={q}_{n_j}$.
 One has
 \begin{equation}\label{G.angleliminf3}
   f(k_j)g(k_j)={q}_{n_j+1} .
\end{equation}
 Combining with  (\ref{G.sinangle}) and Theorem \ref{Maintheoremdecay}
 again, we get
 \begin{equation}\label{G.angleliminf4}
   \lim_{j\to \infty}\frac{\ln \delta_{k_j}}{k_j}=-\beta.
 \end{equation}
  (\ref{G.angleliminf}) follows directly from  (\ref{G.angleliminf2}) and  (\ref{G.angleliminf4}).
\end{proof}
{\bf Proof of Corollary \ref{C.anysolution}}
\begin{proof}
This Corollary follows directly from (\ref{G.Asymptoticstransfer}) and (\ref{G.last}).
\end{proof}
\color{black}
\appendix
\section{Gordon arguments for $\lambda\leq e^{\beta}$}
\begin{proposition}\label{decayingrate}
The almost Mathieu operator
 \begin{equation*}
 (H_{\lambda,\alpha,\theta}u)(n)=u({n+1})+u({n-1})+ 2\lambda \cos2\pi(\theta+n\alpha)u(n),
 \end{equation*}
 has no localized eigenfunctions if $|\lambda|\leq
 e^{\beta}$.\footnote{Localized here means exponentially decaying. One
   can exclude any decaying solutions for $|\lambda|<
 e^{\beta}$\cite{AYZ} but not for $|\lambda|=
 e^{\beta}$\cite{ajz}.}

\end{proposition}
\begin{proof}
Otherwise, there exists a solution   $\{u(n)\}_{n\in\mathbb{Z}}$ of $H_{\lambda,\alpha,\theta}u=Eu$ such that
\begin{equation}\label{Loc}
   | u(n)|\leq C e^{-5c|n|},
\end{equation}
 where $c>0$.
 Without loss of generality, assume the vector  $ \left(\begin{array}{c}
                                                                                            u(0 )\\
                                                                                            u({-1})                                                                                        \end{array}\right)
$
is unit.
\par
Let $\varphi(n)=\left(\begin{array}{c}
        u(n)\\
       u({n-1})
     \end{array}\right)
 $.
For simplicity, denote  $\varphi=\varphi(0)$. By the definition of $\beta(\alpha)$, there exists a subsequence $\tilde{q}_k$
of $q_n$ such that
\begin{equation}\label{Seq1}
  ||\tilde{q}_k\alpha||_{\mathbb{R}/\mathbb{Z}}  \leq e^{-(\beta-\frac{c}{4})\tilde{q}_k}.
\end{equation}
\par
Denote  $B=A_{\tilde{q}_k}(\theta)$. Then we have
\begin{equation}\label{Tr}
    B^2+(\text{Tr}  B)  B +I=0.
\end{equation}
\textbf{Case 1:} if $\text{Tr}  B\leq e^{ 2c\tilde{q}_k}$, one has
either $||B^2\varphi||\geq  \frac{1}{2}$ or $||B\varphi||\geq  \frac{1}{2}e^{-2c\tilde{q}_k}$.
By (\ref{Loc}), we must have $||B^2\varphi||\geq  \frac{1}{2}$. This is impossible. Indeed,  from the following estimate
\begin{eqnarray*}
  ||\varphi(2\tilde{q}_k)-B^2\varphi|| &=& ||A_{\tilde{q}_k}(\theta+\tilde{q}_k\alpha)-A_{\tilde{q}_k}(\theta)||\;\;||\varphi(\tilde{q}_k)||\\
   &\leq & Ce^{ \frac{3}{2}c \tilde{q}_k} e^{-5c\tilde{q}_k}\\
 &\leq&  e^{-3c\tilde{q}_k},
\end{eqnarray*}
where the first inequality holds by (\ref{G.new20}).
Then
\begin{equation*}
   || \varphi(2\tilde{q}_k)||\geq \frac{1}{4} ,
\end{equation*}
contradicting $ ||\varphi(2\tilde{q}_k)||\leq  Ce^{-10c\tilde{q}_k}$.
\par
\textbf{Case 2:} if  $\text{Tr}  B\geq e^{2c\tilde{q}_k}$, from
(\ref{Tr}), it is easy to see that
either $||B\varphi||\geq  \frac{1}{2}e^{2c\tilde{q}_k}$ or $||B^{-1}\varphi||\geq  \frac{1}{2}e^{2c\tilde{q}_k}$ holds.  By (\ref{Loc}) again,
we must have $||B^{-1}\varphi||\geq  \frac{1}{2}e^{2c\tilde{q}_k}$. By (\ref{G.newnew20}), the following holds
\begin{eqnarray*}
  ||\varphi(-\tilde{q}_k)-B^{-1}\varphi|| &=& ||A_{\tilde{q}_k}^{-1}(\theta-\tilde{q}_k\alpha)-A^{-1}_{\tilde{q}_k}(\theta)||\;\;||\varphi||\\
   &\leq & e^{ \frac{3}{2}c  \tilde{q}_k}.
\end{eqnarray*}
Thus
\begin{equation*}
    ||\varphi(-\tilde{q}_k)||\geq \frac{1}{4}e^{2c\tilde{q}_k}.
\end{equation*}
This is also impossible.

\end{proof}

\section{Uniformity}

We start with some basic facts.
\par
 Let $ \frac{p_n}{q_n}$ be the continued fraction approximants to $\alpha$. Then
\begin{equation}\label{GDC1}
\forall 1\leq k <q_{n+1},  \text{dist}( k\alpha,\mathbb{Z})\geq  |q_n\alpha -p_n|,
\end{equation}
and
\begin{equation}\label{GDC2}
      \frac{1}{2q_{n+1}}\leq\Delta_n:=|q_n\alpha-p_n| \leq\frac{1}{q_{n+1}}.
\end{equation}
\begin{lemma} (\text{Lemma } 9.7, \cite{avila2009ten})
Let $\alpha\in \mathbb{R}\backslash \mathbb{Q}$, $x\in\mathbb{R}$ and $0\leq \ell_0 \leq q_n-1$ be such that
$ |\sin\pi(x+\ell_0\alpha)|=\inf_{0\leq\ell\leq q_n-1}    |\sin\pi(x+\ell \alpha)|$, then for some absolute constant $C > 0$,
\begin{equation}\label{G927}
    -C\ln q_n\leq \sum _{\ell=0,\ell\neq \ell_0}^{q_n-1} \ln|\sin\pi(x+\ell\alpha )|+(q_n-1)\ln2\leq  C\ln q_n.
\end{equation}
  \end{lemma}

We now prove
 \begin{lemma}\label{Lanaaddsmallest}
  For any $|i|,|j|\leq 50 C_{\ast}b_{n+1}$, if   $\theta $  is
  $n$-Diophantine with respect to $\alpha$, then the following estimate holds,
  \begin{equation}\label{G.lanaadd1}
       \ln |\sin
 \pi(2\theta+ (j+i )\alpha) | \geq -C\ln q_n.
\end{equation}
  \end{lemma}
  \begin{proof}

    By  the Diophatine condition on $\theta,$ (\ref{DCtheta}),
 one has that there exist $\kappa>0$ and $\nu>0$ such that
 \begin{equation}\label{Gpositivesmall}
  \min_{j,i\in [-q_n,q_n]} |\sin
 \pi(2\theta+ (j+i )\alpha) |\geq  \frac{\kappa}{ q_n^{\nu}} .
 \end{equation}
 Let $\ell_i,\ell_j\in\mathbb{Z}$ be such that $dist(i,  q_n \mathbb{Z})=|i-\ell_i q_n|$ and $dist(j,q_n \mathbb{Z})=|j-\ell_j q_n|$. Then $ |\ell_i|,|\ell_j | \leq 50 C_{\ast}\frac{b_{n+1}}{q_n}+1$.
Let
 $ i'=i- \ell_iq_n$ and $j'=j-\ell_jq_n$, then $i',j'\in[-q_n,q_n]$.
 \par
If $q_{n+1}^{1-t}> \frac{100C_{\ast}}{\kappa}q_n^{\nu+2}$, it is easy
to verify that  $|\ell_k \Delta_n|<\frac{\kappa}{ q_n^{\nu+1}}$.
Combining with   (\ref{Gpositivesmall}),
  we have for any $|i|,|j|\leq 50 C_{\ast}b_{n+1}$
 \begin{equation*}
    |\sin  \pi(2\theta+ (j+i )\alpha) |\;\;\;\;\;\;\;\;\;\;\;\;\;\;\;\;\;\;\;\;\;\;\;\;\;\;\;\;\;\;\;\;\;\;\;\;\;\;\;\;\;\;\;\;\;\;\;\;\;
    \;\;\;\;\;\;\;\;\;\;\;\;\;\;\;\;\;\;\;\;\;\;\;\;\;\;\;\;\;\;\;\;\;\;\;\;\;\;\;\;\;\;\;\;\;\;\;\;
 \end{equation*}
 \begin{equation*}
 \;\;\;\;\;\;\;\;\;\;\;\;\;\;\;\;\;\;
    =|\sin\pi(2\theta+ (j'+i')\alpha )\cos \pi (\ell_i+\ell_j)\Delta_n
\pm \cos  \pi(2\theta+ (j'+i')\alpha ) \sin\pi (\ell_i+\ell_j)\Delta_n|
 \end{equation*}
\begin{equation*}
 \geq\frac{\kappa}{100q^{\nu}_n} \;\;\;\;\;\;\;\;\;\;\;\;\;\;\;\;\;\;\;\;\;\;\;\;\;\;\;\;\;\;\;\;\;\;\;\;\;\;\;\;\;\;\;\;\;\;\;\;\;
    \;\;\;\;\;\;\;\;\;\;\;\;\;\;\;\;\;\;\;\;\;\;\;\;\;\;\;\;\;\;\;\;\;\;\;\;\;\;\;\;\;\;\;\;\;\;\;\;
\end{equation*}
(the choice of $\pm$ depends on   the sign of $q_n \alpha-p_n$).
\par
If $q_{n+1}^{1-t}\leq  \frac{100C_{\ast}}{\kappa}q_n^{\nu+2}$,  
we also have
 for any $|i|,|j|\leq 50 C_{\ast}b_{n+1}$
 \begin{equation*}
 |\sin  \pi(2\theta+ (j+i )\alpha) | \geq\frac{\kappa^{1+\frac{t\nu}{1-t}}}{(100C_{\ast})^{\frac{\nu}{1-t}}q^{\frac{\nu t(\nu+2)}{1-t}}_n}.
\end{equation*}
 Thus in both cases,
 we have
\begin{equation}\label{1Gsum+}
     \min_{|i|,|j|\leq 50 C_{\ast}b_{n+1}} \ln |\sin
 \pi(2\theta+ (j+i )\alpha) | \geq -C\ln q_n.
\end{equation}
\end{proof}

  \begin{lemma}\label{Lanaaddsmallest2}
  Assume $|i|,|j|\leq 50 C_{\ast}b_{n+1}$, and $i-j\neq q_n\mathbb{Z}$.  Then
  \begin{equation}\label{G.lanaadd2}
       \ln |\sin
 \pi(j-i )\alpha | \geq -C\ln q_n.
\end{equation}
  \end{lemma}
  \begin{proof}

  By assumption, $|j-i|=\ell q_n+r$ with $0 \leq \ell\leq  100 C_{\ast}\frac{b_{n+1}}{q_n}$ and $0<r<q_n.$
  Then by (\ref{GDC1}) and  (\ref{GDC2}) again,  we also have
  \begin{eqnarray*}
    ||(j-i )\alpha||_{\mathbb{R}/\mathbb{Z}} &\geq& ||r\alpha||_{\mathbb{R}/\mathbb{Z}}-|\ell| ||q_n\alpha||_{\mathbb{R}/\mathbb{Z}}\\
 &\geq &  \frac{1}{2q_n} -\frac{|\ell|}{q_{n+1}}\\
 &\geq &  \frac{1}{2q_n} -\frac{100 C_{\ast}}{q_{n+1}^{1-t}}\frac{1}{q_n}\\
     &\geq & \frac{1}{4q_n} .
  \end{eqnarray*}
  This implies (\ref{G.lanaadd2}).
  \end{proof}
We are now ready to study the behavior at non-resonant points.
For an  $n$-nonresonant $y$,  let, as before, $n_0$ be the least positive integer such that $4q_{n-n_0}\leq dist(y,  q_n\mathbb{Z})$.
 Let $s$ be the
largest positive integer such that $4sq_{n-n_0}\leq dist(y,q_n\mathbb{Z}) $.
Recall that, automatically, $n_0\leq C(\alpha)$. Set $I_1, I_2\subset \mathbb{Z}$ as follows
\begin{eqnarray*}
  I_1 &=& [-sq_{n-n_0},sq_{n-n_0}-1], \\
   I_2 &=& [ y-sq_{n-n_0},y+sq_{n-n_0}-1 ],
\end{eqnarray*}
We have
\begin{theorem}\label{new}
For an $n-$nonresonant $y$,
assume that
\begin{equation}\label{Gsum+}
     \min_{j,i\in I_1\cup I_2} \ln |\sin
 \pi(2\theta+ (j+i )\alpha) | \geq -C\ln q_n.
\end{equation} and \begin{equation}\label{Gsum-}
     \min_{i\neq j;i,j\in I_1\cup I_2} \ln |\sin
 \pi (j-i )\alpha) | \geq -C\ln q_n.
\end{equation}
Then for  any $ \varepsilon>0$ and   $ n$  large enough,  we have $ y$ is $  (\ln\lambda+8\ln (s q_{n-n_0}/q_{n-n_0+1})/q_{n-n_0}-\varepsilon,4sq_{n-n_0}-1)$ regular with $\delta=\frac{1}{4}$.
\end{theorem}
\begin{proof}
Without loss of generality assume $y>0$.
By the definition of $s$ and $n_0$, we have
$4sq_{n-n_0}\leq dist(y,q_n \mathbb{Z})$ and  $4q_{n-n_0+1}>
dist(y,q_n \mathbb{Z})$. This leads to $sq_{n-n_0}\leq q_{n-n_0+1}$.  Let $\theta_j=\theta+j\alpha$ for $j\in I_1\cup I_2$. The set $\{\theta_j\}_{j\in I_1\cup I_2}$
consists of $4sq_{n-n_0}$ elements.

In (\ref{Def.Uniform}), let $x=\cos2\pi a$, $k=4sq_{n-n_0}-1$  and
take    the logarithm, then
  $$ \ln \prod_{ j\in I_1\cup  I_2  , j\neq i } \frac{|\cos2\pi a-\cos2\pi\theta_j|}
        {|\cos2\pi\theta_i-\cos2\pi\theta_j|}\;\;\;\;\;\;\;\;\;\;\;\;\;\;\;\;\;\;\;\;\;\;\;\;\;\;\;\;\;\;\;\;\;\;\;\;\;\;\;\;\;\;\;\;\;\;\;\;\;\;\;\;\;
        \;\;\;\;\;\;\;\;\;\;\;\;\;\;\;$$
          $$=   \sum _{ j\in I_1\cup  I_2  , j\neq i }\ln|\cos2\pi a-\cos2\pi \theta_j|- \sum _{ j\in I_1\cup  I_2  , j\neq i }\ln|\cos2\pi\theta_i  -\cos2\pi \theta_j| .$$

First,  we  estimate $ \sum _{ j\in I_1\cup  I_2  , j\neq i }\ln|\cos2\pi a-\cos2\pi \theta_j| $.
Obviously,
   $$ \sum _{ j\in I_1\cup  I_2  , j\neq i }\ln|\cos2\pi a-\cos2\pi \theta_j| \;\;\;\;\;\;\;\;\;\;\;\;\;\;\;\;\;\;\;\;\;\;\;\;\;\;\;\;\;\;\;\;\;\;\;\;\;\;\;\;\;\;\;\;\;\;\;\;\;\;\;\;\;\;\;\;$$
$$\;\;\;\;\;\;\;\;\;\;\;\;\;\;\;\;\;\;\;\;\;=\sum_{ j\in I_1\cup  I_2  , j\neq i }\ln|\sin\pi(a+\theta_j)|+\sum_{ j\in I_1\cup  I_2  , j\neq i }\ln |\sin\pi(a-\theta_j)|
+(4sq_{n-n_0}-1)\ln2  $$
\begin{equation*}
    =\Sigma_{+}+\Sigma_-+(4sq_{n-n_0}-1)\ln2.  \;\;\;\;\;\;\;\;\;\;\;\;\;\;\;\;\;\;\;\;\; \;\;\;\;\;\;\;\;\;\;\;\;\;\;\;\;\;\;\;\;\;\;\;\;\;\;\;\;\;\;\;\;
\end{equation*}
Both $\Sigma_+$ and $\Sigma_-$ consist of $4s$ terms of the form of  (\ref{G927}), plus 4s terms of the form
\begin{equation*}
    \ln\min_{j=0,1,\cdots,q_{n-n_0}}|\sin\pi(x+j\alpha)|,
\end{equation*}
minus $\ln|\sin\pi(a\pm\theta_i)|$.
       Thus,  using   (\ref{G927})  4s times for $\Sigma_{+}$ and $\Sigma_{-}$ respectively,  one has
 \begin{equation}\label{G.appnumerator}
   \sum_{j \in I_1  \cup I_2,j\neq i}\ln|\cos2\pi a-\cos2\pi \theta_{j}|\leq-4sq_{n-n_0}\ln2+Cs\ln q_{n-n_0}.
\end{equation}
If $a=\theta_i$,
we obtain
   $$ \sum _{j \in I_1  \cup I_2,j\neq i}\ln|\cos2\pi \theta_i-\cos2\pi \theta_j| \;\;\;\;\;\;\;\;\;\;\;\;\;\;\;\;\;\;\;\;\;\;\;\;\;\;\;\;\;\;\;\;\;\;\;\;\;\;\;\;\;\;\;\;\;\;\;\;\;\;\;\;\;\;\;\;$$
$$\;\;\;\;\;\;\;\;\;\;\;\;\;\;\;\;\;\;\;\;\;=\sum_{j \in I_1  \cup I_2,j\neq i}\ln|\sin\pi(\theta_i+\theta_j)|+\sum_{j \in I_1  \cup I_2,j\neq i}\ln |\sin\pi(\theta_i-\theta_j)|
+(4sq_{n-n_0}-1)\ln2  $$
\begin{equation}\label{G.appsumdenumerate}
    =\Sigma_{+}+\Sigma_-+(4sq_{n-n_0}-1)\ln2,  \;\;\;\;\;\;\;\;\;\;\;\;\;\;\;\;\;\;\;\;\;\;\;\; \;\;\;\;\;\;\;\;\;\;\;\;\;\;\;\;\;\;\;\;\;\;\;\;\;\;\;\;\;\;\;
\end{equation}
 where
  \begin{equation*}
   \Sigma_{+}=\sum_{j \in I_1  \cup I_2,j\neq i}\ln |\sin\pi(2\theta+ (i+j) \alpha)|,
 \end{equation*}
 and
\begin{equation*}
     \Sigma_-=\sum_{j \in I_1  \cup I_2,j\neq i}\ln |\sin\pi( i-j)\alpha|.
 \end{equation*}
 We will estimate $\Sigma_+ $.
 Set
$J_1=[- s,s-1]$ and
$J_2=[s ,3s-1 ]$, which are two adjacent disjoint intervals of length
$2s$.
 Then $I_1\cup
I_2$ can be represented as a disjoint union of segments $B_j,\;j\in
J_1\cup J_2,$ each of length $q_{n-n_0}$.
Applying (\ref{G927}) to  each  $B_j$, we
obtain
\begin{equation}\label{G312}
\Sigma_+ \geq -4sq_{n-n_0}\ln 2+
\sum_{j\in J_1\cup J_2 }\ln
|\sin  \pi\hat \theta_j|-Cs\ln q_{n-n_0}-\ln|\sin2\pi (\theta+i\alpha)|,
\end{equation}
where
\begin{equation}\label{G313}
|\sin  \pi\hat \theta_j|=\min_{\ell \in B_j}|\sin  \pi
(2\theta +(\ell +  i)\alpha )|.
\end{equation}

Next we  estimate $\sum_{j\in J_1  }\ln
|\sin  \pi\hat \theta_j|$.
Assume that $\hat \theta_{j+1}=\hat \theta_j+ q_{n-n_0}  \alpha$ for
every $j,j+1 \in J_1$.
In this case, for any $i,j\in J_1$ and $i\neq j$, we have
\begin{equation}\label{APPaddnew1}
    ||\hat\theta_i-\hat\theta_j||_{\mathbb{R}/\mathbb{Z}}\geq  ||q_{n-n_0}\alpha||_{\mathbb{R}/\mathbb{Z}}.
\end{equation}

Applying   the Stirling formula,  (\ref{Gsum+}) and (\ref{APPaddnew1}),  one has
 \begin{eqnarray}
  \nonumber
   \sum_{j\in J_1}\ln
|\sin 2\pi\hat\theta_j|&>&   2\sum_{j=1}^{s}\ln
( j\Delta_{n-n_0}) -C\ln q_n\\
    &>&  2s\ln\frac s{q_{n-n_0+1}}-C\ln q_n-Cs.\label{G317}
 \end{eqnarray}
 \par
In the other cases, decompose $J_1$ in maximal intervals $T_\kappa$
such that for $j,j+1 \in
T_\kappa$ we have $\hat \theta_{j+1}=\hat \theta_j+ q_{n-n_0}
\alpha$.  Notice that the boundary points of an interval $T_\kappa$
are either boundary points of $J_1$ or satisfy
$\| \hat \theta_j\|_{\mathbb{R}/\mathbb{Z}}+\Delta_{n-n_0} \geq
\frac {\Delta_{n-n_0-1}} {2}$.
This follows from the fact that if
$0<|z|<q_{n-n_0}$,  then $\|  \hat \theta_j+q_{n-n_0}\alpha\|_{\mathbb{R}/\mathbb{Z}}\leq \|  \hat \theta_j\|_{\mathbb{R}/\mathbb{Z}}+\Delta_{n-n_0}$,  and  $\|  \hat \theta_j+(z+q_{n-n_0})\alpha \|_{\mathbb{R}/\mathbb{Z}}\geq   \|z \alpha\|_{\mathbb{R}/ \mathbb{Z}}-\|  \hat \theta_j+q_{n-n_0}\alpha\|_{\mathbb{R}/\mathbb{Z}}
\geq \Delta_{n-n_0-1}-\|  \hat \theta_j\|_{\mathbb{R}/\mathbb{Z}}-\Delta_{n-n_0}$.
Assuming $T_\kappa \neq J_1$, then there exists $j \in
T_\kappa$ such that $\|  \hat \theta_j\|_{\mathbb{R}/\mathbb{Z}}\geq
\frac {\Delta_{n-n_0-1}} {2}- \Delta_{n-n_0} $.
\par
If $T_\kappa$  contains  some $j$ with $ \|  \hat
\theta_j\|_{\mathbb{R}/\mathbb{Z}}<\frac {\Delta_{n-n_0-1}} {10} $, then
 \begin{eqnarray}
 \nonumber
   |T_\kappa| &\geq & \frac{\frac {\Delta_{n-n_0-1}} {2}-\Delta_{n-n_0} -\frac {\Delta_{n-n_0-1}} {10}}{\Delta_{n-n_0}}  \\
    &\geq&\frac{1}{4} \frac{\Delta_{n-n_0-1}}{ \Delta_{n-n_0}}-1\geq \frac{s}{8}-1,\label{G318}
 \end{eqnarray}
 since $sq_{n-n_0}\leq q_{n-n_0+1}$,  where $|T_\kappa|=b-a+1$ if $T_\kappa=[a,b]$.
For such $ T_\kappa$,
 a similar estimate to (\ref{G317}) gives
\begin{eqnarray}
\nonumber
  \sum_{j\in T_\kappa}\ln
|\sin  \pi\hat\theta_j|  &\geq&   |T_\kappa|\ln \frac {|T_\kappa|}{q_{n-n_0+1}} -C s-C\ln q_n  \\
 &\geq&    |T_\kappa|\ln \frac{s}{q_{n-n_0+1}}-Cs-C\ln q_n. \label{G319}
\end{eqnarray}
 If $T_\kappa$ does not   contain any $j$ with $ \|  \hat
\theta_j\|_{\mathbb{R}/\mathbb{Z}}<\frac {\Delta_{n-n_0-1}} {10} $, then by (\ref{GDC2})
\begin{eqnarray}
\nonumber
  \sum_{j\in T_\kappa}\ln
|\sin  \pi\hat\theta_j|  &\geq&  -|T_\kappa|\ln q_{n-n_0}-C|T_\kappa| \\
    &\geq&    |T_\kappa|\ln \frac{s}{q_{n-n_0+1}}-C|T_\kappa|.\label{G320}
\end{eqnarray}
 \par
By (\ref{G319}) and (\ref{G320}), one has
\begin{equation}\label{G321}
  \sum_{j\in J_1}\ln
|\sin  \pi\hat\theta_j|\geq  2s\ln \frac{s}{q_{n-n_0+1}} -C  s -C\ln q_n.
\end{equation}
Similarly,
\begin{equation}\label{G322}
  \sum_{j\in J_2}\ln
|\sin  \pi\hat\theta_j|\geq  2s\ln \frac{s}{q_{n-n_0+1}} -C  s-C\ln q_n .
\end{equation}
Putting  (\ref{G312}), (\ref{G321}) and (\ref{G322}) together, we have
\begin{equation}\label{G323}
\Sigma_+ > -4sq_{n-n_0}\ln 2+ 4s\ln \frac{s}{q_{n-n_0+1}} -Cs\ln q_{n-n_0}-C\ln q_n.
\end{equation}
Now we start to estimate $ \Sigma_-$.
Replacing (\ref{Gsum+}) with (\ref{Gsum-}), and following the proof of (\ref{G323}), we obtain,
\begin{equation}\label{newG323}
\Sigma_- > -4sq_{n-n_0}\ln 2+ 4s\ln \frac{s}{q_{n-n_0+1}} -Cs\ln q_{n-n_0}-C\ln q_n.
\end{equation}

From  (\ref{G.appsumdenumerate}), (\ref{G323}) and (\ref{newG323}), one has
$$ \sum _{j \in I_1  \cup I_2,j\neq i}\ln|\cos2\pi \theta_i-\cos2\pi \theta_j| \;\;\;\;\;\;\;\;\;\;\;\;\;\;\;\;\;\;\;\;\;\;\;\;\;\;\;\;\;\;\;\;\;\;\;\;\;\;\;\;\;\;\;\;\;\;\;\;\;\;\;\;\;\;\;\;$$
\begin{equation}\label{G.app325}
  \geq -4sq_{n-n_0}\ln 2+ 8s\ln \frac{s}{q_{n-n_0+1}} -Cs\ln q_{n-n_0}-C\ln q_n.
\end{equation}
By (\ref{G.appnumerator})  and (\ref{G.app325}),
we have
\begin{equation*}
        \max_{ i\in I_1\cup I_2} \prod_{j \in I_1\cup I_2,j \neq i } \frac{|x-\cos2\pi\theta_{j }|}
        {|\cos2\pi\theta_i-\cos2\pi\theta_{j }|}<e^{ 4sq_{n-n_0}(-2\ln (s /q_{n-n_0+1})/q_{n-n_0}+  \varepsilon  )  }.
      \end{equation*}
    Combining with  Lemma  \ref{Le.Uniform}, there exists some $j_0$ with  $j_0\in I_1\cup I_2$
    such that
      $$ \theta_{j _0}\notin  A_{4sq_{n-n_0}-1,\ln\lambda+2\ln (s /q_{n-n_0+1})/q_{n-n_0}-  \varepsilon }.$$

      First, we assume $j_0\in I_2$.
      \par
      Set $I=[j_0-2sq_{n-n_0}+1,j_0+2sq_{n-n_0}-1]=[x_1,x_2]$.   By  (\ref{Cramer1}), (\ref{Cramer2}) and (\ref{Numerator}),
 it is easy to verify
\begin{equation*}
|G_I(y,x_i)|\leq\exp\{(\ln\lambda+\varepsilon  )(4sq_{n-n_0}-1-|y-x_i|)-4sq_{n-n_0}(\ln\lambda+2\ln (s /q_{n-n_0+1})/q_{n-n_0}-  \varepsilon) \}.
\end{equation*}
Notice that $|y-x_i|\geq sq_{n-n_0}$, so
we obtain
\begin{equation}\label{2AppG.Green}
|G_I(y,x_i)|\leq\exp\{ -(\ln \lambda+ 8\ln (s /q_{n-n_0+1})/q_{n-n_0}- 2 \varepsilon)|y-x_i| \}.
\end{equation}
\par
If $j_0\in I_1$,  we may let $y=0$ or  $y=1 $ in (\ref{2AppG.Green}).
Combining with (\ref{Block}), we get
\begin{equation*}
    |\phi(0)|,|\phi(-1)|\leq 6sq_{n-n_0}\exp\{ -(\ln \lambda+ 8\ln (s /q_{n-n_0+1})/q_{n-n_0}- 2 \varepsilon)sq_{n-n_0}\}.
\end{equation*}
This is in contradiction with $|\phi(0)|^2+|\phi(-1)|^2=1$. Thus
$j_0\in I_2$, and
the theorem follows from (\ref{2AppG.Green}).\end{proof}

\textbf{Proof of Theorem \ref{Th.Nonresonant}.}

In case i),  (\ref{Gsum+}) and  (\ref{Gsum-}) are obtained correspondingly
from Lemmas \ref{Lanaaddsmallest} and \ref{Lanaaddsmallest2}, thus
Theorem \ref{Th.Nonresonant} follows from Theorem \ref{new}.
In case ii)  it is easy to see that (\ref{Gsum+}) and  (\ref{Gsum-})
also hold, so Theorem \ref{new} applies as well.
\qed

Assume  $b_{n+1}\geq \frac{q_n}{2}$. For any  $1\leq j\leq 48C_{\ast} \frac{b_{n+1}}{q_n}$, we  construct $I_1, I_2\subset \mathbb{Z}$  as follows
\begin{eqnarray*}
  I_1 &=& [-\lfloor\frac{1}{2}q_n\rfloor, q_n-\lfloor\frac{1}{2}q_n\rfloor-1], \\
   I_2 &=& [ j  q_n-\lfloor\frac{1}{2}q_n\rfloor, (j +1)q_n-\lfloor\frac{1}{2}q_n\rfloor-1 ],
\end{eqnarray*}
Let $\theta_m=\theta+m\alpha$ for $m\in I_1\cup I_2$.  Then
\begin{theorem}   \label{Thapp}
Suppose $\theta$ is $n$-Diophantine with respect to $\alpha$. Then
 for any $ \varepsilon >0$,    the set $\{\theta_m\}_{m\in I_1\cup I_2}$
is $\frac{\ln q_{n+1}-\ln j}{2q_n}+\varepsilon$-uniform for sufficiently large $n$.
\end{theorem}
\begin{proof}
In (\ref{Def.Uniform}), let $x=\cos2\pi a$, $k=2q_n-1$  and take    the logarithm.  Thus in order to prove the theorem,
 it suffices to show that
  $$ \ln \prod_{ m\in I_1\cup  I_2  , m\neq i } \frac{|\cos2\pi a-\cos2\pi\theta_m|}
        {|\cos2\pi\theta_i-\cos2\pi\theta_m|}\;\;\;\;\;\;\;\;\;\;\;\;\;\;\;\;\;\;\;\;\;\;\;\;\;\;\;\;\;\;\;\;\;\;\;\;\;\;\;\;\;\;\;\;\;\;\;\;\;\;\;\;\;
        \;\;\;\;\;\;\;\;\;\;\;\;\;\;\;$$
          $$\;\;\;\;\;\;=   \sum _{ m\in I_1\cup  I_2  , m\neq i }\ln|\cos2\pi a-\cos2\pi \theta_m|- \sum _{ m\in I_1\cup  I_2  , m\neq i }\ln|\cos2\pi\theta_i  -\cos2\pi \theta_m| $$
       \begin{equation*}
       \leq   (2q_n- 1)(\frac{\ln q_{n+1}-\ln j}{2q_n}+\varepsilon).\;\;\;\;\;\;\;\;\;\;\;\;\;\;\;\;\;\;\;\;\;\;\;\;\;\;\;\;\;\;\;\;\;\;\;\;\;\;\;\;\;\;\;\;\;\;\;\;
     \;\;\;\;\;\;\;\;\;\;\;\;\;\;\;\;
       \end{equation*}

First,  we  estimate $ \sum _{ m\in I_1\cup  I_2  , m\neq i }\ln|\cos2\pi a-\cos2\pi \theta_m| $.
Obviously,
   $$ \sum _{ m\in I_1\cup  I_2  , m\neq i }\ln|\cos2\pi a-\cos2\pi \theta_m| \;\;\;\;\;\;\;\;\;\;\;\;\;\;\;\;\;\;\;\;\;\;\;\;\;\;\;\;\;\;\;\;\;\;\;\;\;\;\;\;\;\;\;\;\;\;\;\;\;\;\;\;\;\;\;\;$$
$$\;\;\;\;\;\;\;\;\;\;\;\;\;\;\;\;\;\;\;\;\;=\sum_{ m\in I_1\cup  I_2  , m\neq i }\ln|\sin\pi(a+\theta_m)|+\sum_{ m\in I_1\cup  I_2  , m\neq i }\ln |\sin\pi(a-\theta_m)|
+(2q_n-1)\ln2  $$
\begin{equation*}
    =\Sigma_{+}+\Sigma_-+(2q_n-1)\ln2.\;\;\;\;\;\;\;\;\;\;\;\;\;\;\;\;\;\;\;\;\; \;\;\;\;\;\;\;\;\;\;\;\;\;\;\;\;\;\;\;\;\;\;\;\;\;\;\;\;\;\;\;\;\;\;\;\;\;\;\;
\end{equation*}
Both  $\Sigma_{+}$ and  $\Sigma_{-}$ consist of 2 terms of the form of (\ref{G927}), plus
two terms of the form
\begin{equation*}
    \min_{k=1,\cdots,q_n}\ln|\sin \pi(x+k\alpha)|,
\end{equation*}
 minus $\ln|\sin \pi(a\pm \theta_i) $.
 Thus one has
\begin{equation*}
\sum_{{m\in I_1\cup I_2} ,{m\not=i}}\ln |\cos 2\pi a-\cos
2\pi\theta_m|\le  -2q_{n} \ln 2 +C\ln q_n.
\end{equation*}

Setting $a=\theta_i$ and
using  the first inequality  of  (\ref{G927}) two times, we obtain
   \begin{eqnarray}
     \nonumber \sum _{m\in I_1\cup  I_2  , m\neq i }\ln|\cos2\pi\theta_i  -\cos2\pi \theta_m| & \geq &  -2q_n\ln 2-C\ln q_n+ 2\min_{m,i\in I_1\cup I_2} \ln |\sin
 \pi(2\theta+ (m+i )\alpha) | \\
       && +\min_{ m\in I_1\cup I_2, m\neq i} \ln |\sin
 \pi (m-i )\alpha |. \label{Gsmall}
   \end{eqnarray}

By Lemma \ref{Lanaaddsmallest}, we also have
\begin{equation}\label{1Gsum+}
     \min_{m,i\in I_1\cup I_2} \ln |\sin
 \pi(2\theta+ (m+i )\alpha) | \geq -C\ln q_n.
\end{equation}

By (\ref{GDC1}) and (\ref{GDC2}), the corresponding minimum term of $\min_{ m\in I_1\cup I_2, m\neq i} \ln |\sin
 \pi( (m-i )\alpha) |$ is achieved at  $j q_n$.
It is easy to check that
\begin{equation}\label{1Gsum-}
  \min\{  \ln|\sin \pi j q_n\alpha|\}>-\ln \frac{q_{n+1}}{j}-C,
\end{equation}
since    $\Delta_n\geq \frac{1}{2q_{n+1}} $.

Putting (\ref{Gsmall}), (\ref{1Gsum+}) and (\ref{1Gsum-}) together,
we obtain
$$\max_{ x\in[-1,1]}\max_{i=1,\cdots,k+1}\prod_{ m=1 , m\neq i }^{k+1}\frac{|x-\cos2\pi\theta_m|}
        {|\cos2\pi\theta_i-\cos2\pi\theta_m|}\leq e^{(2q_n-1)(\frac{\ln q_{n+1}-\ln j}{2q_n}+\varepsilon) }.$$

\end{proof}

 \section*{Acknowledgments}

We would like to thank Ya. Pesin and Q. Zhou for  comments on  earlier versions of
the manuscript and Ya. Pesin also for the related discussions of non-regular dynamics.  S.J. was supported by the Simons Foundation and W.L. was supported by the AMS-Simons Travel
Grant 2016-2018.  This research was
partially
 supported by NSF DMS-1401204 and NSF DMS-1700314. We are grateful to the Isaac
    Newton Institute for Mathematical Sciences, Cambridge, for its
    hospitality, supported by EPSRC Grant Number EP/K032208/1, during the programme Periodic and Ergodic Spectral
    Problems where this work was started.

\end{document}